\documentclass{lmcs} 
\pdfoutput=1
\usepackage[utf8]{inputenc}

\usepackage{lastpage}
\lmcsdoi{19}{3}{11}
\lmcsheading{}{\pageref{LastPage}}{}{}%
{Sep.~13,~2022}{Aug.~09,~2023}{}


\usepackage{amsmath,amssymb}
\usepackage{hyperref}
\usepackage{pgfplots}
\pgfplotsset{compat=1.16}
\usepackage{tikz-3dplot}
\usepackage{xspace}
\usepackage{tikz}
\usepackage{float}
\usepackage{multirow}
\usepackage{marvosym}
\usepackage{enumitem}
\usepackage{booktabs}
\usepackage{subcaption}
\usepackage{tcolorbox}
\usepackage{mdframed}
\usepackage{minibox}
\usepackage{xcolor,colortbl}
\usepackage{diagbox}
\usepackage{wrapfig}

\newtheoremstyle{cited}%
{.5\baselineskip\@plus.2\baselineskip
    \@minus.2\baselineskip}
{.5\baselineskip\@plus.2\baselineskip
    \@minus.2\baselineskip}
{\itshape}
{\parindent}
{}
{.}
{.5em}
{\textsc{\thmname{#1}} \thmnote{\normalfont#3}}

\theoremstyle{cited}
\newtheorem{citedprop}{Proposition}

\def\punto{$\hspace*{\fill}\Box$}
\definecolor{light-gray}{gray}{0.7.2}
\definecolor{goodgreen}{rgb}{0.1, 0.5, 0.1}
\definecolor{burntorange}{rgb}{0.8, 0.33, 0.0}

\newcommand{\JOIN}{\text{\larger[1]$\Join$}}

\newcommand{\vars}{\mathit{vars}}
\newcommand{\free}{\mathit{free}}

\newcommand{\atoms}{\mathit{atoms}}
\newcommand{\dep}{\textit{dep}}
\newcommand{\sibling}{\textsf{has\_sibling}}

\newcommand{\bigO}[1]{\mathcal{O}(#1)}

\newcommand{\VO}{\mathsf{VO}}
\newcommand{\canonicalVO}{\mathsf{canonVO}}
\newcommand{\freeTopVO}{\mathsf{freeTopVO}}

\newcommand{\fw}{\mathsf{w}}
\newcommand{\dfw}{\delta}

\newcommand{\linenumber}{\makebox[2ex][r]{\rownumber\TAB}}

\newcommand{\bound}{\mathit{bound}}

\newcommand{\eps}{\epsilon}

\newcommand{\Dom}{\mathsf{Dom}}
\newcommand{\inst}[1]{\mathbf{#1}}
\newcommand{\tup}[1]{\mathbf{#1}}

\newcommand{\anc}{\mathsf{anc}}

\newcommand{\floor}[1]{\left\lfloor #1 \right\rfloor}

\newcommand{\calT}{\mathcal{T}}

\newcommand{\calB}{\mathcal{B}}

\newcommand{\calF}{\mathcal{F}}

\newcommand{\calS}{\mathcal{S}}
\newcommand{\calX}{\mathcal{X}}
\newcommand{\calY}{\mathcal{Y}}

\newcommand{\calZ}{\mathcal{Z}}

\newcommand{\TAB}{\makebox[2.5ex][r]{}}%
\newcommand{\STAB}{\makebox[1.5ex][r]{}}%
\newcommand{\LET}{\textbf{let}\xspace}%
\newcommand{\IF}{\textbf{if}\xspace}%
\newcommand{\ELSE}{\textbf{else}\xspace}%
\newcommand{\WHILE}{\textbf{while}\xspace}%
\newcommand{\FOREACH}{\textbf{foreach}\xspace}%
\newcommand{\DO}{\textbf{do}\xspace}%
\newcommand{\RETURN}{\textbf{return}\xspace}%
\newcommand{\MATCH}{\textbf{switch}\xspace}%
\newcommand{\EOF}{\textbf{EOF}\xspace}%

\usepackage{todonotes}

\newcounter{magicrownumbers}
\newcommand\rownumber{\footnotesize\stepcounter{magicrownumbers}\arabic{magicrownumbers}}
\theoremstyle{plain} 

\newtheoremstyle{cited}%
{.5\baselineskip\@plus.2\baselineskip
    \@minus.2\baselineskip}
{.5\baselineskip\@plus.2\baselineskip
    \@minus.2\baselineskip}
{\itshape}
{\parindent}
{}
{.}
{.5em}
{\textsc{\thmname{#1}} \thmnote{\normalfont#3}}

\theoremstyle{cited}
\newtheorem{citedthm}{Theorem}

\begin{document}

\title[Trade-offs in Static and Dynamic Evaluation of Hierarchical Queries]{Trade-offs in Static and Dynamic Evaluation of Hierarchical Queries}

\author[A.~Kara]{Ahmet Kara\lmcsorcid{0000-0001-8155-8070}}[a]	
\author[M.~Nikolic]{Milos Nikolic\lmcsorcid{0000-0002-1548-6803}}[b]
\author[D.~Olteanu]{Dan Olteanu\lmcsorcid{0000-0002-4682-7068}}[a]
\author[H.~Zhang]{Haozhe Zhang\lmcsorcid{0000-0002-0930-1980}}[a]

\address{University of Zurich}	
\email{kara@ifi.uzh.ch, olteanu@ifi.uzh.ch, zhang@ifi.uzh.ch}  

\address{University of Edinburgh}	
\email{milos.nikolic@ed.ac.uk}  






\begin{abstract}
  \noindent We investigate trade-offs in static and dynamic evaluation of hierarchical queries with arbitrary free variables. In the static setting, the trade-off is between the time to partially compute the query result and the delay needed to enumerate its tuples. In the dynamic setting, we additionally consider the time needed to update the query result under single-tuple inserts or deletes to the database.

Our approach observes the degree of values in the database and uses different computation and maintenance strategies for high-degree (heavy) and low-degree (light) values. For the latter it partially computes the result, while for the former it computes enough information to allow for on-the-fly enumeration.

We define the preprocessing time, the update time, and the enumeration delay as functions of the light/heavy threshold. By appropriately choosing this threshold, our approach recovers a number of prior results when restricted to hierarchical queries. 

We show that for a restricted class of hierarchical queries, our approach achieves worst-case optimal update time and enumeration delay  conditioned on the Online Matrix-Vector Multiplication Conjecture.
\end{abstract}

\maketitle


\section{Introduction}
\label{sec:introduction}

The problems of static evaluation, i.e., computing the result of a query~\cite{Yannakakis:VLDB:81,OlteanuZ15,Khamis:PODS:17,Ngo:JACM:18}, and dynamic evaluation, i.e., maintaining the result of a query under inserts and deletes of tuples to the input relations~\cite{Koch:ring:PODS:2010,Chirkova:Views:2012:FTD,DBT:VLDBJ:2014,BerkholzKS17,Idris:dynamic:SIGMOD:2017,Kara:ICDT:19}, are fundamental to relational databases. 

We consider a refinement of these two problems that decomposes the overall evaluation time into the {\em preprocessing time}, which is used to compute a data structure that represents the query result, the {\em update time}, which is the time to update the data structure under inserts and deletes to the input data, and the {\em enumeration delay}, which is the time between the start of the enumeration process and the output of the first tuple in the query result, the time between outputting any two consecutive tuples, and the time between outputting the last tuple and the end of the enumeration process~\cite{DurandFO07}. In this paper we investigate the relationship between preprocessing, update, and delay and answer questions such as, how much preprocessing time is needed to achieve sublinear enumeration delay.

We consider the static and dynamic evaluation of a subclass of $\alpha$-acyclic queries called hierarchical queries:
\begin{defiC}[\cite{Suciu:PDB:11, BerkholzKS17}]\label{def:hierarchical}
A conjunctive query is {\em hierarchical} if for any two variables, their sets of atoms in the query are either disjoint or one is contained in the other.
\end{defiC}
For instance, the query $Q(\calF) = R(A,B),S(B,C)$ is hierarchical, while $Q(\calF) =$ $R(A,B),$ $S(B,C),$ $T(C)$ is not, for any $\calF\subseteq\{A,B,C\}$.
In our study, we do not set any restriction on the set of free variables of a hierarchical query.

Hierarchical queries enjoy highly desirable tractability properties in a variety of computational settings, making them an important yardstick for database computation.
The notion of hierarchical queries used in this paper 
 has been initially introduced in the context of probabilistic databases~\cite{Suciu:PDB:11}.
The Boolean conjunctive queries without repeating relation symbols that can be computed in polynomial time on tuple-independent probabilistic databases are hierarchical; non-hierarchical queries are hard for \#P~\cite{Suciu:PDB:11}. This dichotomy was extended to non-Boolean queries with negation~\cite{Fink:TODS:16}. 
Hierarchical queries are the conjunctive queries whose provenance admits a factorized representation where each input tuple occurs a constant number of times; any factorization of the provenance of a non-hierarchical query would require a number of occurrences of the provenance of some input tuple dependent on the input database size~\cite{Olteanu:ICDT:12}. For hierarchical queries without self-joins, this read-once factorized representation explains their tractability for exact probability computation over probabilistic databases.
In the Massively Parallel Computation (MPC) model, the hierarchical queries admit parallel evaluation with one communication step~\cite{Koutris:PODS:11}. The $r$-hierarchical queries, which are conjunctive queries that become hierarchical by repeatedly removing the atoms whose complete set of variables occurs in another atom, can be evaluated in the MPC model using a constant number of steps and optimal load on every single database instance~\cite{Hu:PODS:19}. 
Hierarchical queries also admit one-step streaming evaluation in the finite cursor model~\cite{Grohe:TCS:09}. 
Under updates, the $q$-hierarchical queries are the conjunctive queries that admit constant-time update and delay~\cite{BerkholzKS17}. The $q$-hierarchical queries are a proper subclass of both the free-connex $\alpha$-acyclic and hierarchical queries. Besides being hierarchical, the following condition holds on the free variables of a $q$-hierarchical query: if the set of atoms of a free variable is strictly contained in the set of another variable, then the latter must also be free.

In this paper we characterize trade-offs in the static and dynamic evaluation of hierarchical queries. 
In the static setting, we are interested in the trade-off between preprocessing time and enumeration delay.
In the dynamic case, we additionally consider the update time. 
Section~\ref{sec:trade_offs_static_evaluation} states our main result in the static setting 
and explains how it recovers prior results on static query evaluation.
Section~\ref{sec:trade_offs_dynamic_evaluation} gives our main result in the dynamic setting
and discusses its implications. 
These two sections also overview prior work on static and dynamic query evaluation. 
Section~\ref{sec:preliminaries} introduces the basic notions underlying our approach. 
Sections~\ref{sec:preprocessing}-\ref{sec:updates} detail the preprocessing, enumeration, and 
update stages of our approach.
Section~\ref{sec:lower_bound} shows that for a restricted class of hierarchical queries, our approach achieves worst-case optimal update time and enumeration delay, conditioned on the Online Matrix-Vector Multiplication Conjecture.
We illustrate our approach using two detailed examples 
in Section~\ref{sec:examples} and conclude in 
Section~\ref{sec:conclusion}.
The proofs of the main theorems in Sections~\ref{sec:trade_offs_static_evaluation} and \ref{sec:trade_offs_dynamic_evaluation} and the propositions in Sections~\ref{sec:preprocessing}-\ref{sec:updates} are deferred to Appendices~\ref{appendix:main_results_static}-\ref{appendix:updates}. 
The proofs of the propositions in Sections~\ref{sec:preliminaries} and 
\ref{sec:lower_bound} can be found in the technical report~\cite{Trade_Offs_arxiv}.

A preliminary version of this work appeared in PODS 2020~\cite{KaraNOZ2020}. We extended it 
as follows. 
We overviewed in greater depth and breadth the related work for a more complete picture of the state of the art
 (Sections~\ref{sec:introduction}-\ref{sec:trade_offs_dynamic_evaluation}).
We added new motivating examples to demonstrate that our approach achieves better overall evaluation time than existing approaches both in the static and dynamic cases  
(Sections~\ref{sec:trade_offs_static_evaluation} and~\ref{sec:trade_offs_dynamic_evaluation}).
We included necessary background on the computational model and width measures  (Section~\ref{sec:preliminaries}).
We added a detailed description of the algorithms (\textsc{Union} and \textsc{Product}) used by the enumeration procedure
of our approach (Section~\ref{sec:enumeration}). 
We included the procedures for major and minor rebalancing in case of updates and gave the procedure for the maintenance of a query result under sequences of updates (Section~\ref{sec:updates}). 
Finally, we included complete proofs of the main results and the main statements on the preprocessing, enumeration, 
and update stages of our approach (Appendices~\ref{appendix:main_results_static}-\ref{appendix:updates}).

\section{Trade-offs in Static Query Evaluation}
\label{sec:trade_offs_static_evaluation}

\begin{figure*}[t]
\scalebox{0.85}{
	\begin{small}
		\hspace{-0.6cm}
		\begin{tikzpicture}

			\node at(-1,3.65){\bf Complexities for};
			\node at(-1,3.25){\bf Hierarchical Queries};
			
		
			\node at(-0.5,0){
				\begin{tikzpicture}[scale=1.1]
					\begin{axis}[
							legend style={cells={align=left}},
							grid=none,
							xmin=0, xmax=1.1, ymin=0, ymax=2.68,
							every axis plot post/.append style={mark=none},
							xtick ={0, 1},
							ytick ={0, 1, 1.5, 2.5},
							xticklabels={\footnotesize{$0$}, \footnotesize{$1$}},
							yticklabels={$ $,\footnotesize{$1$},\footnotesize{$\fw-1$},\footnotesize{$\fw$}},
							y=1.2cm,
							x=2.3cm,
							axis lines=middle,
							axis line style={->},
							x label style={at={(axis description cs:1.16,-0.06)}},
							xlabel={\small{$\eps$}},
							y label style={at={(axis description cs:-0.3,1.17)},align=center},
							ylabel=\footnotesize{$\log_N$ time},
							axis x line*=bottom,
							axis y line*=left,
							legend style={at={(-0.2,0.5)},draw=none},
						]

						\addplot[color=goodgreen,mark=none,domain=0:1,line width = 1pt,dotted] coordinates{
								(0, 1)
								(1, 2.5)
							};

						\addplot[color=red,mark=none,domain=0:1,thick,dashed] coordinates{
								(0, 0)
								(1, 2.5)
							};

						\addplot[color=violet,mark=none,domain=0:1,thick] coordinates{
								(0, 1)
								(1, 0)
							};

						\addplot[color=red,mark=none,domain=0:1,thick,dashed] coordinates{
								(0, 0)
								(1, 1.5)
							};

					\end{axis}

					\draw[dotted, line width = 0.6pt] (2.3,0) -- (2.3,3);

					\draw[dotted, line width = 0.6pt] (0,1.2) -- (2.3,1.2);

					\draw[dotted, line width = 0.6pt] (0,1.8) -- (2.3,1.8);

					\draw[dotted, line width = 0.6pt] (0,3) -- (2.3,3);

					\draw[dotted, line width = 1pt, color= goodgreen] (-0.8,-0.8) -- (-0.3,-0.8);
					\node at (2,-0.8){ \footnotesize{preprocessing time}
						\footnotesize{\ ${1+(\fw-1)\eps}$}};

					\draw[dashed, line width = 0.7pt, color= red] (-0.8,-1.2) -- (-0.3,-1.2);
					\node at (0.9,-1.2){\footnotesize{update time}
						\footnotesize{\ $\dfw\eps$}};

					\draw[line width = 0.7pt, color= violet] (-0.8,-1.6) -- (-0.3,-1.6);
					\node at (0.65,-1.6){ \small{delay}
						\footnotesize{\ $1 - \eps$}};

				\end{tikzpicture}

			};

			\node at(4.5,3.65){\bf Trade-offs in};
			\node at(4.5,3.25){\bf Static Query Evaluation};
			
			\node at (4.5,0){

				\vspace{-1cm}
				\begin{tikzpicture}
					\begin{axis}[
							grid=none,
							xmin=0, xmax=3.1, ymin=0, ymax=1.2,
							every axis plot post/.append style={mark=none},
							xtick ={0,  1, 2.8},
							ytick ={0,  1},
							xticklabels={\footnotesize{$0$},\footnotesize{$1$}, \footnotesize{$\fw$}},
							yticklabels={\footnotesize{$ $},\footnotesize{$1$}},
							y=1.5cm,
							x=1.2cm,
							axis lines=middle,
							axis line style={->},
							x label style={at={(axis description cs:1.18,-0.06)}},
							xlabel={},
							y label style={at={(axis description cs:-0.9,1.1)},align=center},
							axis x line*=bottom,
							axis y line*=left
						]

					\end{axis}
					\node at(-0.1,2.1) {\footnotesize{$\log_N$ delay}};
					\node at(2.4,-0.6) {\footnotesize{$\log_N$ preprocessing time}};

					\filldraw[blue] (3.4,0) circle (2pt) node {};
					\node at(4.05,0.25) { {conjunctive}};

					\filldraw[blue] (1.2,1.5) circle (2pt) node {};

					\node at(1.3,1.8) {{$\alpha$-acyclic}};

					\filldraw[blue] (1.2,0) circle (2pt) node {};
					\node at(1.3,0.3) { {free-connex}};

					\draw[line width = 1pt, color= blue] (1.2,1.5) -- (3.4,0);
					\node[rotate = 326] at(2.38,1) {{\color{blue} hierarchical}};

					\draw[dotted, line width = 0.6pt] (0,1.5) -- (3.4,1.5);

					\draw[dotted, line width = 0.6pt] (1.2,0) -- (1.2,1.5);

					\draw[dotted, line width = 0.6pt] (3.4,0) -- (3.4,1.5);

				\end{tikzpicture}

			};


			\node at(10.5,3.65){\bf Trade-offs in};
			\node at(10.5,3.25){\bf Dynamic Query Evaluation};

			\node at (10.3,0) {

				\tdplotsetmaincoords{75}{170}
				\begin{tikzpicture}[xscale=1.8, yscale=0.9,tdplot_main_coords]
					\coordinate (O) at (0,0,0);


					\draw[thick,->] (0,0,0) -- (1.2,0,0) node[anchor=north]{\footnotesize $\log_{N}${delay}\phantom{ABC}};
					\draw[thick,->] (0,0,0) -- (0,5,0);
					\draw[] (0,2.5,0) -- (0,4,0);
					\draw[very thick,->] (0,4,0) -- (0,5,0) node[anchor=north]{\footnotesize $\log_{N}${preprocessing time}};

					\draw[thick,->] (0,0,0) -- (0,0,3.5) node[anchor=south]{\footnotesize $\log_{N}${update time}};

					\node[] at (0.06,0,0.15) () {\footnotesize $0$};
					\coordinate (P1) at (1,1,0);
					\node[] at (0.9,1,-0.3) () {\footnotesize $(1,0,1)$};
					\coordinate (P1x) at (1,0,0);
					\node[] at (0.9,-0.6,0) () {\footnotesize $1$};

					\coordinate (P1y) at (0,1,0);
					\node[] at (0.05,1,-0.2) () {\footnotesize $1$};

					\coordinate (P2) at (0,4,3);
					\coordinate (P2xy) at (0,4,0);
					\coordinate (P2y) at (0,4,0);

					\coordinate (P2z) at (0,0,3);

					\coordinate (P3) at (0,4,2);
					\coordinate (P3z) at (0,0,2);

					\draw[dotted, color=black, line width =0.6] (P1x) -- (P1);
					\draw[dotted, color=black,line width =0.6] (P1y) -- (P1);
					\node[] at (-0.05,3.2,-0.3) () {\footnotesize $\fw$};
					\node[] at (0.1,0,3) () {\footnotesize $\dfw$};

					\draw[dotted, color=black,line width =0.6] (P2) -- (P2y);
					\draw[dotted, color=black,line width =0.6] (P2) -- (P2z);

					\draw[dotted, color=black, line width=0.6pt] (0,1,0.8) -- (0,0,0.8);
					\draw[dotted, color=black] (0,1,0.8) -- (0,1,0);
					\filldraw[blue] (0,1,.8) ellipse(1pt and 2pt);
					\node[] at(-0.7,0,0.7) () { free-connex};
					\node[] at(0.057,0,0.65) () {\footnotesize 1};

					\node[] at(-0.3,4,3.3) () { conjunctive};

					\node[] at(-0.77,0,-0.12) () { q-hierarchical};

					\draw[color=white, very thick, opacity=0.8] (0,0,0.7) -- (0,0,1.4);
					\draw[color=blue, thick] (P1) -- (P2);

					\node[color=blue,rotate = 37] at(0.3,1.8,1.7) { hierarchical};
					\filldraw [color=blue] (0,4,3) ellipse(1pt and 2pt);
					\filldraw[color=blue] (0,1,0) ellipse(1pt and 2pt);

				\end{tikzpicture}

			};
		\end{tikzpicture}
	\end{small}
}	
	\caption{
		Left: Preprocessing time, enumeration delay, and amortized update time for a hierarchical query with
		static width $\fw$ and dynamic width $\dfw$ ($\dfw$ can be $\fw$ or $\fw-1$, hence the two red
		lines for the update time).
		Middle and right: Trade-offs in static and dynamic evaluation.
		Our approach achieves each blue point and
		each point on the blue lines. Prior approaches are represented by the blue points.}
	\label{fig:tradeoffs_3_dim}
	\vspace*{-.2em}
\end{figure*}

Our main result for the static evaluation of hierarchical queries is stated next. 
\begin{thm}\label{thm:main_static}
Given a hierarchical query with static width $\fw$, a database  
of size $N$, and $\eps \in [0,1]$, the query result can be enumerated with $\bigO{N^{1-\eps}}$ delay after $\bigO{N^{1 + (\fw -1)\eps}}$ 
preprocessing time.
\end{thm} 

The measure $\fw$, previously introduced as $s^\uparrow$~\cite{OlteanuZ15}, generalizes the fractional hypertree width~\cite{Marx:TRANSALG:10} from Boolean to arbitrary conjunctive queries. This is equivalent to the FAQ-width in case of  Functional Aggregate Queries over a single semiring~\cite{FAQ:PODS:2016}. In this paper, we refer to this measure as the static width of the query (Definition~\ref{def:fac_width}).

Theorem~\ref{thm:main_static} expresses the runtime components as functions of a parameter $\eps$.
The dotted green line and the purple line in the left plot in Figure~\ref{fig:tradeoffs_3_dim} depict 
the preprocessing time and respectively the enumeration delay. 
The middle plot in Figure~\ref{fig:tradeoffs_3_dim} visualizes the trade-off between the two components. 
Our approach achieves each blue point and each point on the blue line. 
Prior approaches are represented by the blue points in the trade-off space.
By appropriately setting $\eps$, our approach recovers prior results restricted to hierarchical queries.
For $\eps=0$, both the preprocessing time and the delay become $O(N)$, as for $\alpha$-acyclic queries~\cite{BaganDG07}.
For $\eps=1$, we obtain $\bigO{N^\fw}$ preprocessing time and $O(1)$ delay as for conjunctive queries~\cite{OlteanuZ15}. 
Free-connex acyclic queries are a special class of queries that enjoy linear preprocessing time and constant delay~\cite{BaganDG07}. We recover this result as follows. First, we observe that any free-connex hierarchical query has static width $\fw=1$. This means that the preprocessing time remains $O(N)$ regardless of $\eps$; we then choose $\eps=1$ to obtain $O(1)$ delay.
For bounded-degree databases, i.e., where each value appears at most $c$ times for some constant $c=N^\beta$,
first-order queries admit $O(N)$ preprocessing time and $O(1)$ delay~\cite{DurandFO07,Kazana11}. We recover the $O(1)$ delay using $\eps=1$.
The preprocessing time becomes $\bigO{N\cdot (N^\beta)^{\fw-1}}=\bigO{N}$ if our approach uses the constant upper bound $c$ instead of the upper bound $N^{\eps}$ on the degrees. The left Venn diagram in Figure~\ref{fig:venn_diagrams} 
depicts the relationship of our result in Theorems~\ref{thm:main_static} with prior results.

\begin{figure*}[t]
	\newcommand{\sspace}{\hspace{0.13em}}
	\begin{small}
\scalebox{0.79}{
		\begin{tikzpicture}[scale=1.1]
			\node at(0,4.05){\bf Static Query Evaluation};
			\node at(0,3.65){\textit{preprocessing time{\sspace}/{\sspace}enumeration delay}};
			\draw[line width = 0.5pt] (0,-0.25) ellipse (2.3cm and 3.1cm);
			\node at (0,2.47)  {\bf conjunctive};
			\node at (0,2.1)  {$\bigO{N^{\fw}}${\sspace}/{\sspace}$\bigO{1}$};
			\node at (0,1.7){\cite{OlteanuZ15}};

			\draw[line width = 0.5pt] (-0,-0.85) ellipse (2cm and 2.3cm);
			\node at (0,1.1)  {\bf $\alpha$-acyclic};
			\node at (0,0.7)  {$\bigO{N}${\sspace}/{\sspace}$\bigO{N}$};
			\node at (0,0.35)  {\cite{BaganDG07}};

			\draw[line width = 0.5pt] (0,-0.8) ellipse (1.85cm and 0.9cm);
			\node at (0,-0.2)  {\color{blue} \bf hierarchical};
			\node at (0,-0.63)  {\color{blue} $\bigO{N^{1 + (\fw -1 )\eps}}${\sspace}/{\sspace}$\bigO{N^{1 - \eps}}$};
			\node at (0,-1.03)  {\color{blue} $\eps \in [0,1]$};

			\draw[line width = 0.5pt]  (0,-2.15) ellipse (1.1cm and 0.8cm);
			\node at (0,-1.85)  {\bf free-connex};
			\node at (0,-2.25)  {$\bigO{N}${\sspace}/{\sspace}$\bigO{1}$};
			\node at (0,-2.6)  {\cite{BaganDG07}};

		\end{tikzpicture}
}		
\hspace{-0.7cm}		
		\scalebox{0.79}{
		\begin{tikzpicture}[scale=1.2]
			\node at(0,4.05){\bf Dynamic Query Evaluation};
			\node at(0,3.65){\textit{preprocessing time{\sspace}/{\sspace}enumeration delay{\sspace}/{\sspace}update time}};
			\draw[line width = 0.5pt] (0,0.3) ellipse (5.5cm and 3.1cm);
			\node at (0,3.15)  {\bf conjunctive};
			\node at (0,2.75)  {$\bigO{N^{\fw}}${\sspace}/{\sspace}$\bigO{1}${\sspace}/{\sspace}$\bigO{N^{\dfw}}$
			\ \ \cite{Nikolic:SIGMOD:18}};

			\node at (0,2.25)  {triangle query\ \ \ $\bigO{N^{\frac{3}{2}}}${\sspace}/{\sspace}$\bigO{1}${\sspace}/{\sspace}$\bigO{N^{\frac{1}{2}}}^{\ast}$ \ \ \cite{Kara:ICDT:19}};

			\draw[line width = 0.5pt] (0,-0.2) ellipse (5.2cm and 2.2cm);
			\node at (0,1.7)  {\bf $\alpha$-acyclic};

			\draw[line width = 0.5pt]  (2.15,-0.28) ellipse (2.8cm and 1.45cm);
			\node[rotate=-15] at (3.5,0.63)  {\bf free-connex};
			\node[rotate=-15] at (3.5,0.2)  {$\bigO{N}${\sspace}/{\sspace}$\bigO{1}${\sspace}/{\sspace}$\bigO{N}$};
			\node at (3.25,-0.35){\cite{Idris:dynamic:SIGMOD:2017}};

			\node at (0.9,0.3)  {\bf q-hierarchical};
			\node at (0.9,0.05)  {$=$};
			\node at (0.9,-0.15)  {\color{blue}\bf $\dfw_0$-hierarchical};
			\node at (1,-0.53)  {\color{blue}
				$\fw = 1$, $\dfw = 0$};
			\node at (1,0.65){\cite{BerkholzKS17}};

			\draw[line width = 0.5pt] (-0.05,-0.7505) -- (2.32,-0.7505);

			\draw[line width = 0.5pt] (0,-0.75) arc(100:273:5cm and 0.452cm);

			\draw[line width = 0.5pt] (-1.25,-0.3) ellipse (3.7cm and 1.75cm);
			\node[rotate = 15] at (-2.6,0.9)  {\color{blue} \bf hierarchical};

			\node[rotate = 15] at (-2.3,0.4)  {\color{blue}
			$\bigO{N^{1 + (\fw -1 )\eps}}${\sspace}/{\sspace}$\bigO{N^{1 - \eps}}${\sspace}/{\sspace}$\bigO{N^{\dfw\eps}}^{\ast}$};

			\node[rotate = 15] at (-1.9,0.1)  {\color{blue}
				$\eps \in [0,1]$};

			\node at (-0.7,-1.05)  {\color{blue}\bf $\dfw_1$-hierarchical};

			\node at (-0.7,-1.4)  {\color{blue}
				$\fw \in \{1,2\}$, $\dfw =1$};

		\end{tikzpicture}
		}
	\end{small}
	\caption{Landscape of static and dynamic query evaluation.
		$\fw$: static width; $\dfw$: dynamic width; *:~amortized time.}
	\label{fig:venn_diagrams}
\end{figure*}

The next example demonstrates how the complexities of our approach in the static case 
imply lower overall evaluation time than existing approaches.
\begin{exa}
\label{ex:static_overall_time} 
Consider the hierarchical query 
$Q(A,C) = R(A,B), S(B,C).$ Let us assume that the input relations are of size $N$. Then, it takes quadratic time to compute the list of tuples in the query result of $Q$. (As it will become clearer later, this query has static width $2$, which explains the $O(N^2)$ time complexity for the evaluation of $Q$.)

An {\em eager} evaluation does just this: It readily computes the list of tuples in the query result of $Q$. This requires quadratic preprocessing time, after which the tuples in the query result can be enumerated with constant delay~\cite{Olteanu:ICDT:12}.

In contrast, a {\em lazy} evaluation approach computes the first tuple in the query result then the second tuple and so on. This can be done using linear preprocessing time followed by linear enumeration delay for each tuple in the result~\cite{BaganDG07}. It is conjectured that the delay cannot be lowered to constant after linear-time preprocessing for $Q$~\cite{BaganDG07}. (The explanation is that $Q$ is not free-connex, a notion we will introduce in Section~\ref{sec:preliminaries}.) 

Our approach achieves $\bigO{N^{1+\eps}}$ preprocessing time and 
  $\bigO{N^{1-\eps}}$ enumeration delay for any $\eps \in [0,1]$.
The complexities of the eager, lazy, and our approach are as follows:
     \begin{center}
	\renewcommand{\arraystretch}{1.1}
	\begin{tabular}{l| l l }
	approach  & preprocessing & delay \\\hline  
		lazy  & $\bigO{N}$   & $\bigO{N}$\\
		eager & $\bigO{N^2}$  & $\bigO{1}$\\		
		ours  & $\bigO{N^{1+\eps}}$    & $\bigO{N^{1-\epsilon}}$
	\end{tabular} 
	\end{center}
 Our approach recovers the lazy approach at $\eps = 0$ and the eager approach at $\eps =1$.
 For any $\eps \in (0,1)$, it achieves new trade-offs between preprocessing time and enumeration delay.

Given that the input relations have size $N$, the AGM bound~\cite{AtseriasGM13}
implies that  the result of the query has at most $N^2$ tuples. 
Assume that we want to enumerate $N^{\gamma}$ tuples from the result, for some $0 \leq \gamma \leq 2$.
In the following table, the second to  fourth rows give the exponents of the overall evaluation times 
achieved by the lazy, eager, and our approaches for different values of $\gamma$.
The last row gives the $\eps$ values at which we achieve the complexities of our approach.
   \begin{center}
	\renewcommand{\arraystretch}{1.2}
   \begin{tabular}{l| l l l l l l}
	$\gamma$  & $0$ & $\frac{1}{2}$ & $1$ & $1\frac{1}{2}$ & $2$   \\[2pt]\hline 
	  lazy   & $1$ & $1\frac{1}{2}$ & $2$ & $2\frac{1}{2}$ & $3$     \\[1pt]
	eager  & $2$ & $2$ & $2$ & $2$ &   $2$   \\[2pt] \hline
              ours & \cellcolor{yellow!50}$1$ & \cellcolor{teal!50!yellow!50}$1\frac{1}{4}$ & \cellcolor{teal!50!yellow!50}$1\frac{1}{2}$ &\cellcolor{teal!50!yellow!50} $1\frac{3}{4}$ &\cellcolor{yellow!50}$2$  \\[1pt]
                 $\eps$ & $0$ & $\frac{1}{4}$ & $\frac{1}{2}$ & $\frac{3}{4}$ & $1$ 
\end{tabular}
\end{center}
For instance, if $\gamma = 1\frac{1}{2}$, the lazy approach requires $\bigO{N + N^{1+\frac{1}{2}}\cdot N} = \bigO{N^{2+\frac{1}{2}}}$, 
the eager approach requires {$\bigO{N^2 + N^{1+\frac{1}{2}}\cdot 1} = \bigO{N^{2}}$}, and our approach needs
only  {$\bigO{N^{1+\frac{3}{4}} + N^{1+\frac{1}{2}}N^{\frac{1}{4}}} = \bigO{N^{1+\frac{3}{4}}}$} time
at $\eps = \frac{3}{4}$. 
In case $\gamma$ is equal to $\frac{1}{2}$, $1$, or $1\frac{1}{2}$, the overall computation time of our approach 
(highlighted in green) is strictly lower than
the eager and lazy approaches. For the other two cases shown in the table, our approach recovers the lower complexity of the 
prior approaches (highlighted in yellow). 
  \qed
\end{exa}

\subsection{Further Prior Work on Static Query Evaluation}
We complement our discussion with 
further prior work on static query evaluation.
Figure~\ref{fig:taxonomy} gives a  taxonomy of works in this area.

\begin{figure*}[t]
\footnotesize
\begin{tabular}{@{\hskip 0.0in}l@{\hskip 0.08in}l@{\hskip 0.08in}l@{\hskip 0.08in}l@{\hskip 0.08in}l@{\hskip 0.0in}}
\toprule
Class of Queries & Preprocessing & Delay & Extra Space & Source \\
\midrule

f.c. $\alpha$-acyclic CQ$^{\neq}$ & $\bigO{N}$ & $\bigO{1}$ & $\bigO{N}$ & \cite{BaganDG07} \\
f.c. $\beta$-acyclic negative CQ & $\bigO{N}$ & $\bigO{1}$ & -- & \cite{BraultPhD13,Brault2012} \\
f.c. signed-acyclic CQ & $\bigO{N\,(\log N)^{|Q|}}$ & $\bigO{1}$ & -- & \cite{BraultPhD13} \\
Acyclic CQ$^{\neq}$ & $\bigO{N}$ & $\bigO{N}$ & $\bigO{N}$  & \cite{BaganDG07} \\
CQ$^{\neq}$ of f.c. treewidth $k$ & $\bigO{|\text{Dom}|^{k+1} + N}$ & $\bigO{1}$ & -- & \cite{BaganDG07} \\
CQ & $\bigO{N^{\fw(Q)}}$ & $\bigO{1}$ & $\bigO{N^{\fw(Q)}}$ & \cite{OlteanuZ15,FAQ:PODS:2016}\\
Full CQ with access patterns & $\bigO{N^{\rho^{*}(Q)}}$ & $\bigO{\tau}$ & $\bigO{N + N^{\rho^{*}(Q)}/\tau}$ & \cite{Deep:2018}\\[0.5ex]
\midrule
CQ on \underbar{X}-structures (trees, grids) & $\bigO{N}$ & $\bigO{N}$ & -- & \cite{BaganDFG10} \\
FO on bound. degree & $\bigO{N}$ & $\bigO{1}$ & -- & \cite{DurandFO07,Kazana11} \\
FO on bound. expansion   & $\bigO{N}$ & $\bigO{1}$ & -- & \cite{Kazana13PODS} \\
FO on local bounded expansion & $\bigO{N^{1+\gamma}}$ & $\bigO{1}$ & -- & \cite{Segoufin17} \\
FO on low degree  & $\bigO{N^{1+\gamma}}$ & $\bigO{1}$ & $\bigO{N^{2+\gamma}}$ & \cite{Durand14} \\
FO on nowhere dense & $\bigO{N^{1+\gamma}}$ & $\bigO{1}$ & $\bigO{N^{1+\gamma}}$ & \cite{Schweikardt18} \\
MSO on Bounded treewidth & $\bigO{N}$ & $\bigO{1}$ & -- & \cite{Bagan06,Kazana13TOCL} \\
\bottomrule
\end{tabular}

\caption{Prior work on the trade-off between preprocessing time, enumeration delay, and extra space  for different classes of queries (Conjunctive Queries, First-Order, Monadic Second-Order) and static databases under data complexity; f.c. stands for  free-connex.  Parameters: Query $Q$ with factorization width $\fw$~\cite{OlteanuZ15} and fractional edge cover number $\rho^*$~\cite{AtseriasGM13}; database of size $N$; slack $\tau$ is a function of $N$ and $\rho^*$; $\gamma > 0$. Most works do not discuss the extra space utilization (marked by --).}
\label{fig:taxonomy}
\end{figure*}

Prior work exhibits a dependency between the space and enumeration delay for conjunctive queries with access patterns~\cite{Deep:2018}.
It constructs a succinct representation of the query result that allows for enumeration of tuples over some variables under value bindings for all other variables. It does not support enumeration for queries with projection, as addressed in our work. It also states Example~\ref{ex:intro} as an open problem.

The result of any $\alpha$-acyclic conjunctive query can be enumerated with constant delay after linear-time preprocessing if and only if it is free-connex. This is under the conjecture that Boolean multiplication of $n\times n$ matrices cannot be done in $O(n^2)$ time~\cite{BaganDG07}. 
More recently, this was shown to hold also under the hypothesis that the existence of a triangle in a hypergraph of $n$ vertices cannot be tested in time $\bigO{n^2}$ and that for any $k$, testing the presence of a $k$-dimensional tetrahedron cannot be decided in linear time~\cite{BraultPhD13}. 
The free-connex characterization generalizes in the presence of functional dependencies~\cite{CarmeliK18}. An in-depth pre-2015 overview on constant-delay enumeration is provided by Segoufin~\cite{Segoufin15}. 

There are also enumeration algorithms for document spanners~\cite{AmarilliBMN19} and satisfying valuations of circuits~\cite{AmarilliBJM17}.   

\section{Trade-offs in Dynamic Query Evaluation}
\label{sec:trade_offs_dynamic_evaluation}
Our main result for the dynamic query evaluation generalizes the static case.

\begin{thm}\label{thm:main_dynamic}
Given a hierarchical query with static width $\fw$ and dynamic width $\dfw$, a database of size $N$, and $\eps \in [0,1]$, 
the query result can be enumerated with $\bigO{N^{1-\eps}}$ delay after $\bigO{N^{1 + (\fw -1)\eps}}$ preprocessing time and 
$\bigO{N^{\dfw\eps}}$ amortized update time for single-tuple updates.
\end{thm}

The left plot in Figure~\ref{fig:tradeoffs_3_dim} depicts  
the preprocessing time (dotted green line), the update time (dashed red lines), and the enumeration delay (purple line) of our approach in the dynamic case.
For hierarchical queries, the dynamic width $\dfw$ can be equal to either the static width $\fw$ or $\fw-1$ (Proposition~\ref{prop:width_delta_inequal}). The plot hence shows two dashed red lines for the update time.
The right plot in Figure~\ref{fig:tradeoffs_3_dim} depicts the trade-off between the three components.
Our approach can achieve sublinear amortized update time and delay for hierarchical queries with arbitrary free variables 
(Figure~\ref{fig:tradeoffs_3_dim} left and right).
For any $\eps=\frac{1}{\dfw+\alpha}>0$ with $\alpha > 0$, our algorithm has update time $\bigO{N^{1-\alpha\cdot\frac{1}{\dfw+\alpha}}}$ and delay $\bigO{N^{1-\frac{1}{\dfw+\alpha}}}$.

The update time for a single tuple is at most the preprocessing time: $\dfw\eps\leq \fw\eps\leq \eps+(\fw-1)\eps \leq 1+(\fw-1)\eps$. 
If $\dfw=\fw-1$, then $\dfw\eps = (\fw-1)\eps$, i.e., the update time is an $\bigO{N}$ factor less than the preprocessing time. 
The complexity of preprocessing thus amounts to inserting $N$ tuples in an initially empty database using our update mechanism. If $\dfw=\fw$, then inserting $N$ tuples would need $\bigO{N^{1 + (\fw-1)\eps + \eps}}$ time, which is an $\bigO{N^\eps}$ factor more than the complexity of one bulk update using our preprocessing algorithm. This suggests a gap between single-tuple updates and bulk updates. 
A similar gap highlighting a fundamental limitation of single-tuple updates has been shown for the Loomis-Whitney query that generalizes the triangle query from a join of three binary relations to a join of $n$ $(n-1)$-ary relations: The amortized update time for single-tuple updates is $\bigO{N^{1/2}}$, which is worst-case optimal unless the {Online Matrix-Vector Multiplication} conjecture fails~\cite{ExtendedTriangleCount}. Inserting $N$ tuples in the empty database would cost $\bigO{N^{3/2}}$, yet the query can be computed 
in the static setting in time $\bigO{N^{\frac{n}{n-1}}}$~\cite{Ngo:JACM:18}.

Amortized $\bigO{N^{\delta\eps}}$ update time means that, given any sequence of updates, the {\em average} cost of a single 
update is  $\bigO{N^{\delta\eps}}$. Since updates can change the data structure, our approach needs to do a rebalancing 
step whenever the data structure gets out of balance.  
The time needed for a single update without rebalancing is $\bigO{N^{\delta\eps}}$ in the worst case. 
A rebalancing step can require super-linear time (Propositions~\ref{prop:major_time}
and \ref{prop:minor_time}). We show that for any update sequence, the overall time needed for the updates and rebalancing steps, when averaged over the number of updates in the sequence, remains $\bigO{N^{\delta\eps}}$ in the worst case (Proposition~\ref{prop:amortized_update_time}).
Using classical de-amortization techniques~\cite{KosarajuP98}, we can adapt our update mechanism to obtain non-amortized $\bigO{N^{\delta\eps}}$ update time. 
The de-amortization strategy is analogous to the one used for the update mechanism  of triangle queries (Section 10 in \cite{Kara:TODS:2020}), which performs more frequent but less time-consuming rebalancing steps.

Theorem~\ref{thm:main_dynamic} recovers prior work on conjunctive queries~\cite{Nikolic:SIGMOD:18}, free-connex acyclic queries\cite{Idris:dynamic:SIGMOD:2017}, and q-hierarchical queries~\cite{BerkholzKS17} by setting $\eps = 1$ (Figure~\ref{fig:tradeoffs_3_dim} right).
For hierarchical queries in general, our approach achieves the same complexities as prior work on conjunctive queries when restricted to hierarchical queries. 
For free-connex queries, we obtain linear-time preprocessing and update and constant-time delay since $\fw = 1$ (Proposition~\ref{prop:free-connex_static_one}) and then $\dfw \in \{0,1\}$ (Proposition~\ref{prop:width_delta_inequal}) for these queries.
For q-hierarchical queries, we obtain linear-time preprocessing and constant-time update and delay since $\fw = 1$ and $\dfw = 0$.
Existing maintenance approaches, e.g, classical first-order IVM~\cite{Chirkova:Views:2012:FTD} and higher-order recursive IVM~\cite{DBT:VLDBJ:2014}, 
DynYannakakis~\cite{Idris:dynamic:SIGMOD:2017}, and F-IVM~\cite{Nikolic:SIGMOD:18}, can achieve constant delay for general hierarchical queries yet after at least linear-time updates.
The  {right Venn diagram in Figure~\ref{fig:venn_diagrams}} relates Theorem~\ref{thm:main_dynamic}  with prior results.

The next example illustrates that our approach achieves better overall evaluation time than existing approaches when considering 
a sequence of updates.

\begin{exa}
\label{ex:dynamic_overall_time} 
Let us consider the (free-connex hierarchical) query 
$Q(A) = R(A,B), S(B).$
The query has static  and dynamic width $1$. 
We assume that the input relations are of size $N$ and consider the 
dynamic setting. 

A {\em lazy} evaluation approach requires no preprocessing: For each single-tuple update, it only updates the input relations
without propagating the changes to the query result. Before enumerating the $A$-values in the query result, 
it first scans the relation $R$ to collect all $A$-values that are paired with $B$-values contained in $S$.
This takes linear time. Afterwards, the approach  can enumerate the $A$-values with constant delay. 

An {\em eager} evaluation approach precomputes the initial result in linear time.
On a single-tuple update, it computes the delta query obtained by fixing the variables of one relation
to constants. For an update $\delta R(a,b)$ to $R$, the delta query $\delta Q(a) = \delta R(a,b), S(b)$ can be computed in constant time. For an update $\delta S(b)$ 
 to $S$, the delta query $\delta Q(A) = R(A,b),\delta S(b)$ can be computed in linear time. In general, the update time is linear. Since the query result is materialized and eagerly maintained, the $A$-values in the result can be enumerated after an update with constant delay.   

For this query, our approach achieves $\bigO{N}$ preprocessing time,
 $\bigO{N^{\eps}}$ update time, and 
  $\bigO{N^{1-\eps}}$ enumeration delay for any $\eps \in [0,1]$. 
The following table summarizes the preprocessing-update-delay trade-off achieved by the three approaches:

\begin{center}	
	\renewcommand{\arraystretch}{1.1}
	\begin{tabular}{l| l l l }
	approach  & preprocessing & update     & delay \\\hline
		lazy  & $\bigO{1}$    & $\bigO{1}$ & $\bigO{N}$\\	 
		eager & $\bigO{N}$  & $\bigO{N}$ & $\bigO{1}$\\
		ours  & $\bigO{N}$    & $\bigO{N^{\epsilon}}$ & $\bigO{N^{1-\epsilon}}$
	\end{tabular} 
	\end{center}

\smallskip	
Our approach recovers the lazy and eager approaches  by setting $\epsilon$ to $0$ and respectively $1$, with one exception: it cannot recover the constant preprocessing time in the lazy approach as it requires one pass over the input data.


Consider now a sequence of $N^m$ updates, each followed by one access request to enumerate $N^{\gamma}$ values out of the at most $N$ $A$-values in the query result, for $m, \gamma \in [0,1]$.
With the eager and lazy approaches, this sequence takes time (excluding preprocessing) $\bigO{N^m(N + N^{\gamma})}$, which is $\bigO{N^{m+1}}$ since $\gamma \leq 1$. 
With our approach, the sequence takes  $\bigO{N^{m} (N^\eps + N^{\gamma}N^{1-\eps})}$
$=$ $\bigO{N^{m+\eps} + N^{m+\gamma +1 -\eps}}$. Depending on the values of $m$ and $\gamma$, we can tune our approach (by appropriately setting $\eps$) to minimize the overall time to execute the bulk of updates and access requests. For $\gamma < 1$ and any $m$, our approach has consistently lower complexity than the lazy/eager approaches, while for $\gamma = 1$ and any $m$ it matches that of the lazy/eager approaches. The complexity of processing the sequence of updates and access requests is shown in the  next table for various values of $m$ and $\gamma$:
\begin{center}	
\begin{tikzpicture}
\node at (-1,2){\small our approach};
\draw [decorate,
    decoration = {brace,
        amplitude=5pt},line width=0.7pt] (-3.1,1.5) --  (0.3,1.5);
        
        \begin{scope}[xshift=3.5cm]
        \node at (-0.8,2){\small eager/lazy approaches};
\draw [decorate,
    decoration = {brace,
        amplitude=5pt},line width=0.7pt] (-2.7,1.5) --  (0.7,1.5);
        \end{scope}
\node at (0,0){
\begin{minipage}{0.95\textwidth}
\begin{center}	
{\renewcommand{\arraystretch}{1.2}
\begin{tabular}{ l | l l l l l || l l l l l}
	\diagbox{$m$}{$\gamma$} & 0   & $\frac{1}{4}$ & $\frac{1}{2}$ & $\frac{3}{4}$   & 
	$1$ & $0$  & $\frac{1}{4}$ & $\frac{1}{2}$ & $\frac{3}{4}$ & $1$ \\\hline
	0   & $\cellcolor{teal!50!yellow!50}\frac{1}{2}$ & $\cellcolor{teal!50!yellow!50}\frac{5}{8}$ & $\cellcolor{teal!50!yellow!50}\frac{3}{4}$   & 
	$\cellcolor{teal!50!yellow!50}\frac{7}{8}$ & $\cellcolor{yellow!50}1$  & $1$ & $1$ & $1$ & $1$ & $\cellcolor{yellow!50}1$\\
	 $\frac{1}{2}$ & $\cellcolor{teal!50!yellow!50}1$   & $\cellcolor{teal!50!yellow!50}1\frac{1}{8}$ & 
	 $\cellcolor{teal!50!yellow!50}1\frac{1}{4}$ & $\cellcolor{teal!50!yellow!50}1\frac{3}{8}$   & $\cellcolor{yellow!50}1\frac{1}{2}$ & $1\frac{1}{2}$ & $1\frac{1}{2}$ & $1\frac{1}{2}$ & $1\frac{1}{2}$ & $\cellcolor{yellow!50}1\frac{1}{2}$\\
$1$   & $\cellcolor{teal!50!yellow!50}1\frac{1}{2}$ & $\cellcolor{teal!50!yellow!50}1\frac{5}{8}$ & $\cellcolor{teal!50!yellow!50}1\frac{3}{4}$   & $\cellcolor{teal!50!yellow!50}1\frac{7}{8}$ & $\cellcolor{yellow!50}2$  & $2$ & $2$ & $2$ & $2$  & $\cellcolor{yellow!50}2$ \\\hline
	$\epsilon$ & $\frac{1}{2}$ & $\frac{5}{8}$ & $\frac{3}{4}$ & $\frac{7}{8}$ & 1 & & & & & \\
\end{tabular} 
}
\end{center}	
\end{minipage}
};
\end{tikzpicture}
\end{center}
The middle five columns (highlighted by green and yellow) show the complexities for our approach. The last row states the values of $\epsilon$ for which the complexities in the same columns are obtained. The rightmost five columns show the complexities for the lazy/eager approaches for $\gamma \in \{0,\frac{1}{4}, \frac{1}{2}, \frac{3}{4}\}$. They are all higher than for our approach, except for the last column for which $\gamma=1$: Regardless of $m$, the complexity gap is 
$\bigO{N^{\frac{1}{2}}}$ for $\gamma=0$,
$\bigO{N^{\frac{3}{8}}}$ for $\gamma=\frac{1}{4}$,
$\bigO{N^{\frac{1}{4}}}$ for $\gamma=\frac{1}{2}$, and
$\bigO{N^{\frac{1}{8}}}$ for $\gamma=\frac{3}{4}$
For $\gamma = 1$, our approach defaults to the eager approach and achieves the lowest complexities for $\epsilon=1$. \punto
\end{exa}

\subsection{Further Prior Work on Dynamic Query Evaluation}
We discuss further prior work on dynamic query evaluation.
Figure~\ref{fig:dynamic_taxonomy} gives a taxonomy of works in this field.

\begin{figure*}[t]
    \footnotesize
\begin{tabular}{@{\hskip 0in}l@{\hskip 0.05in}l@{\hskip 0.05in}l@{\hskip 0.05in}l@{\hskip 0.05in}l@{\hskip 0.05in}l@{\hskip 0in}}
\toprule
Class of Queries & Preprocessing & Update & Delay & Extra Space & Source \\
\midrule

$q$-hierarchical CQ  & 
$\bigO{N}$ & 
$\bigO{1}$ & 
$\bigO{1}$ &
-- & 
\cite{BerkholzKS17,Idris:dynamic:SIGMOD:2017} \\

Triangle count & 
$\bigO{N^{\frac{3}{2}}}$ & 
$\bigO{N^{\max\{\eps, 1-\eps \}}}^{\text{\Cross}}$ & 
$\bigO{1}$ &
$\bigO{N^{1 + \min\{\eps, 1-\eps\}}}$ & 
\cite{Kara:ICDT:19} \\

Full triangle query & 
$\bigO{N^{\frac{3}{2}}}$ & 
$\bigO{N^{\frac{1}{2}}}^{\text{\Cross}}$ & 
$\bigO{1}$ &
$\bigO{N^{\frac{3}{2}}}$ & 
\cite{Kara:TODS:2020} \\ 
 
$q$-hierarchical UCQ & 
$\bigO{N}$ & 
$\bigO{1}$ & 
$\bigO{1}$ &
 -- & 
 \cite{BerkholzKS18} \\

\midrule

FO+MOD on bound. degree  
 & 
$\bigO{N}$ & 
$\bigO{1}$ & 
$\bigO{1}$ &
 -- & 
 \cite{BerkholzKS17ICDT} \\  
 
MSO on Strings 
 & 
$\bigO{N}$ & 
$\bigO{\log N}$ & 
$\bigO{1}$ &
 -- & 
 \cite{NiewerthS18} \\  

\bottomrule
\end{tabular}

\caption{Prior work on the trade-off between preprocessing time, update time, enumeration delay, and extra space  for different classes of queries (Conjunctive Queries, Count Queries, First-Order Queries with modulo-counting quantifiers, Monadic Second Order Logic) and databases under updates in data complexity. Parameters: Query $Q$; database of size $N$; $\eps \in [0,1]$. Most works do not discuss the extra space utilization (marked by --). \Cross: amortized update time.}
\label{fig:dynamic_taxonomy}
\end{figure*}

The q-hierarchical queries are the conjunctive queries that admit linear-time preprocessing and constant-time update and delay~\cite{BerkholzKS17,Idris:dynamic:SIGMOD:2017}. 
If a conjunctive query without repeating relation symbols is not q-hierarchical, there is no $\gamma > 0$ such that 
the query result can be enumerated with $\bigO{N^{\frac{1}{2}-\gamma}}$ delay and update time, unless 
the Online Matrix Vector Multiplication conjecture fails.  
The constant delay and update time carry over to first-order queries with modulo-counting quantifiers on bounded degree databases, unions of q-hierarchical queries~\cite{BerkholzKS18}, and q-hierarchical queries  with small domain constraints~\cite{BerkholzKS17ICDT}.

Prior work characterizes the preprocessing-space-update trade-off for counting triangles under updates~\cite{Kara:ICDT:19}. 
A follow-up work generalizes this approach to the triangle queries with arbitrary free variables, adding the enumeration delay to the trade-off space~\cite{Kara:TODS:2020}. 
In this work, we consider arbitrary hierarchical queries instead of the triangle queries, and 
we use a less trivial adaptive maintenance technique, where the same relation may be subject to partition on different tuples of variables and where the overall number of cases for each partition is reduced to only two: the all-light case and the at-least-one-heavy case.

MSO queries on strings admit linear-time preprocessing, constant delay, and logarithmic update time. Here, updates can relabel, insert, or remove positions in the string. Further work considers MSO queries on trees under updates~{\cite{LosemannM14, AmarilliBM18}}.

DBToaster~\cite{DBT:VLDBJ:2014}, F-IVM~\cite{Nikolic:SIGMOD:18}, and DynYannakakis~\cite{Idris:dynamic:SIGMOD:2017, IdrisUVVL18} are recent systems implementing incremental view maintenance approaches. 

\section{Preliminaries}
\label{sec:preliminaries}

\paragraph{Data Model}
A schema $\calX = (X_1, \ldots, X_n)$ is a non-empty tuple of distinct variables. 
Each variable $X_i$ has a discrete domain $\Dom(X_i)$. 
We treat schemas and sets of variables interchangeably, assuming a fixed ordering of variables.
A tuple $\inst{x}$ of data values over schema $\mathcal{X}$ is an element from $\Dom(\calX) = \Dom(X_1) \times \dots \times \Dom(X_n)$.

A relation $R$  over schema $\mathcal{X}$ is a function 
$R: \Dom(\mathcal{X}) \to \mathbb{Z}$  such that   
the multiplicity $R(\inst{x})$ is non-zero for finitely many tuples $\inst{x}$. 
A tuple $\inst{x}$ is in $R$, denoted by $\inst{x} \in R$, if $R(\inst{x}) \neq 0$. 
The notation $\exists{R}$ denotes the use of $R$ with set semantics: $\exists{R}(\inst{x})$ equals $1$ if $\inst{x} \in R$ and $0$ otherwise; also, $\nexists{R}(\inst{x}) = 1 - \exists{R}(\inst{x})$.
The size $|R|$ of $R$ is the size of the set $\{ \inst{x} \mid \inst{x} \in R \}$. 
A database is a set of relations and has size given by the sum of the sizes of its relations. 

Given a tuple $\inst{x}$ over schema  $\mathcal{X}$ and  $\mathcal{S}\subseteq\mathcal{X}$, 
$\inst{x}[\mathcal{S}]$ denotes the restriction of 
$\inst{x}$ to $\calS$ such that the values in 
$\inst{x}[\mathcal{S}]$ follow the ordering in $\calS$.
For instance, $(a,b,c)[(C,A)] = (c,a)$ for the tuple $(a,b,c)$ over the schema $(A,B,C)$.
For a relation $R$ over $\mathcal{X}$, 
schema $\mathcal{S} \subseteq \calX$, 
and tuple $\inst{t} \in \Dom(\mathcal{S})$, 
$\sigma_{\mathcal{S} = \inst{t}} R = 
\{\, \inst{x} \,\mid\, \inst{x} \in R \land \inst{x}[\mathcal{S}] = \inst{t} \,\}$ denotes the set of tuples in $R$
that agree with $\inst{t}$ on the variables in $\calS$, while
$\pi_{\mathcal{S}}R = \{\, \inst{x}[\mathcal{S}] \,\mid\, \inst{x} \in R \,\}$ denotes the set of restrictions of the tuples in $R$ to the variables in $\mathcal{S}$.

\paragraph*{Computational Model.}
We consider the RAM model of computation
{where schemas and data values are of constant size.}
{We assume that each relation 
$R$  over schema $\mathcal{X}$ is implemented by a data structure that stores key-value entries $(\inst{x},R(\inst{x}))$ for each tuple $\inst{x}$ with $R(\inst{x}) \neq 0$ and needs $O(|R|)$ space.} 
This data structure can:
(1) look up, insert, and delete entries in constant time,
(2) enumerate all stored entries in $R$ with constant delay, and
(3) report $|R|$ in constant time.
For a schema $\mathcal{S} \subset \mathcal{X}$, 
we use an index data structure that for any $\tup{t} \in \Dom(\mathcal{S})$ can: 
(4) enumerate all tuples in $\sigma_{\mathcal{S}=\tup{t}}R$ with constant delay,
(5) check $\tup{t} \in \pi_{\mathcal{S}}R$ in constant time; 
(6) return $|\sigma_{\mathcal{S}=\tup{t}}R|$ in constant time; and
(7) insert and delete index entries in constant time.

{
In an idealized setting, 
the above requirements can be ensured using hashing.
In practice, hashing can only achieve {\em amortized} constant time 
for some of the above operations.
In our paper, whenever we claim constant time for hash operations, we mean amortized constant time.
We give a hash-based example data structure that supports the above operations in amortized constant time.}
Consider a relation $R$ over schema $\mathcal{X}$. 
A hash table with chaining stores key-value entries $(\inst{x},R(\inst{x}))$ for each tuple $\inst{x}$ over $\mathcal{X}$ with $R(\inst{x}) \neq 0$. 
The entries are doubly linked to support enumeration with constant delay. 
The hash table can report the number of its entries in constant time and supports lookups, inserts, and deletes in {amortized} constant time.
%
To support index operations on a schema $\mathcal{F} \subset \mathcal{X}$, 
we create another hash table with chaining where each table entry stores an $\calF$-value 
$\tup{t}$ as key and a doubly-linked list of pointers to the entries in $R$ having 
$\tup{t}$ as $\mathcal{F}$-value.
Looking up an index entry given $\tup{t}$ takes {amortized} constant time,
and its doubly-linked list enables enumeration of the matching entries in $R$ with constant delay. 
Inserting an index entry into the hash table additionally prepends a new pointer to the doubly-linked list for a given $\tup{t}$; overall, this operation takes {amortized} constant time.
For efficient deletion of index entries, each entry in $R$ also stores back-pointers to its index entries (one back-pointer per index for $R$). 
When an entry is deleted from $R$, locating and deleting its index entries in doubly-linked lists takes constant time per index.
{An alternative data structure that can meet our requirements is a tree-structured index
such as a B$^+$-tree. This would, however, require {\em worst-case} logarithmic time and imply an additional logarithmic factor in our complexity results. 
}

\paragraph*{Modeling Updates Using Multiplicities.}
We restrict multiplicities of tuples in the {input relations to be strictly positive}. Multiplicity 0 means the tuple is not present. A single-tuple update to a relation $R$ is expressed as $\delta R = \{\tup{x} \rightarrow m\}$. 
The update is an insert of the tuple $\tup{x}$ in $R$ if the multiplicity $m$ is strictly positive. It is a delete of $\tup{x}$ from $R$ if $m$ is negative. Such a delete is rejected if the existing multiplicity of $\tup{x}$ in $R$ is less than $|m|$.
{
  A batch update may consist of both inserts and deletes. 
  Applying $\delta{R}$ to $R$ means creating a new version of $R$ that is the union 
  of  $\delta{R}$ and $R$.}

\paragraph{Partitioning.}
{The number of occurrences of a value in a relation is called the degree of the value in the relation.} 
We partition relations based on value degree. 

\begin{defi}
\label{def:partition}
Given a relation $R$ over schema $\mathcal{X}$, a schema $\calS \subset\mathcal{X}$, and a threshold $\theta$,
the pair $(H, L)$ of relations is a {\em partition} of $R$ on $\calS$ with threshold $\theta$ 
if it satisfies the following four conditions:
\\[4pt]
\begin{tabular}{@{\hskip 0in}rl}
{(union)} & $R(\tup{x}) = H(\tup{x}) + L(\tup{x})$ for $\tup{x} \in \Dom(\calX)$ \\[4pt]
{(domain partition)} & 
$\pi_{\calS}H \cap \pi_{\calS}L = \emptyset$ \\[4pt]
(heavy part) & for all $\tup{t} \in \pi_{\calS} H$:
 $|\sigma_{\calS= \tup{t}} H| \geq \frac{1}{2}\theta$\\[4pt]
(light part) & for all $\tup{t} \in \pi_{\calS} L$:
 $|\sigma_{\calS=\tup{t}} L| < \frac{3}{2}\theta$
\end{tabular}\\[6pt]
The pair $(H, L)$ is a {\em strict} partition of $R$ on $\calS$ with threshold 
$\theta$ if it satisfies the union and 
domain partition conditions and strict versions
of the heavy and light part conditions: 
\\[6pt]
\begin{tabular}{@{\hskip 0.2in}rl}
{(strict heavy part)} & for all $\tup{t} \in \pi_{\calS}H:\; |\sigma_{\calS=\tup{t}} H| \geq 
\theta$ \\[4pt]
{(strict light part)} & for all $\tup{t} \in \pi_{\calS}L:\; |\sigma_{\calS=\tup{t}} L| < \theta$
\end{tabular}\\[6pt]
The relations $H$ and $L$ are the {\em heavy} and {\em light} parts of $R$.
\end{defi}

Assuming $|R|=N$ and the strict partition $(H,L)$ of $R$ on $\mathcal{S}$ with threshold $\theta=N^\eps$ for $\eps\in[0,1]$, we have:
$\forall\tup{t} \in \pi_{\calS} L: |\sigma_{\calS=\tup{t}} L| < \theta = N^\eps$;
 and $|\pi_{\calS} H| \leq \frac{|R|}{\theta} = N^{1-\eps}$.     
We subsequently denote the light part of $R$ on $\calS$ by $R^{\calS}$.

\paragraph{Queries.} 
A conjunctive query (CQ) has the form 
$$Q(\calF) = R_1(\calX_1), \ldots, R_n(\calX_n).$$
We denote by: 
$(R_i)_{i\in[n]}$ the relation symbols;
$(R_i(\calX_i))_{i\in[n]}$ the atoms;
$\vars(Q) = \bigcup_{i\in[n]}\calX_i$ the set of variables;
{$\free(Q)=\calF \subseteq \vars(Q)$} the set of {\em free} variables; 
$\bound(Q)=\vars(Q)-\free(Q)$ the set of {\em bound} variables;
$\atoms(Q)=\{R_i(\calX_i) \mid i\in[n]\}$ the 
set of the atoms;
and $\atoms(X)$ the set of the atoms containing $X$.
The query $Q$ is {\em full} if $\free(Q)=\vars(Q)$. 

The hypergraph $G=(\vars(Q),\atoms(Q))$ of a query $Q$ has one node per variable and one hyperedge per atom that covers all nodes representing its variables. 
{A  {\em join tree} for $Q$ is a tree with the following properties:
 (1) Its nodes are exactly the atoms of $Q$; (2) if any two nodes have variables in common, then all nodes along the path between them also have these variables. The query $Q$ is called {\em $\alpha$-acyclic} if it has a join tree.}
{It is {\em free-connex} if it is $\alpha$-acyclic and remains $\alpha$-acyclic when we add to its body a fresh atom over its free variables~\cite{BraultPhD13}}.
It is {\em hierarchical} if for any two of its variables, either their sets of atoms are disjoint or one is contained in the other. 
It is {\em q-hierarchical} if it is hierarchical and for every variable $A\in\free(Q)$, if there is a variable $B$ such that $\atoms(A)\subset\atoms(B)$ then $B\in\free(Q)$~\cite{BerkholzKS17}.

\begin{exa}\label{ex:query-prelim}
\rm
The following query is $\alpha$-acyclic:
$$Q(A,C,F)=R(A,B,C),S(A,B,D),T(A,E,F),U(A,E,G)$$ 
A join tree is the path $U(AEG)-T(AEF)-R(ABC)-S(ABD)$.
It is free-connex since we can extend this join tree as follows: $U(AEG)-T(AEF)-Q(ACF)-R(ABC)-S(ABD)$.
It is also hierarchical but not q-hierarchical: The bound variables $B$ and $E$ dominate the free variables $C$ and respectively $F$.
\qed
\end{exa}

\paragraph{Variable Orders.}
{Two variables  {\em depend} on each other if they occur in the same atom.}

\begin{defi}[adapted from \cite{OlteanuZ15}]
A {\em variable order} $\omega$ for a conjunctive query 
$Q$ is a pair
$(T, \dep_\omega)$ such that the following holds:  
\begin{itemize}[leftmargin=*]
\item {$T$ is a rooted forest with one node per variable in $Q$.}
The variables of each atom in $Q$ lie along the same root-to-leaf path in $T$.

\item The function $\dep_\omega$ maps each variable 
$X$ to the subset of its ancestor variables in $T$
on which the variables in the subtree rooted at $X$
depend, i.e., for every variable $Y$ that is a child 
of variable $X$, $\dep_\omega(Y) \subseteq \dep_\omega(X) \cup \{X\}$.
\end{itemize}  
\end{defi}

{
An {\em extended} variable order is a variable order where we add as new leaves the atoms corresponding to relations. 
We add each atom as the child of its variable placed lowest in the variable order. 
Whenever we refer to a variable order, we mean its extension with atoms at leaves. 
For ease of presentation, we often use $\omega$ to refer to the tree of $\omega$.}

{
The subtree of a variable order $\omega$ rooted at $X$ is denoted by $\omega_X$.}
The sets $\vars(\omega)$, $\atoms(\omega)$, and $\anc(X)$ consist of 
all variables of $\omega$, the atoms at the leaves of $\omega$, and  
the variables on the path from $X$ to the root excluding $X$, respectively. 
The flag $\sibling(X)$ is true if $X$ has siblings.
The variable order $\omega$ is {\em free-top} if no bound variable is an ancestor of a free variable 
(called d-tree extension~\cite{OlteanuZ15}).
It is {\em canonical} if the variables of the leaf atom of each root-to-leaf path are the inner nodes of the path.
The sets $\freeTopVO(Q)$, $\canonicalVO(Q)$, and $\VO(Q)$ consist of free-top, canonical, and all variable orders of $Q$.

{
\begin{figure}[t]
  \small
  \centering
  \begin{minipage}[b]{0.49\linewidth}
    \centering
    \begin{tikzpicture}[xscale=1.3, yscale=0.8]
      \node at (0.4, 0.0) (A) {{\color{black} $\underline{A}$}};
      \node at (-1, -1.0) (B) {{\color{black} $B$}} edge[-] (A);
      \node at (1.8, -1.0) (E) {{\color{black} $E$}} edge[-] (A);
      \node at (-1.7, -2.0) (C) {{\color{black} $\underline{C}$}} edge[-] (B);
      \node at (-0.3, -2.0) (D) {{\color{black} $D$}} edge[-] (B);
      \node at (1.1, -2.0) (F) {{\color{black} $\underline{F}$}} edge[-] (E);
            \node at (2.5, -2.0) (G) {{\color{black} $G$}} edge[-] (E);
      \node at (-0.3, -3.0) (S) {{\color{black} $S(\underline{A},B,D)$}} edge[-] (D);
      \node at (-1.7, -3.0) (R) {{\color{black} $R(\underline{A},B,\underline{C})$}} edge[-] (C);
          \node at (1.1, -3.0) (T) {{\color{black} $T(\underline{A},E, \underline{F})$}} edge[-] (F);
            \node at (2.5, -3.0) (U) {{\color{black} $U(\underline{A},E,G)$}} edge[-] (G);
    \end{tikzpicture}
  \end{minipage}
  \hspace{0.05cm}
  \begin{minipage}[b]{0.49\linewidth}
    \centering
    \begin{tikzpicture}[xscale=1.3, yscale=0.8]
      \node at (0.4, 0.0) (A) {{\color{black} $\underline{A}$}};
      \node at (-1, -0.6) (C) {{\color{black} $\underline{C}$}} edge[-] (A);
      \node at (1.8, -0.6) (F) {{\color{black} $\underline{F}$}} edge[-] (A);
      \node at (-1, -1.5) (B) {{\color{black} $B$}} edge[-] (C);
      \node at (1.8, -1.5) (E) {{\color{black} $E$}} edge[-] (F);
      \node at (-0.3, -2.0) (D) {{\color{black} $D$}} edge[-] (B);
            \node at (2.5, -2.0) (G) {{\color{black} $G$}} edge[-] (E);
      \node at (-0.3, -3.0) (S) {{\color{black} $S(\underline{A},B,D)$}} edge[-] (D);
      \node at (-1.6, -3.0) (R) {{\color{black} $R(\underline{A},B,\underline{C})$}} edge[-] (B);
          \node at (1.1, -3.0) (T) {{\color{black} $T(\underline{A},E, \underline{F})$}} edge[-] (E);
            \node at (2.5, -3.0) (U) {{\color{black} $U(\underline{A},E,G)$}} edge[-] (G);
    \end{tikzpicture}
  \end{minipage}
  \caption{{Canonical (left) and non-canonical but free-top (right) variable order for the query  
  $Q(A,C,F)=$ $R(A,B,C),$ $S(A,B,D),$ $T(A,E,F),$ $U(A,E,G)$ in Example~\ref{ex:query-prelim}.
  Free variables are underlined.}} 
  \label{fig:canonical_free_top}
\end{figure}
}

{
\begin{exa}
The left variable order in Figure~\ref{fig:canonical_free_top} is a canonical 
variable order for the query from Example~\ref{ex:query-prelim}.
This variable order is not free-top since the bound variables $B$ and $E$ sit on top of the free variables $C$ and respectively $F$. 
The right variable order in Figure~\ref{fig:canonical_free_top} is a free-top variable order 
for the query. 
This variable order is not canonical: the atom at the leaf of the path $A-C-B-D-S(ABD)$ does not have the variable $C$.\qed
\end{exa}
}

Hierarchical queries admit canonical variable orders, while q-hierarchical queries admit canonical free-top variable orders.
The canonical variable order of a hierarchical query is unique 
up to orderings of variables sharing the same set of atoms.

\paragraph*{Width Measures.}
Given a conjunctive query $Q$ and $\calF \subseteq \vars(Q)$,    
a {\em fractional edge cover}
of $\calF$ is a solution 
$\boldsymbol{\lambda} = (\lambda_{R(\calX)})_{R(\calX) \in \atoms(Q)}$ to the following 
linear program \cite{AtseriasGM13}: 
\begin{align*}
\text{minimize} & \TAB\sum_{R(\calX) \in\, \atoms(Q)} \lambda_{R(\calX)} && \\[3pt]
\text{subject to} &\TAB \sum_{{R(\calX)\in\, \atoms(Q) \text{ s.t. } X \in \calX}}\hspace{-0.5cm} \lambda_{R(\calX)} \geq 1 && \text{ for all } X \in \calF \text{ and } \\[3pt]
& \TAB\lambda_{R(\calX)} \in [0,1] && \text{ for all } R(\calX) \in \atoms(Q)
\end{align*}
The optimal objective value of the above program 
is called the {\em fractional edge cover number} of the variable set $\calF$ 
and is denoted as $\rho_{Q}^{\ast}(\calF)$.  
An {\it integral edge cover} of $\calF$ is a feasible solution 
to the variant of the above program with 
$\lambda_{R(\calX)}\in\{0,1\}$ for each $R(\calX) \in \atoms(Q)$.
The optimal objective value of this program 
is called the {\em integral edge cover number} of $\calF$
and is denoted as $\rho_{Q}(\calF)$.
If $Q$ is clear from the context, we omit 
the index $Q$ in the expressions $\rho_{Q}^{\ast}(\calF)$
and $\rho_{Q}(\calF)$.
For a database of size $N$, the result of the query $Q$ can be computed in time 
$\bigO{N^{\rho^*}}$~\cite{Ngo:JACM:18}.

For hierarchical queries, the integral and fractional edge cover 
numbers are equal.
{The proofs of the following propositions in this section are given in the technical report~\cite{Trade_Offs_arxiv} (Appendices B and C).}

\begin{prop}
\label{prop:rho_rhostar}
For any hierarchical query $Q$ and $\calF \subseteq \vars(Q)$, it holds
$\rho^{\ast}(\calF) = \rho(\calF)$.
\end{prop}

\begin{defi}
\label{def:fac_width}
The {\em static width} of a conjunctive query $Q$ is
\begin{align*}
\fw(Q) &= \min_{\omega\in\freeTopVO(Q)} \fw(\omega)\\
\fw(\omega) &= \max_{X \in \vars(Q)} \rho^*(\{X\} \cup \dep_\omega(X))
\end{align*}
\end{defi}
If $Q$ is Boolean, then $\fw$ is the {\em fractional hypertree width}~\cite{Marx:TRANSALG:10}. {\em FAQ-width} generalizes $\fw$ to queries over several semirings~\cite{FAQ:PODS:2016}
\footnote{To simplify presentation, we focus on queries that contain 
at least one atom with non-empty schema. This implies that the static width 
of queries is at least $1$. Queries where all atoms 
have empty schemas
obviously admit constant preprocessing time,
update time, and enumeration delay.}.

\begin{defi}
\label{def:delta_width}
The {\em dynamic width} of a conjunctive query $Q$ is
\begin{align*}
\dfw(Q) &= \min_{\omega\in\freeTopVO(Q)} \dfw(\omega)\\
\dfw(\omega) &= 
\max_{X \in \vars(Q)} \
\max_{R(\calY) \in \atoms(\omega_X)}
 \rho^*((\{X\} \cup \dep_{\omega}(X)) - \calY)
\end{align*}
\end{defi}
While the static width of a free-top variable order 
$\omega$ is defined over the sets 
$\{X\} \cup \dep_\omega(X)$ with $X \in \vars(Q)$, 
the dynamic width 
of $\omega$  is defined over 
restrictions of these sets obtained 
by dropping the variables in the schema of one atom. 
For any canonical variable order $\omega$, variable $X$
in $\omega$, and atom $R(\calY)$ in $\atoms(\omega_X)$,
the set $(\{X\} \cup \dep_{\omega}(X)) - \calY$ is empty.
Hence, queries that admit {\em canonical} free-top variable orders have dynamic width $0$.

\begin{prop}
  \label{prop:width_delta_inequal}
  Given a hierarchical query with static width 
  $\fw$ and dynamic width $\dfw$, it holds that
  $\dfw=\fw$ or $\dfw=\fw-1$.
  \end{prop}
  
Free-connex hierarchical queries have static width 1.

\begin{prop}
  \label{prop:free-connex_static_one}
  Any free-connex hierarchical query has static width 1.
  \end{prop} 

We give a syntactic classification of hierarchical queries based on their dynamic width.

\begin{defi}
\label{def:delta_hierarchical}
A hierarchical query is  {\it $\dfw_i$-hierarchical}  for $i \in {\mathbb{N}}$ if $i$ is the smallest number such that for each bound variable $X$ and atom $R(\calY)$ of $X$, there are $i$ atoms $R_1(\calY_1),  \ldots ,R_i(\calY_i)$ such that all free variables in the atoms of $X$ are included in $\calY \cup \bigcup_{j \in [i]}\calY_j$.
\end{defi}

For instance, the query $Q(Y_0, \ldots , Y_i) = R_0(X,Y_0), \ldots, R_i(X,Y_i)$ is
a $\dfw_i$-hierarchical query for $i \in  {\mathbb{N}}$.
The class of hierarchical queries can be partitioned into subclasses of
 $\dfw_i$-hierarchical queries
for $i \in {\mathbb{N}}$. Then, the $\dfw_0$-hierarchical queries are precisely the q-hierarchical
queries from prior work~\cite{BerkholzKS17}. 

\begin{prop}
\label{prop:delta_0_q_hierarchical}
A query is q-hierarchical if and only if 
it is $\delta_0$-hierarchical.
\end{prop}

As depicted in Figure~\ref{fig:venn_diagrams} (right), all free-connex hierarchical queries  
are either $\dfw_0$- or $\dfw_1$-hierarchical.

\begin{prop}
\label{prop:free-connex_dynamic_0_1}
Any free-connex hierarchical  query is 
$\dfw_0$- or $\dfw_1$-hierarchical.
\end{prop}

The following proposition relates $\delta_i$-hierarchical queries 
to their dynamic width.

\begin{prop}
\label{prop:delta_hierarchical_dynamic_width}
A hierarchical query is 
$\dfw_i$-hierarchical for $i \in {\mathbb{N}}$ if and only if 
it has dynamic width $i$.
\end{prop} 

Proposition~\ref{prop:delta_hierarchical_dynamic_width} and Theorem~\ref{thm:main_dynamic}
imply the following corollary.

\begin{cor}
\label{theo:delta_hierarchical_queries}

Given a $\dfw_i$-hierarchical query with $i \in {\mathbb{N}}$ and static width $\fw$, a database of size $N$, and $\eps \in [0,1]$, the query result can be enumerated with $\bigO{N^{1-\eps}}$ delay after $\bigO{N^{1 + (\fw -1)\eps}}$ preprocessing time and $\bigO{N^{i\eps}}$ amortized time for single-tuple updates.
\end{cor}


\section{Preprocessing}
\label{sec:preprocessing}
In the preprocessing stage, we construct a data structure that represents the result of a given hierarchical query. 
{The data structure consists of a set of view trees, where each view tree computes one part of the query result.
A view tree is a tree-shaped hierarchy of materialized views with input relations as leaves and upper views defined in terms of their child views.}
The construction of view trees exploits the structure of the query and the degree of data values in base relations. We construct different sets of view trees for the static and dynamic evaluation of a given hierarchical query. 

We next assume that the canonical variable order of the given hierarchical query consists of a single connected component. For several connected components, the preprocessing procedure is executed on each connected component separately. 

\subsection{View Trees Encoding the Query Result}
\label{sec:preprocessing-factorization}

Given a hierarchical query $Q(\mathcal{F})$ and a canonical variable order $\omega$ for $Q$, the function \textsc{BuildVT} in Figure~\ref{fig:factorized_view_tree_algo} constructs a view tree that encodes the query result. 
The function proceeds recursively on the structure of $\omega$ and constructs a view {over schema $\mathcal{F}_X$ at each inner node $X$; the leaves correspond to the atoms in the query. 
The view is defined over the join of its child views projected onto $\mathcal{F}_X$ (Figure~\ref{fig:view_creation}). The schema $\mathcal{F}_X$ includes the ancestors of $X$ in $\omega$ since they are needed for joins at nodes above $X$. Each constructed view has a name to help us identify the place and purpose of the view in the view tree.}

If $X$ is free, then it is included in the schema of the view constructed at $X$ (and not included if bound). 
It is also kept in the schemas of the views on the path to the root until it reaches a view whose schema does not have bound variables.
{
  The constructed view tree has the upper levels only with views over the free variables. The hierarchy of such views represents the query result and allows its enumeration with constant delay.
}

In the dynamic case, at each child $Z$ of $X$ we construct a view with schema $\anc(Z)$ on top of the view created at $Z$ (Figure~\ref{fig:aux_view_creation}). This auxiliary view aggregates away $Z$ from the latter view. 
The children of the view created at $X$ then share the same schema $\mathcal{F}_X$. 
This property enables the efficient maintenance of the view at $X$ since processing a change coming from any child view requires only constant-time lookups into that child's sibling views.

Our preprocessing is particularly efficient for free-connex hierarchical queries in the static case and for their strict subclass of $\dfw_0$-hierarchical queries in the dynamic case.

\begin{figure}[t]
\centering
\setlength{\tabcolsep}{3pt}
\renewcommand{\linenumber}{\makebox[2ex][r]{\rownumber\TAB}}
\setcounter{magicrownumbers}{0}
\begin{tabular}[t]{@{}c@{}c@{}l@{}}
    \toprule
    \multicolumn{3}{l}{\textsc{BuildVT}({\text{string} $\mathit{prefix}$}, \text{variable order} $\omega$, {schema} $\mathcal{F}$) : view tree} \\
    \midrule
    \multicolumn{3}{l}{\MATCH $\omega$:}\\
    \midrule
    \phantom{ab} & $R(\calY)$\hspace*{2.5em} & \linenumber \RETURN $R(\calY)$\\[2pt]
    \cmidrule{2-3} \\[-6pt]
    &    
    \begin{minipage}[t]{1.5cm}
        \vspace{-0.32cm}
        \hspace*{-0.55cm}
        \begin{tikzpicture}[xscale=0.5, yscale=1]
            \node at (0,-2)  (n4) {$X$};
            \node at (-1,-3)  (n1) {$\omega_1$} edge[-] (n4);
            \node at (0,-3)  (n2) {$\ldots$};
            \node at (1,-3)  (n3) {$\omega_k$} edge[-] (n4);
        \end{tikzpicture}
    \end{minipage}
    &
    \begin{minipage}[t]{8.8cm}
        \vspace{-0.4cm}
        \linenumber \LET $T_i = \textsc{BuildVT}(V, \omega_i, \mathcal{F}), \forall i\in[k]$ \\[0.5ex]
        \linenumber \LET ${\mathit{viewname} = \mathit{prefix} + \text{``\_''} + X.\mathit{name}}$ \\[0.5ex]
        \linenumber \IF $(\anc(X) \cup \{X\}) \subseteq \mathcal{F}$ \quad\quad  \\[0.5ex]
        \linenumber \TAB \LET $\mathcal{F}_X = \anc(X) \cup \{X\}$ \\[0.5ex]
        \linenumber \TAB \LET $subtrees = \{\, \textsc{AuxView}(\text{root of } \omega_i, T_i) \,\}_{i\in[k]}$ \\[0.5ex]
        \linenumber \TAB \RETURN $\textsc{NewVT}({\mathit{viewname}}, \mathcal{F}_X, subtrees)$ \\[0.5ex]
        \linenumber \LET $\mathcal{F}_X = \anc(X) \cup (\mathcal{F} \cap \vars(\omega))$ \\[0.5ex]
        \linenumber \LET $\mathit{subtrees} = \{ T_i \}_{i\in[k]}$ \\[0.5ex]
        \linenumber \RETURN $\textsc{NewVT}({\mathit{viewname}}, \mathcal{F}_X, \mathit{subtrees})$
    \end{minipage}\\[2.75ex]
    \bottomrule
\end{tabular}\vspace{-0.1em}
\caption{Construction of a view tree for a canonical variable order $\omega$ of a hierarchical query with free variables $\mathcal{F}$. {View names share a given $\mathit{prefix}$.}}
\label{fig:factorized_view_tree_algo}
\end{figure}

\begin{figure}[t]
\centering
\setlength{\tabcolsep}{3pt}
\renewcommand{\linenumber}{{\rownumber\TAB}} 
\setcounter{magicrownumbers}{0}
\begin{tabular}[t]{l}
    \toprule
    \textsc{NewVT}({string $\mathit{viewname}$}, schema $\calS$, view trees $T_1, \ldots, T_k$) : view tree \\
    \midrule
    \linenumber \LET $V_i(\calS_i)$ = root of $T_i, \forall i\in[k]$ \\[0.5ex]
    \linenumber \LET $V(\calS) = {\text{ join of } V_1(\calS_1),\ldots,V_k(\calS_k) \text{ projected onto } \calS}$ \\[0.5ex]
    \linenumber {$V.\mathit{name} := viewname$} \\[0.5ex]
    \linenumber \RETURN $\left\{
      \begin{array}{@{~~}c@{~~}l@{~~}}
        T_1 &\text{,\; $k = 1 \land \calS = \calS_1$}\\
        \raisebox{-4 ex}{\tikz {
            \node at (0,-1)  (n4) {$V(\calS)$};
            \node at (-0.5,-1.75)  (n1) {$T_1$} edge[-] (n4);
            \node at (0,-1.75)  (n2) {$\ldots$};
            \node at (0.5,-1.75)  (n3) {$T_k$} edge[-] (n4);
        }} & \text{,\; otherwise}
      \end{array}\right.$ \\[0.5ex]
    \bottomrule
\end{tabular}
\vspace{-0.1em}
\caption{Construction of a view tree with a {given root view name}, root schema $\mathcal{S}$, and children $T_1, \ldots, T_k$.}
\label{fig:view_creation}
\end{figure}

\begin{figure}[t]
\centering
\setlength{\tabcolsep}{3pt}
\renewcommand{\linenumber}{{\rownumber\TAB}}
\setcounter{magicrownumbers}{0}
\begin{tabular}[t]{l}
    \toprule
    \textsc{AuxView}(node $Z$, view tree $T$) : view tree \phantom{\hspace{2.65cm}} \\
    \midrule
    \linenumber \LET $V(\calS) = \text{ root of } T$ \\[0.5ex]
    \linenumber {\LET $\mathit{viewname} = V.\mathit{name} \,+\, \text{``\,'\,''}$} \\[0.5ex]
    \linenumber \IF $mode = \text{`dynamic'} \land \sibling(Z) \land \anc(Z) \subset \calS$  \\[0.5ex]
    \linenumber \TAB \RETURN $\textsc{NewVT}({\mathit{viewname}}, \anc(Z), \{T\})$ \\[0.5ex]
    \linenumber \RETURN $T$\\[0.5ex]
    \bottomrule
\end{tabular}\vspace{-0.2em}
\caption{A tree $T$ constructed at variable $Z$ is extended with a new root view that aggregates away $Z$.}
\label{fig:aux_view_creation}
\vspace*{-0.2em}
\end{figure}

For a canonical variable order of a hierarchical query, the free-connex property fails if there are free variables such that they are below a bound join variable and are not covered by one atom.
Indeed, assume two branches out of a bound join variable $X$ and with free variables $Y$ and respectively $Z$. Then, there are two atoms in $Q$ whose sets of variables include $\{X,Y\}$ and respectively $\{X,Z\}$, while $\{Y,Z\}$ are included in the head atom of $Q$. This creates a cycle in the hypergraph of $Q$, which means that $Q$ is not free-connex. 

For $\dfw_0$-hierarchical queries, there is no bound variable whose set of atoms strictly contains the atoms of a free variable. Such queries thus admit canonical free-top variable orders where all free variables occur above the bound ones.

For any free-connex hierarchical query, each view created by \textsc{BuildVT} is defined over variables from one atom of the query and can be materialized in linear time. We can thus recover the linear-time preprocessing for such queries used for static~\cite{BaganDG07} and dynamic~\cite{BerkholzKS17,Idris:dynamic:SIGMOD:2017} evaluation.

\begin{exa}\label{ex:free-connex-view-tree}
\rm
Consider the free-connex query 
$$Q(A,D,E)=R(A,B,C),S(A,B,D),T(A,E)$$
 and its canonical variable order in Figure~\ref{fig:free-connex-view-tree}.
We construct the view tree bottom-up as follows. At $C$, we create the view $V_C(A,B)$ that aggregates away the bound variable $C$ but keeps its ancestors $A$ and $B$ to define views up in the tree. 
{Since $D$ is free and has only one child,
we skip creating a view at $D$; see the first case in Line 4 of \textsc{NewVT} from
Figure~\ref{fig:view_creation}.} Similarly, no view is created at $E$. At $B$, we create the view $V_B(A,D) = V_C(A,B), S(A,B,D)$, which keeps $D$ as it is free and $A$ as the ancestor of $B$. 
At $A$, we create the views $V_A(A) = V_B(A,D), T(A,E)$ in the static case 
and $V_A(A) = V'_B(A), T'(A)$ in the dynamic case, where $V'_B(A) = V_B(A,D)$ and $T'(A) = T(A,E)$.
Each view can be computed in linear time by aggregating away variables and semi-join reduction.
The result of $Q$ can be enumerated using $V_A(A)$, $V_B(A,D)$, and $T(A,E)$ with constant delay.\qed
\end{exa}

\begin{figure}[t]
  \small
  \centering
  \begin{minipage}[b]{0.49\linewidth}
    \centering
    \begin{tikzpicture}[xscale=1.4, yscale=0.8]
      \node at (0.25, 0.0) (A) {{\color{black} $\underline{A}$}};
      \node at (-0.5, -1.0) (B) {{\color{black} $B$}} edge[-] (A);
      \node at (1.5, -1.0) (E) {{\color{black} $\underline{E}$}} edge[-] (A);
      \node at (1.5, -3.0) (T) {{\color{black} $T(\underline{A},\underline{E})$}} edge[-] (E);
      \node at (-1.3, -2.0) (C) {{\color{black} $C$}} edge[-] (B);
      \node at (0.3, -2.0) (D) {{\color{black} $\underline{D}$}} edge[-] (B);
      \node at (0.3, -3.0) (S) {{\color{black} $S(\underline{A},B,\underline{D})$}} edge[-] (D);
      \node at (-1.3, -3.0) (R) {{\color{black} $R(\underline{A},B,C)$}} edge[-] (C);
    \end{tikzpicture}
  \end{minipage}
  \begin{minipage}[b]{0.49\linewidth}
    \small
    \centering
    \begin{tikzpicture}[xscale=1.4, yscale=0.8]
      \node at (0.25, 0.0) (A) {{\color{black} $V_A(\underline{A})$}};
      \node at (-1, -0.75) (Bp) {{\color{black} $V'_B(\underline{A})$}} edge[-, shorten >=-0.05cm, shorten <=-0.05cm] (A);
      \node at (-1, -1.5) (B) {{\color{black} $V_B(\underline{A},\underline{D})$}} edge[-, shorten >=-0.05cm, shorten <=-0.05cm] (Bp);
      \node at (1.5, -0.75) (Tp) {{\color{black} $T'(\underline{A})$}} edge[-, shorten >=-0.05cm, shorten <=-0.05cm] (A);
      \node at (1.5, -3.0) (T) {{\color{black} $T(\underline{A},\underline{E})$}} edge[-] (Tp);
      \node at (-2, -2.25) (C) {{\color{black} $V_C(\underline{A},B)$}} edge[-, shorten >=-0.05cm, shorten <=-0.05cm] (B);
      \node at (0.0, -3.0) (S) {{\color{black} $S(\underline{A},B,\underline{D})$}} edge[-] (B);
      \node at (-2, -3.0) (R) {{\color{black} $R(\underline{A},B,C)$}} edge[-, shorten >=-0.05cm] (C);

      \draw[black, dashed] (-1.5,-0.43) rectangle (-0.5,-1.07);
      \draw[black, dashed] (1,-0.43) rectangle (2,-1.07);
    \end{tikzpicture}
  \end{minipage}
  \caption{Canonical variable order and view tree for $Q(A,D,E) = R(A,B,C), S(A,B,D), T(A,E)$ in Example~\ref{ex:free-connex-view-tree}. The views $V'_B$ and $T'$ are created in the dynamic case. Free variables are underlined.} 
  \label{fig:free-connex-view-tree}
\end{figure}

\subsection{Skew-Aware View Trees}
\label{sec:skew_aware_trees}
For free-connex queries, the procedure \textsc{BuildVT} constructs in linear time a data structure  
that allows for constant-time enumeration delay (Proposition~\ref{prop:enumeration} and Lemma~\ref{lem:fact_view_tree}). 
For $\dfw_0$-hierarchical queries, it also admits constant-time updates (Lemma~\ref{lem:fact_view_tree_update}).
We now focus on the bound join variables that violate the free-connex property in the static case or the $\dfw_0$-hierarchical property in the dynamic case. For each such violating bound variable $X$, we use two evaluation strategies.
 
The first strategy materializes a subset of the query result obtained for the {\it light} values over the set of variables $\anc(X)\cup\{X\}$ in the variable order. It also aggregates away the bound variables in the subtree rooted at $X$. Since the light values have a bounded degree, this materialization is inexpensive.

The second strategy computes a compact representation of the rest of the query result obtained for those values over $\anc(X)\cup\{X\}$ that are {\it heavy} (i.e., have high degree) in at least one relation.
This second strategy treats $X$ as a free variable and proceeds recursively to resolve further bound variables located below $X$ in the variable order and to potentially fork into more strategies.

The union of these strategies precisely cover the entire query result, yet not necessarily disjointly. 
To enumerate the distinct tuples in the query result, we then use an adaptation of the union algorithm~\cite{Durand:CSL:11} where the delay is given by the number of heavy values of the variables we partitioned on and by the number of strategies. 

\paragraph{Heavy and Light Indicators.}
{We consider a bound join variables that violates the free-connex property in the static case or the $\dfw_0$-hierarchical 
property in the dynamic case.}
We compute heavy and light indicator views consisting of disjoint sets of values for each such variable $X$. 
The heavy indicator has the values that exist in all relations and are heavy in at least one relation. 
The light indicator has the values that exist in all relations and are light in all relations.
Indicator views have set semantics. They allow us to rewrite the query into an equivalent union of two queries.

Partitioning the query result only based on the degree of $X$-values  may blow up the enumeration delay: 
the path from $X$ to the root may contain several bound join variables, each creating buckets of values per bucket of their ancestors, thus leading to an explosion of the number of buckets that need to be unioned together during enumeration. 
However, one remarkable property holds for hierarchical queries: each base relation located in the subtree rooted at $X$ contains $X$ but also all the ancestors of $X$.
Thus, by partitioning each relation jointly on $X$ and its ancestors, we can ensure the enumeration delay remains linear in the number of distinct heavy values over $\anc(X)\cup\{X\}$.

\begin{figure}[t]
\centering
\setlength{\tabcolsep}{3pt}
\renewcommand{\linenumber}{{\rownumber\TAB}}
\renewcommand{\arraystretch}{1.05}
\setcounter{magicrownumbers}{0}
\begin{tabular}[t]{l}
    \toprule
    \textsc{IndicatorVTs}(variable order $\omega$) : triple of view trees \\
    \midrule
    \linenumber \LET $X = \text{root of } \omega$ \\[0.2ex]
    \linenumber \LET $keys = \anc(X) \cup \{X\}$ \\[0.2ex]
    \linenumber \LET $alltree = \textsc{BuildVT}({\text{``All\,''}}, \omega, keys)$ \\[0.2ex]
    \linenumber \LET $ltree = \textsc{BuildVT}({\text{``\,L\,''}}, \omega^{keys}, keys)$ \\[0.2ex]
    \linenumber \LET $allroot = \text{root of } alltree$ \\[0.2ex]
    \linenumber \LET $lroot = \text{root of } ltree$ \\[0.2ex]
    \linenumber \LET $htree = \textsc{NewVT}({\text{``\,H\_\,''} + {X}.\mathit{name}}, keys, \{ allroot, \nexists{lroot}\})$ \\[0.2ex]
    \linenumber \RETURN $(alltree, ltree, htree)$ \\[0.2ex]
    \bottomrule
\end{tabular}
\caption{Construction of the heavy and light indicator view trees for a canonical variable order $\omega$ of a hierarchical query. The variable order $\omega^{keys}$ has the  structure of $\omega$ but each atom $R(\calY)$ is replaced with the light part $R^{keys}(\calY)$ of relation $R$ partitioned on $keys$.
{The view $\nexists{lroot}$ maps all tuples contained in $lroot$  
to $0$ and all other tuples to $1$}.}
\label{fig:skew_aware_views}
\end{figure}

Figure~\ref{fig:skew_aware_views} shows how to construct a triple of view trees for computing the indicators for $\anc(X)\cup\{X\}$, where $X$ is the root of a variable order $\omega$ that is a subtree in the variable order of a hierarchical query (thus $\anc(X)$ may be non-empty).
We first construct a view tree that computes the tuples of values for variables $keys = \anc(X) \cup \{X\}$  over the join of the relations from $\omega$.
We then build a similar view tree for the light indicator for $keys$ using a modified variable order $\omega^{keys}$ of the same structure as $\omega$ but with each relation $R$ replaced by the light part of $R$ partitioned on $keys$. 
Finally, the view tree for the heavy indicator computes the difference of all $keys$-values and those from the light indicator.

\begin{figure}[t]
\centering
\setlength{\tabcolsep}{3pt}
\renewcommand{\linenumber}{\makebox[2ex][r]{\rownumber\STAB\,}}
\renewcommand{\arraystretch}{1.05}
\setcounter{magicrownumbers}{0}
\begin{tabular}[t]{@{}c@{}cl@{}}
    \toprule
    \multicolumn{3}{l}{$\tau(\text{variable order } \omega, \text{free variables } \mathcal{F})$ : set of view trees} \\
    \midrule
    \multicolumn{3}{l}{\MATCH $\omega$:}\\
    \midrule
    \phantom{abc} & $R(\calY)$\hspace*{2.1em} & 
    \linenumber \RETURN $\{R(\calY)\}$\\[2pt]
    \cmidrule{2-3} \\[-6pt]
    &    
    \begin{minipage}[t]{1.7cm}
        \hspace*{-0.4cm}        
        \begin{tikzpicture}[xscale=0.45, yscale=1]
            \node at (0,-2)  (n4) {$X$};
            \node at (-1,-3)  (n1) {$\omega_1$} edge[-] (n4);
            \node at (0,-3)  (n2) {$\ldots$};
            \node at (1,-3)  (n3) {$\omega_k$} edge[-] (n4);
        \end{tikzpicture}        
    \end{minipage}
    &
    \begin{minipage}[t]{11.25cm}
        \vspace{-1.55cm}
        \linenumber \LET $keys = \anc(X) \cup \{X\}$ \\[0.5ex]
        \linenumber \LET $\mathcal{F}_X = \anc(X) \cup (\mathcal{F} \cap \vars(\omega))$ \\[0.5ex]
        \linenumber \LET $Q_X(\mathcal{F}_X) = \text{join of } \atoms(\omega)$ \\[0.5ex]
        \linenumber \IF $(mode = \text{`static'} \land Q_X(\mathcal{F}_X) \text{ is free-connex})\,\lor$ \\[0.5ex]
        \linenumber \TAB $(mode = \text{`dynamic'} \land Q_X(\mathcal{F}_X) \text{ is $\dfw_0$-hierarchical})$ \\[0.5ex]
        \linenumber \TAB \RETURN $\{\, \textsc{BuildVT}({\text{``\,V\,''}}, \omega, \mathcal{F}_X) \,\}$ \\[0.5ex]
        \linenumber \IF $X \in \mathcal{F}$ \\[0.5ex]
        \linenumber \TAB \RETURN $\{\, \textsc{NewVT}({\text{``\,V\_\,''} + X.\mathit{name}}, keys, \{\hat{T}_1, {.}{.}{.}, \hat{T}_k\})$ \\[0.5ex] 
        \linenumber \TAB\TAB\TAB\TAB\STAB $\mid 
        T_i \in \tau(\omega_i, \mathcal{F})_{i \in [k]},$ \\[0.5ex]
        \linenumber \TAB\TAB\TAB\TAB\TAB $\hat{T}_i = \textsc{AuxView}(\text{root of } \omega_i, T_i )_{i\in[k]} \}$ \\[0.5ex]        
        \linenumber \LET $(\_,\_,H_{X}) = \text{roots of } \textsc{IndicatorVTs}(\omega)$ \\[0.5ex]
        \linenumber \LET $htrees = \{\, \textsc{NewVT}({\text{``\,V\_\,''} + X.\mathit{name}}, keys, \{\exists H_{X}, \hat{T}_1, {.}{.}{.}, \hat{T}_k\})$ \\[0.5ex] 
        \linenumber \TAB\TAB\TAB\TAB\TAB\; $\mid T_i \in \tau(\omega_i, \mathcal{F})_{i\in[k]},$\\[0.5ex] 
        \linenumber \TAB\TAB\TAB\TAB\TAB\TAB $\hat{T}_i = \textsc{AuxView}(\text{root of } \omega_i, T_i )_{i\in[k]}\}$ \\[0.5ex]
        \linenumber \LET $ltree = \textsc{BuildVT}({\text{``\,V\,''}}, \omega^{keys}, \mathcal{F}_X)$ \\[0.5ex]
        \linenumber \RETURN $htrees \cup \{\, ltree \,\}$\\[-1.5ex]
    \end{minipage}\\
    \bottomrule
\end{tabular}
\caption{Construction of skew-aware view trees for a canonical variable order $\omega$ of a hierarchical query with free variables $\mathcal{F}$. The global parameter $mode \in \{\text{`static'},\text{`dynamic'}\}$ specifies the evaluation mode.
The variable order $\omega^{keys}$ has the structure of $\omega$ but each atom $R(\calY)$ is replaced by the light part $R^{\hspace{0.2mm}keys}(\calY)$ of relation $R$ partitioned on $keys$.
{The view $\exists H_{X}$ maps all tuples contained in $H_X$  
to $1$ and all other tuples to $0$.}}
\label{fig:view_forest_main_algo}
\end{figure}

\begin{figure*}[t]
\begin{center}
  \begin{minipage}[b]{0.43\textwidth}
    \scalebox{0.8}{
    \begin{tikzpicture}[xscale=1.1, yscale=1]
      \node at (0.0, 0.0) (A) {{\small \color{black} $A$}};
      \node at (-1.8, -0.8) (B) {{\small\color{black} $B$}} edge[-] (A);
      \node at (1.8, -0.8) (C) {{\small\color{black} $\underline{C}$}} edge[-] (A);
      \node at (-2.7, -1.6) (D) {{\small\color{black} $\underline{D}$}} edge[-] (B);
      \node at (-1, -1.6) (E) {{\small\color{black} $\underline{E}$}} edge[-] (B);
      \node at (0.9, -1.6) (F) {{\small\color{black} $\underline{F}$}} edge[-] (C);
      \node at (2.7, -1.6) (G) {{\small\color{black} $G$}} edge[-] (C);
      \node at (-2.7, -2.4) (R) {{\small \color{black} $R(A,B,\underline{D})$}} edge[-] (D);
      \node at (-1, -2.4) (S) {{\small \color{black} $S(A,B,\underline{E})$}} edge[-] (E);
      \node at (0.9,  -2.4)(T) {{\small \color{black} $T(A,\underline{C},\underline{F})$}} edge[-] (F);
      \node at (2.7,  -2.4)(U) {{\small \color{black} $U(A,\underline{C},G)$}} edge[-] (G);
    \end{tikzpicture}
    }
  \end{minipage}
  \begin{minipage}[b]{0.45\textwidth}
    \small
    \scalebox{0.8}{
    \begin{tikzpicture}[xscale=1.4, yscale=1]

        \node at (0.0, -0.65) (A) {{\small \color{black} $V_{A}({ \underline{C},\underline{D},\underline{E}, \underline{F}})$}};
        \node at (-1.4, -1.5) (B) {{\small \color{black} $V_B(A,\underline{D},\underline{E})$}} edge[-] (A);
        \node at (1.4, -1.5) (C) {{\small \color{black} $V_C(A,\underline{C},\underline{F})$}} edge[-] (A);
        \node at (2.175, -2.25) (G) {{\small \color{black} $V_G(A,\underline{C})$}} edge[-] (C);
        \node at (-2.175, -3.0) (R) {{\small \color{black} $R^A(A,B,\underline{D})$}} edge[-] (B);
        \node at (-0.75, -3.0) (S) {{\small \color{black} $S^A(A,B,\underline{E})$}} edge[-] (B);
        \node at (0.7, -3.0) (T) {{\small \color{black} $T^A(A,\underline{C},\underline{F})$}} edge[-] (C);
        \node at (2.175, -3.0) (U) {{\small \color{black} $U^A(A,\underline{C},G)$}} edge[-] (G);        
    \end{tikzpicture}
    }
  \end{minipage}
  
  \smallskip
  \begin{minipage}[b]{0.43\textwidth}
    \small
    \scalebox{0.8}{
    \begin{tikzpicture}[xscale=1.35, yscale=1]
      \node at (0.5, 0.0) (V_{A}) {{\small \color{black} $V_{A}({\color{black} A})$}};
      \node at (-1.1, -0.75) (H_{A}) {{\small \color{black} $\exists H_{A}({ A})$}} edge[-] (V_{A});
      \draw[black, dashed] (-0.05,-0.5) rectangle (1.05,-1.0);
      \node at (0.5, -0.75) (Vp_{B}) {{\small \color{black} $V'_{B}({ A})$}} edge[-] (V_{A});
      \node at (0.5, -1.5) (V_{B}) {{\small \color{black} $V_{B}({ A,\underline{D},\underline{E}})$}} edge[-] (Vp_{B});

      \node at (-1.1, -3.0) (R) {{\small \color{black} $R^{AB}(A,B,\underline{D})$}} edge[-] (V_{B});
      \node at (0.5, -3.0) (S) {{\small \color{black} $S^{AB}(A,B,\underline{E})$}} edge[-] (V_{B});

      \draw[black, dashed] (1.37,-0.5) rectangle (2.47,-1.0);
      \node at (1.92, -0.75) (Vp_{C}) {{\small \color{black} $V'_{C}({ A})$}} edge[-] (V_{A});
      \node at (1.92, -1.5) (V_{C}) {{\small \color{black} $V_{C}({ A,\underline{C}})$}} edge[-] (Vp_{C});
      \draw[black, dashed] (1.37,-2.0) rectangle (2.47,-2.5);
      \node at (1.92, -2.25) (VFp) {{\small \color{black} $T'(A,\underline{C})$}} edge[-] (V_{C});
      \node at (3.2, -2.25) (VG) {{\small \color{black} $V_G(A,\underline{C})$}} edge[-] (V_{C});
      \node at (1.92, -3.0) (T) {{\small \color{black} $T(A,\underline{C},\underline{F})$}} edge[-] (VFp);
      \node at (3.2, -3.0) (U) {{\small \color{black} $U(A,\underline{C},G)$}} edge[-] (VG);
    \end{tikzpicture}
    }
  \end{minipage}
  \begin{minipage}[b]{0.56\textwidth}
    \small
    \scalebox{0.8}{
    \begin{tikzpicture}[xscale=1.65, yscale=1]
      \node at (0.0, 0.0) (V_{A}) {{\small \color{black} $V_{A}({\color{black} A})$}};
      \node at (-1.8, -0.75) (H_{A}) {{\small \color{black} $\exists H_{A}({ A})$}} edge[-] (V_{A});
      \draw[black, dashed] (-0.55,-0.5) rectangle (0.55,-1.0);
      \node at (0.0, -0.75) (Vp_{B}) {{\small \color{black} $V'_{B}({\color{black} A})$}} edge[-] (V_{A});
      \node at (0.0, -1.5) (V_{B}) {{\small \color{black} $V_{B}({\color{black} A,B})$}} edge[-] (Vp_{B});
      \node at (-1.6, -1.8) (H_{AB}) {{\small \color{black} $\exists H_{B}({ A,B})$}} edge[-] (V_{B});
      \draw[black, dashed] (-1.05,-2.0) rectangle (-0.15,-2.5);
      \node at (-0.6, -2.25) (V_{D}) {{\small \color{black} $R'({ A,B})$}} edge[-] (V_{B});
      \node at (-0.6, -3.0) (R) {{\small \color{black} $R(A,B,\underline{D})$}} edge[-] (V_{D});
      \draw[black, dashed] (0.15,-2.0) rectangle (1.05,-2.5);
      \node at (0.6, -2.25) (V_{E}) {{\small \color{black} $S'({ A,B})$}} edge[-] (V_{B});
      \node at (0.6, -3.0) (S) {{\small \color{black} $S(A,B,\underline{E})$}} edge[-] (V_{E});

      \draw[black, dashed] (1.25,-0.5) rectangle (2.35,-1.0);
      \node at (1.8, -0.75) (Vp_{C}) {{\small \color{black} $V'_{C}({ A})$}} edge[-] (V_{A});
      \node at (1.8, -1.5) (V_{C}) {{\small \color{black} $V_{C}({ A,\underline{C}})$}} edge[-] (Vp_{C});
      \draw[black, dashed] (1.3,-2.0) rectangle (2.3,-2.5);
      \node at (1.8, -2.25) (VFp) {{\small \color{black} $T'(A,\underline{C})$}} edge[-] (V_{C});
      \node at (2.85, -2.25) (VG) {{\small \color{black} $V_G(A,\underline{C})$}} edge[-] (V_{C});
      \node at (1.8, -3.0) (T) {{\small \color{black} $T(A,\underline{C},\underline{F})$}} edge[-] (VFp);
      \node at (2.85, -3.0) (U) {{\small \color{black} $U(A,\underline{C},G)$}} edge[-] (VG);
    \end{tikzpicture}
    }
  \end{minipage}

\smallskip
  \begin{minipage}[b]{0.45\linewidth}
    \small
    \scalebox{0.8}{
    \begin{tikzpicture}[xscale=1.13, yscale=1]
      \node at (0.0, 0.0) (All_{A}) {{\small \color{black} $All_{A}({ A})$}};
      \node at (-1.9, -0.8) (All_{B}) {{\small \color{black} $All_{B}({ A})$}} edge[-] (All_{A});
      \node at (1.9, -0.8) (All_{C}) {{\small \color{black} $All_{C}({ A})$}} edge[-] (All_{A});
      \node at (-2.7, -1.6) (All_{D}) {{\small \color{black} $All_{D}({ A,B})$}} edge[-] (All_{B});
      \node at (-1.2, -1.6) (All_{E}) {{\small \color{black} $All_{E}({ A,B})$}} edge[-] (All_{B});
      \node at (1.1, -1.6) (All_{F}) {{\small \color{black} $All_{F}({ A,\underline{C}})$}} edge[-] (All_{C});
      \node at (2.7, -1.6) (All_{G}) {{\small \color{black} $All_{G}({ A,\underline{C}})$}} edge[-] (All_{C});
      \node at (-2.7, -2.4) (R) {{\small \color{black} $R(A,B,\underline{D})$}} edge[-] (All_{D});
      \node at (-1.2, -2.4) (S) {{\small \color{black} $S(A,B,\underline{E})$}} edge[-] (All_{E});
      \node at (1.1,  -2.4)(T) {{\small \color{black} $T(A,\underline{C},\underline{F})$}} edge[-] (All_{F});
      \node at (2.7,  -2.4)(U) {{\small \color{black} $U(A,\underline{C},G)$}} edge[-] (All_{G});
    \end{tikzpicture}
    }
  \end{minipage}
  \hspace{0.5cm}
  \begin{minipage}[b]{0.45\linewidth}
    \small
    \scalebox{0.8}{
    \begin{tikzpicture}[xscale=1.3, yscale=1]
      \node at (0.0, 0.0) (L_{A}) {{\small \color{black} $L_{A}({ A})$}};
      \node at (-1.2, -0.8) (All_{B}) {{\small \color{black} $L_{B}({ A})$}} edge[-] (L_{A});
      \node at (1.2, -0.8) (All_{C}) {{\small \color{black}  $L_{C}({ A})$}} edge[-]  (L_{A});

      \node at (-2.25,  -1.6) (All_{D}) {{\small \color{black} $L_{D}({ A,B})$}} edge[-] (All_{B});
      \node at (-0.75, -1.6) (All_{E}) {{\small \color{black} $L_{E}({ A,B})$}} edge[-] (All_{B});
      \node at (0.75, -1.6) (All_{F}) {{\small \color{black}  $L_{F}({ A,\underline{C}})$}} edge[-] (All_{C});
      \node at (2.25,  -1.6) (All_{G}) {{\small \color{black}  $L_{G}({ A,\underline{C}})$}} edge[-] (All_{C});

      \node at (-2.25,  -2.4) (R^L) {{\small \color{black} $R^{A}(A,B,\underline{D})$}} edge[-]  (All_{D});
      \node at (-0.75, -2.4) (S^L) {{\small \color{black} $S^{A}(A,B,\underline{E})$}} edge[-]  (All_{E});
      \node at (0.75,  -2.4) (T^L) {{\small \color{black}  $T^{A}(A,\underline{C},\underline{F})$}} edge[-] (All_{F});
      \node at (2.25,   -2.4) (U^L) {{\small \color{black}  $U^{A}(A,\underline{C},G)$}} edge[-] (All_{G});
    \end{tikzpicture}
    }
  \end{minipage}
  
  \smallskip
  \begin{minipage}[b]{0.24\linewidth}
    \small
    \scalebox{0.8}{
    \begin{tikzpicture}[xscale=1.15, yscale=1]
      \node at (0.0, 0.0) (H_{B}) {{\small \color{black} $H_{A}({ A})$}};
      \node at (-0.6, -1.6) (All_{B}) {{\small \color{black} $All_A({ A})$}} edge[-] (H_{B});
      \node at (0.6, -1.6) (L_{B}) {{\small \color{black} $\nexists L_{A}({ A})$}} edge[-] (H_{B});
    \end{tikzpicture}
    }
  \end{minipage}
  \hspace{-1cm}
  \begin{minipage}[b]{0.24\linewidth}
    \small
    \centering
    \scalebox{0.8}{
    \begin{tikzpicture}[xscale=1.15, yscale=1]
      \node at (0.0, 0.0) (All_{A}) {{\small \color{black} $All_{B}({ A,B})$}};
      \node at (-0.8, -0.8) (V_{B}) {{\small \color{black} $All_{D}({ A,B})$}} edge[-] (All_{A});
      \node at ( 0.8, -0.8) (V_{C}) {{\small \color{black} $All_{E}({ A,B})$}} edge[-] (All_{A});

      \node at (-0.8, -1.6) (R) {{\small \color{black} $R(A,B,\underline{D})$}} edge[-] (V_{B});
      \node at ( 0.8, -1.6) (S) {{\small \color{black} $S(A,B,\underline{E})$}} edge[-] (V_{C});
      
    \end{tikzpicture}
    }
  \end{minipage}
\hspace{0.2cm}
  \begin{minipage}[b]{0.26\linewidth}
    \small
    \centering
    \scalebox{0.8}{
    \begin{tikzpicture}[xscale=1.3, yscale=1]
      \node at (0.0, 0.0) (L_{A}) {{\small \color{black} $L_{B}({ A,B})$}};
      \node at (-0.85, -0.8) (V_{B}) {{\small \color{black} $L_{D}({ A,B})$}} edge[-] (L_{A});
      \node at ( 0.85, -0.8) (V_{C}) {{\small \color{black} $L_{E}({ A,B})$}} edge[-] (L_{A});
      \node at (-0.85, -1.6) (R^L) {{\small \color{black} $R^{{AB}}(A,B,\underline{D})$}} edge[-] (V_{B});
      \node at ( 0.85, -1.6) (S^L) {{\small \color{black} $S^{{AB}}(A,B,\underline{E})$}} edge[-] (V_{C});
    \end{tikzpicture}
    }
  \end{minipage}
\hspace{0.4cm}
  \begin{minipage}[b]{0.24\linewidth}
    \small
    \centering
    \scalebox{0.8}{
    \begin{tikzpicture}[xscale=1.15, yscale=1] 
      \node at (0.0, 0.0) (H_{B}) {{\small \color{black} $H_{B}({ A,B})$}};
      \node at (-0.8, -1.6) (All_{B}) {{\small \color{black} $All_{B}({ A,B})$}} edge[-] (H_{B});
      \node at (0.8, -1.6) (L_{B}) {{\small \color{black} $\nexists L_{B}({ A,B})$}} edge[-] (H_{B});
    \end{tikzpicture}
    }
  \end{minipage}

  \end{center}
  \caption{
    Canonical variable order for the query $Q(C,D,E,F)\allowbreak =\allowbreak R(A,B,D),\allowbreak S(A,B,E),\allowbreak T(A,C,F),\allowbreak U(A,C,G)$ (top left). 
    The three view trees constructed for the query (top right and second row).
    The indicator view trees for computing $H_A$  and $H_B$ (third and fourth row).
    The views with a dashed box are only needed for dynamic query evaluation.}
  \label{fig:view_tree_example_21}
\end{figure*}

\paragraph{View Trees with Indicators.}
Figure~\ref{fig:view_forest_main_algo} gives the algorithm for constructing the view trees for a variable order $\omega$ of a hierarchical query $Q(\mathcal{F})$.
The algorithm traverses the variable order $\omega$ top-down, maintaining the invariant that all ancestors of a node are free variables (or treated as such in case of bound join variables whose values are heavy).

The free variables at node $X$ are the ancestors of $X$ and the free variables in the subtree rooted at $X$ (Line 3). 
If the residual query $Q_X$ at node $X$ (Line 4) is free-connex in the static case or $\dfw_0$-hierarchical in the dynamic case, we return a view tree for $Q_X$ (Lines 5-7).
If $X$ is free, we recursively compute a set of view trees for each child of $X$. 
We may extend the root of each child tree with an auxiliary view in the dynamic mode to support constant-time propagation of updates coming via the siblings of $X$.
For each combination of the child view trees, we form a new view joining the roots of the child view trees and using $X$ and its ancestors as free variables (Lines 8-11).
If $X$ is bound, we create two evaluation strategies for the residual query $Q_X$ based on the degree of values of $X$ and its ancestors in the relations of $Q_X$.
We construct the indicator view trees for $X$ and its ancestors (Line 12). 
The heavy indicator restricts the joins of the child views to only heavy values for the tuple of $X$ and its ancestors (Lines 13-15). 
We also construct a view tree over the light parts of the relations in $\omega$ (Line 16). 

{The algorithm from Figure 11 uses different criteria for the static and dynamic cases (Lines 5-6) to decide on whether to stop recursively traversing the variable order. Since the class of $\delta_0$-hierarchical queries is a proper subset of the class of free-connex queries, 
the algorithm may partition input relations on more attributes and create more view trees in the dynamic case than in the static case for the same variable order and free variables.}

We next showcase our approach on a non-free-connex query. 
Section~\ref{sec:examples} provides additional examples with $\dfw_1$-hierarchical queries.

\begin{exa}
\rm
Figure~\ref{fig:view_tree_example_21} shows the view trees for the query $$Q(C,D,E,F) = R(A,B,D),\, S(A,B,E),\, T(A,C,F),\, U(A,C,G).$$

We start from the root $A$ in the variable order. Since $Q$ is not free-connex (and also  not $\dfw_0$-hierarchical) and $A$ is bound, we create the view trees for the indicators $H_A(A)$ and $L_A(A)$. Materializing the views in these view trees takes linear time. 

In the light case for $A$, we create a view tree with the root $V_A(C,D,E,F)$ and the leaves being the light parts of the input relations partitioned on $A$ (bottom-left). 
Computing $V_G(A,C)$ and $V_C(A,C,F)$ takes linear time. 
We compute the view $V_B(A,D,E)$ in time $\bigO{N^{1+\eps}}$: For each $(a,b,d)$ tuple in $R^A$, we iterate over at most $N^\eps$ $(a,b,e)$ values in $S^{A}$. The view $V_B(A,D,E)$ contains at most $N^{1+\eps}$ tuples.
Similarly, we compute $V_A(C,D,E,F)$ in time $\bigO{N^{1+2\eps}}$: For each $(a,d,e)$ tuple in $V_B$, we iterate over at most $N^\eps$ $(a,c,f)$ values in $V_C$. 
The view $V_A(C,D,E,F)$ allows constant delay enumeration of its result.

In the heavy case for $A$, we recursively process the subtrees of $A$ in $\omega$ and treat $A$ as free.
The right subquery, $Q_C(A,C,F)=T(A,C,F),U(A,C,G)$ is free-connex and 
$\dfw_0$-hierarchical, thus we compute its view tree with the root $V_C(A,C)$ in the static case and the root $V'_C(A)$ in the dynamic case (view trees in the second row) in linear time.
The left subquery $Q_B(A,D,E) = R(A,B,D),S(A,B,E)$, however, is neither free-connex nor $\dfw_0$-hierarchical. Since $B$ is bound, we create the indicator relations $H_B(A,B)$ and $L_B(A,B)$ in linear time. We distinguish two new cases:
In the light case for $(A,B)$, we construct a view tree with the root $V_B(A,D,E)$ $= R^{AB}(A,B,D),S^{AB}(A,B,E)$ (second row left) and compute $V_B(A,D,E)$ in time $\bigO{N^{1+\eps}}$ by iterating over $R^{AB}$ and,  for each $(a,b,d)$, iterating over at most $N^\eps$ $E$-values in $S^{AB}$. 
In the heavy case for $(A,B)$, we process the subtrees of $B$ considering $B$ as free variable. The two subqueries, $Q_D(A,B,D) = R(A,B,D)$ and $Q_E(A,B,E) = S(A,B,E)$, are $\dfw_0$-hierarchical.

Overall, we create three view trees for $Q$ and two sets of view trees for the indicator relations at $A$ and $B$. The time needed to compute these view trees is $\bigO{N^{1+2\eps}}$.
\qed
\end{exa}

Given a hierarchical query, our algorithm  effectively rewrites it into an equivalent union of queries, with one query defined by the join of the leaves of a view tree.

\begin{prop}
\label{prop:equivalence}
Let $\{T_1,\ldots,T_k\} = \tau(\omega,\calF)$ 
be the set of view trees constructed by the algorithm in Figure~\ref{fig:view_forest_main_algo} for a given hierarchical query $Q(\mathcal{F})$ and a canonical variable order $\omega$ for $Q$.
Let $Q^{(i)}(\mathcal{F})$ be the query defined by the conjunction of the leaf atoms in $T_i$, $\forall i\in[k]$. Then, $Q(\mathcal{F}) \equiv \bigcup_{i\in[k]} Q^{(i)}(\mathcal{F})$.
\end{prop}


The preprocessing time of our approach is given by the time to materialize the views in the view trees.

\begin{prop}
\label{prop:preproc_time}
Given a hierarchical query $Q(\calF)$ with
 static width $\fw$, a canonical variable order 
$\omega$ for $Q$,
a database  of size $N$, and $\eps \in [0,1]$, 
the views in the set of view trees $\tau(\omega,\calF)$
can be materialized in $\bigO{N^{1+(\fw-1)\eps}}$ time.
\end{prop}

\section{Enumeration}
\label{sec:enumeration}
\begin{figure}[t]
{  
  \begin{center}
  \renewcommand{\arraystretch}{1.15}
  \renewcommand{\linenumber}{\makebox[3ex][r]{\rownumber\TAB}}
  \newcommand\hfilll{\hspace{0pt plus 1filll}}
  \setcounter{magicrownumbers}{0}
  \begin{tabular}{@{\hskip 0.1in}l}
  \toprule
  $T.\mathit{open}(\text{tuple } \mathit{ctx})$ \\
  \midrule
  \linenumber \LET $V(\calS) = $ root of $T$ \\
  \linenumber $V.\mathit{open}(\mathit{ctx})$\\
  \linenumber $T.\mathit{buckets} := \emptyset$ \\
  \linenumber \LET $\{T_1, \ldots, T_k\} = $ children of $T$ \\  
  \linenumber \IF $\exists i \in [k] \text{ such that } T_i = \exists H$ \qquad \hfilll \text{\color{gray} // heavy indicator as child}\\
  \linenumber \TAB $\exists H.\mathit{open}(\mathit{ctx})$ \\
  \linenumber \TAB \WHILE $(h := \exists H.\mathit{next}()) \neq \EOF$ \\
  \linenumber \TAB \TAB $T' := $ shallow copy of $T$ without $\exists H$ \\
  \linenumber \TAB \TAB $T'.open(\mathit{h})$\\
  \linenumber \TAB\TAB $T.\mathit{buckets} := T.\mathit{buckets} \cup \{T'\}$  \\
  \linenumber \ELSE \IF $\calS \subset \text{free variables in $T$}$ {\hfilll\text{\color{gray} // need to recurse}} \\
  \linenumber \TAB $T.\mathit{ctx} := V.\mathit{next}() $ {\hfilll\text{\color{gray} // current context for entire tree $T$}}\\
  \linenumber \TAB \FOREACH $i \in [k]$ \DO $T_i.\mathit{open}(T.\mathit{ctx})$\\    
  \linenumber $T.\mathit{next}()$ \qquad\qquad\qquad {\hfilll\text{\color{gray} // initializes $T.\mathit{tuple}$ to first tuple to be returned}}\\
  \bottomrule
  \end{tabular}
  \end{center}
}
  \caption{Open the view iterators in a view tree.}
  \label{fig:enumeration-open}
\end{figure}
For any hierarchical query, Section~\ref{sec:preprocessing} constructs a set of view trees that together represent the query result.
We now show how to enumerate the distinct tuples in the query result with their multiplicity using the $open/next/close$ iterator model for such view trees. 

{
Each view in a view tree follows the iterator model. The function $V.\mathit{open}(ctx)$ initializes the iterator on view $V$ using the tuple $ctx$ as context, setting the range of the iterator to those tuples that are consistent with $ctx$ in $V$, that is, $ctx$ is part of each such tuple in $V$.
The function $V.\mathit{next}()$ returns a tuple consistent with $ctx$ in $V$; or it returns \EOF if the tuples in the range of the iterator are exhausted.
The tuples returned by $V.\mathit{next}()$ are distinct.
Both functions operate in constant time, as per our computational model.
}

Given a subtree $T$ of a view tree and the current tuple $ctx$ in its parent view,
the call $\mathit{T.open(ctx)}$ described in Figure~\ref{fig:enumeration-open} 
sets the range of the iterator of $T$ to those tuples in its root view that agree with 
 $\mathit{ctx}$ and positions the iterator at the first tuple in this range. The $open$ call is recursively propagated down the view tree with an increasingly more specific context tuple. 
{A $T.close()$ call resets the iterators of tree $T$. Each subtree $T$ has an attribute $T.\mathit{tuple}$ storing the next tuple to be reported. The $\mathit{open}$ method ends with a call to $T.\mathit{next}()$ to set $T.\mathit{tuple}$ to the first tuple to be reported.}

There are two cases that need special attention. If the schema of a view $V$ includes all free variables in the subtree rooted at $V$, then there is no need to open the views in this subtree since $V$ already has the tuples over these free variables; e.g., this is the case of the view $V_A(C,D,E,F)$ in Figure~\ref{fig:view_tree_example_21}. 
The views with heavy indicators, e.g., the views $V_A(A)$ and $V_B(A,B)$ in Figure~\ref{fig:view_tree_example_21}, also require special treatment. If $V$ has as child a heavy indicator $\exists H$, the tree $T$ rooted at $V$ represents possibly overlapping relations in the contexts given by the different tuples $h\in\exists H$. 
We {\em ground} the heavy indicator by creating an iterator for each heavy tuple agreeing with the current tuple $\mathit{ctx}$ at the parent view of $V$ and keep this iterator in a {\it shallow} copy of $T$. {Creating a shallow copy of $T$ means creating a tree of iterators of the same structure as $T$ but without copying the content of views under $T$.}

\begin{figure}[t]
{  
  \begin{center}
  \renewcommand{\arraystretch}{1.15}
  \renewcommand{\linenumber}{\makebox[3ex][r]{\rownumber\TAB}}
  \newcommand\hfilll{\hspace{-3pt plus 1filll}}
  \setcounter{magicrownumbers}{0}
  \begin{tabular}{l}
  \toprule
  $T.\mathit{next}(\,) : \text{tuple}$ \\
  \midrule
  \linenumber \LET $V(\calS) = $ root of $T$ \\
  \linenumber \IF $(T \mbox{ has no children}) \lor (\calS = \mbox{free variables in }T)$ \qquad {\hfilll\text{\color{gray} // no need to recurse}}\\
  \linenumber \TAB $t := T.\mathit{tuple};\ T.\mathit{tuple} := V.\mathit{next}();\,$ \RETURN $t$ \\
  \linenumber \IF $T.\mathit{buckets} \neq \emptyset$ \\
  \linenumber \TAB $t := T.\mathit{tuple};\ T.\mathit{tuple} := \textsc{Union}(T.\mathit{buckets});\,$ \RETURN $t$\\
  \linenumber \LET $\{T_1, \ldots, T_k\} = $ children of $T$ \\
  \linenumber \WHILE ($T.\mathit{ctx} \neq $ \EOF) \DO \\  
  \linenumber \TAB \IF $(\mathit{n} := \textsc{Product}(T_1,\ldots,T_k, T.\mathit{ctx})) \neq $ \EOF \quad{\hfilll\text{\color{gray} // next tuple in Cartesian product}}\\
  \linenumber \TAB \TAB $t := T.\mathit{tuple};\ T.\mathit{tuple} := \mathit{n};\,$ \RETURN $t$ \\
  \linenumber \TAB $T.\mathit{ctx} := V.\mathit{next}()$ \qquad{\hfilll\text{\color{gray} // Cartesian product exhausted, next tree context}}\\
  \linenumber \TAB \FOREACH $i \in [k]$ \DO  $T_i.\mathit{close}();\;T_i.\mathit{open}(T.\mathit{ctx})$\\
  \linenumber $t := T.\mathit{tuple};\ T.\mathit{tuple} := \EOF;\,$ \RETURN $t$\\
  \bottomrule
  \end{tabular}
  \end{center}
}
  \caption{Find the next tuple in a view tree.}
  \label{fig:enumeration-next}
\end{figure}

After the first $open$ call for a view tree $T$, we can enumerate the distinct tuples from $T$ with their multiplicity by calling $T.next()$, see Figure~\ref{fig:enumeration-next}. 
The $next$ call propagates recursively down $T$ and observes the same cases as the $open$ call. If a view $V$ in $T$ already covers all free variables in $T$, then it suffices to enumerate from $V$. If $T$ has as child a heavy indicator, we return the next tuple and its multiplicity from the union of all its groundings using the \textsc{Union} algorithm (Section~\ref{sec:union_algorithm}). Otherwise, we synthesize the returning tuple out of the tuples at the iterators of $T$'s children.
{Given the current context at $T$'s view}, we return the next tuple and its multiplicity from the Cartesian product of the tuples produced by $T$'s children using the \textsc{Product} algorithm (Section~\ref{sec:product_algorithm}). 

{
  For a view tree with no heavy indicators, 
  calling $\mathit{open}$ and $\mathit{next}$ on the view tree translates to calling $\mathit{open}$ and $\mathit{next}$ on its views, where each such call on a view takes constant time and the number of such calls is independent of the size of the database. Thus, calling $\mathit{open}$ and $\mathit{next}$ on a view tree with no heavy indicators takes constant time. 

  In the presence of heavy indicators, 
  the time to initialize a view tree and produce the next tuple is dominated by the number of shallow view trees created in the grounding step.
  The delay of the \textsc{Union} algorithm is the sum of the delays of the grounded view trees. 
  Their number is determined by the size of the heavy indicators in the view tree, which is $\bigO{N^{1-\eps}}$. Thus, calling $\mathit{open}$ and $\mathit{next}$ on a view tree with heavy indicators takes $\bigO{N^{1-\eps}}$ time. 
}

So far we discussed the case of enumerating from one view tree. In case of a set of view trees we again use the \textsc{Union} algorithm.
In case the query has several connected components, i.e., it is a Cartesian product of hierarchical queries, we use the \textsc{Product} algorithm with an empty context.

The multiplicity for a tuple returned by the \textsc{Union} algorithm is the sum of the multiplicities of its occurrences across the buckets, while for a tuple returned by the \textsc{Product} algorithm it is the multiplication of the multiplicities of the constituent tuples. Since all tuples in the database have positive multiplicities, the derived multiplicities are always strictly positive and therefore the returned tuple is part of the result.

We next explain the \textsc{Union} and \textsc{Product} algorithms used by $T.next()$ in Figure~\ref{fig:enumeration-next}.

\subsection{The \textsc{Union} Algorithm}
\label{sec:union_algorithm}

\begin{figure}[t]
  \begin{center}
  \renewcommand{\arraystretch}{1.2}
  \renewcommand{\linenumber}{{\rownumber\TAB}}
  \setcounter{magicrownumbers}{0}
  \begin{tabular}{l@{}}
    \toprule
    \textsc{Union}$(\text{view trees } T_1, \ldots, T_n):$ tuple \\
    \midrule
    \linenumber \IF ($n=1$) \RETURN $T_n.next(\,)$ \\
    \linenumber \IF (${(t, m)} := $ \textsc{Union}($T_1,\ldots, T_{n-1}$)) $\neq$ \EOF \\
    \linenumber \TAB \IF $T_n.\mathit{lookup}\hspace{0.06em}({t}) \neq 0$ \\
    \linenumber \TAB\TAB $(t_n, m_n) := T_n.next(\,)$ \\
    \linenumber \TAB\TAB \RETURN ${(t_n, m_n + \sum_{i\in[n-1]} T_i.\mathit{lookup}\hspace{0.06em}(t_n))}$ \\  
    \linenumber \TAB \RETURN ${(t, m)}$ \\
    \linenumber \IF $((t_n, m_n) := T_n.next(\,)) \neq $ \EOF \\
    \linenumber \TAB \RETURN ${(t_n, m_n + \sum_{i\in[n-1]} T_i.\mathit{lookup}\hspace{0.06em}(t_n))}$ \\
    \linenumber \RETURN \EOF \\
    \bottomrule
  \end{tabular}
\end{center}
\caption{Find the next tuple in a union of view trees.}
\label{fig:enumeration-union}
\end{figure}
\begin{figure}[t]
  \begin{center}
  \renewcommand{\arraystretch}{1.15}
  \renewcommand{\linenumber}{\makebox[3ex][r]{\rownumber\TAB}}
  \setcounter{magicrownumbers}{0}
  \begin{tabular}{l@{}}  
  \toprule
  \textsc{Product}$(\text{view trees } T_1,\ldots,T_k, \text{tuple } \mathit{ctx}):$ tuple \\
  \midrule
  \linenumber \WHILE ($T_1.\mathit{tuple} \neq $ \EOF) \DO \\
              \TAB \TAB $\ldots$ \\      
  \linenumber \TAB \TAB \WHILE ($T_{k-1}.\mathit{tuple} \neq $ \EOF) \DO \\
  \linenumber \TAB \TAB \TAB \WHILE ($T_k.\mathit{tuple} \neq $ \EOF) \DO \\
  \linenumber \TAB \TAB \TAB \TAB \LET $(t_i, m_i) = T_i.\mathit{tuple}, \forall i\in[k]$ \\
  \linenumber \TAB \TAB \TAB \TAB \LET $(t_{\mathit{ctx}}, \_) = \mathit{ctx}, \text{ where } t_{\mathit{ctx}} \text{ is over schema } \calS$\\
  \linenumber \TAB \TAB \TAB \TAB $t := t_{\mathit{ctx}} \circ \bigcirc_{i\in[k]} \pi_{\text{free variables in }T_i-\calS\,}t_i$ \\
  \linenumber \TAB \TAB \TAB \TAB $m := \prod_{i\in[k]} m_i$ \\
  \linenumber \TAB \TAB \TAB \TAB $T_k.next()$ \\      
  \linenumber \TAB \TAB \TAB \TAB \RETURN $(t,m)$ \\
  \linenumber \TAB \TAB \TAB $T_k.\mathit{close}()$;\, $T_k.\mathit{open}(\mathit{ctx})$;\, $T_{k-1}.\mathit{next}()$\\
              \TAB \TAB $\ldots$ \\      
  \linenumber \TAB $T_2.\mathit{close}()$;\, $T_2.\mathit{open}(\mathit{ctx})$;\, $T_1.\mathit{next}()$\\
  \linenumber \RETURN \EOF \\
  \bottomrule
  \end{tabular}
  \end{center}
  \caption{Find the next tuple in a product of view trees. In case $k=1$, the innermost loop is executed.}
    \label{fig:enumeration-product}
\end{figure}

The \textsc{Union} algorithm is given in Figure~\ref{fig:enumeration-union}. It
is an adaptation of prior work~\cite{Durand:CSL:11}. 
It takes as input $n$ view trees that represent possibly overlapping sets of tuples over the same relation and returns a tuple and its multiplicity in the union of these sets, where the tuple is distinct from all tuples returned before. 

We first explain the algorithm on two views $T_1$ and $T_2$ that have been already open and with their iterators positioned at the first respective tuples. On each call, we return one tuple together with its multiplicity or $\tup{EOF}$. We check whether the next tuple $t_1$ in $T_1$ is also present in $T_2$. If so, we return the next tuple in $T_2$ and its total multiplicity from $T_1$ and $T_2$; otherwise, we return $t_1$ and its multiplicity in $T_1$. If $T_1$ is exhausted, we return the next tuple in $T_2$ and its total multiplicity from $T_1$ and $T_2$, or $\tup{EOF}$ if $T_2$ is also exhausted.

In case of $n>2$ views, we consider one view defined by the union of the first $n-1$ views and a second view defined by $T_n$, and we then reduce the general case to the previous case of two views. 

The delay of this algorithm is given by the delay of iterating over each view, the cost of lookups into the views, and the cost of computing output multiplicities. The lookup costs are constant when using a hierarchy of materialized views for representing the query result~\cite{OlteanuZ15}.
Given $n$ views, computing an output multiplicity takes $\bigO{n}$ time.
The overall delay is the sum of the delays of the $n$ views, which is $\bigO{n}$. 

In our paper, we employ the \textsc{Union} algorithm in two cases: (1) on the set of view trees obtained after grounding the heavy indicators; and (2) on the set of view trees obtained by using skew-aware indicators in the preprocessing stage. In the first case, the number of the view trees is in $\bigO{N^{1-\eps}}$, since the number of heavy tuples in any heavy indicator view is at most $N^{1-\eps}$. In the second case, the number of view trees does not depend on the database size $N$, but it may depend exponentially on the number of bound join variables in the input hierarchical query.

\subsection{The \textsc{Product} Algorithm}
\label{sec:product_algorithm}
The \textsc{Product} algorithm is given in Figure~\ref{fig:enumeration-product}. It takes as input a set of view trees $T_1,\ldots,T_k$ and a context, which is the current tuple in the parent view, and outputs the next tuple and its multiplicity in the Cartesian product of {the tuples returned by the $k$ view trees given the context. 
By construction, the parent view joins the roots of the $k$ view trees and thus yields only contexts for which each of the view trees produces a non-empty result.}

In case $k=1$, we execute the innermost loop for $T_k$: On a call, we take the current tuple in $T_k$ and project away the variables that are in common with the context tuple, retaining only the free variables in $T_k$. We concatenate this projection with the context tuple. The concatenation operator is $\circ$. Before we return this concatenated tuple and its multiplicity, we advance the iterator to the next tuple-multiplicity pair in $T_k$. Eventually, we reach the end of the iterator for $T_k$, in which case we return $\tup{EOF}$.

In case $k>1$, we hold the current tuple-multiplicity pairs for $T_1,\ldots,T_{k-1}$ and iterate over $T_k$. Whenever $T_k$ reaches  $\tup{EOF}$, we reset it and advance the iterator for $T_{k-1}$. We concatenate the context tuple and the current tuples of all iterators, projected onto the variables that are not in the schema of the context tuple (since those fields are already in the context).  We multiply the current multiplicities of all iterators and advance the iterator for $T_k$ before returning the concatenated tuple and its multiplicity. 

The delay for a \textsc{Product} call is given by the sum of the delays of the $k$ input view trees. 
{In the worst case, the algorithms makes $k-1$ $\mathit{open}$ calls and $k$ $\mathit{next}$ calls before returning the next tuple.}
We use this algorithm in two cases: (1) enumerating from a view with several children in a tree (in which case the context is given as the current tuple in the view); (2) a collection of view trees, one per connected component of the input query (in which case the context is the empty tuple). In both cases, the number of parameters to the \textsc{Product} call is independent of the size of the database and only dependent on the number of atoms and respectively of connected components in the input query. This means that the delay (in data complexity) is the maximum delay of any of its parameter view trees, which is $\bigO{N^{1-\eps}}$.

We next state the complexity of enumeration in our approach.

\begin{prop}
\label{prop:enumeration}
The tuples in the result of a hierarchical query $Q(\calF)$ over a database of size $N$ can be enumerated with $\bigO{N^{1-\eps}}$ delay using the view trees constructed by $\tau(\omega,\calF)$ for a canonical variable order $\omega$ for $Q$.
\end{prop}

\section{Updates}
\label{sec:updates}

We present our strategy for maintaining the views in the set of view trees $\tau(\omega, \mathcal{F})$ constructed for a canonical variable order $\omega$ of a hierarchical query $Q(\mathcal{F})$ under updates to input relations. 
We specify here the procedure for processing a single-tuple update to any input relation.
Processing a sequence of such updates builds upon this procedure and occasional rebalancing steps (Section~\ref{app:rebalancing}). 

We write $\delta{R} = \{ \tup{x} \rightarrow m \}$ to denote a single-tuple update $\delta{R}$ mapping the tuple $\tup{x}$ to the non-zero multiplicity $m \in \mathbb{Z}$ and any other tuple to 0; i.e., $|\delta{R}| = 1$. Inserts and deletes are updates represented as relations in which tuples have positive and negative multiplicities.
We assume that after applying an update to the database, all relations and views contain no tuples with negative multiplicities.

Compared to static evaluation, our strategy for dynamic evaluation may construct additional views to support efficient updates to {\em all} input relations.
In Figure~\ref{fig:view_tree_example_21}, the view tree created for the case of heavy $(A,B)$-values (second row right) has five such additional views, marked with dashed boxes. These views enable an update to any leaf view to be propagated to the root view in constant time. 
For instance, the views $R'$ and $S'$ eliminate the need to iterate over the $D$-values in relation 
$R$ for updates to relation $S$ and $\exists{H}_B$ and respectively over the $E$-values in $S$ for updates to $R$ and $\exists{H}_B$.
Figure~\ref{fig:aux_view_creation} gives the rule for creating such views: If node $Z$ has a sibling in the variable order, then we create an auxiliary view that aggregates away $Z$ to avoid iterating over the $Z$-values for updates coming via the (auxiliary) views constructed for the siblings of $Z$.


\begin{figure}[t]
  \centering
  \setlength{\tabcolsep}{3pt}
  \renewcommand{\arraystretch}{1.1}
  \setcounter{magicrownumbers}{0}
  \begin{tabular}[t]{@{}c@{}c@{}l@{}}
      \toprule
      \multicolumn{3}{l}{\textsc{Apply}(\text{view tree} $T$, \text{update} $\delta R$) : delta view} \\
      \midrule
      \multicolumn{3}{l}{\MATCH $T$:}\\
      \midrule
      \phantom{ab} & $K(\calX)$\hspace*{1.3em} 
      & \linenumber \IF $K = R$ \\
      & & \linenumber \TAB {$R(\calX) := R(\calX) + \delta R(\calX)$}\\
      & & \linenumber \TAB \RETURN $\delta R$\\
      & & \linenumber \RETURN $\emptyset$\\[0pt]
      \cmidrule{2-3} \\[-6pt]
      &    
      \begin{minipage}[t]{2.5cm}
          \vspace{-0.4cm}
          \hspace*{-0.25cm}
          \begin{tikzpicture}[xscale=0.65, yscale=1]
              \node at (0,-2)  (n4) {$V(\calX)$};
              \node at (-1,-3)  (n1) {$T_1$} edge[-] (n4);
              \node at (0,-3)  (n2) {$\ldots$};
              \node at (1,-3)  (n3) {$T_k$} edge[-] (n4);
          \end{tikzpicture}
      \end{minipage}
      &
      \begin{minipage}[t]{12cm}
          \vspace{-0.4cm}
          \linenumber \LET $V_i(\calX_i) = {\text{root of }} T_i, \text{ for } i\in[k]$ \\[0.25ex]
          \linenumber \IF $\exists \, j \in [k]$ such that $R \in T_j$ \\[0.25ex]
          \linenumber \TAB {$\delta V_j := \textsc{Apply}(T_j, \delta R)$}\\[0.25ex]
          \linenumber \TAB {\LET $\delta V(\calX) = \text{join of } V_1(\calX_1), \ldots, \delta V_j(\calX_j), \ldots, V_k(\calX_k)\text{ proj. onto $\calX$}$} \\[0.25ex]
          \linenumber \TAB {$V(\calX) := V(\calX) + \delta V(\calX)$}\\[0.25ex]
          \linenumber \TAB \RETURN $\delta V$\\[0.25ex]
          \linenumber \RETURN $\emptyset$\\[-6pt]
      \end{minipage}\\
      \bottomrule
  \end{tabular}
  \caption{
  Updating views in a view tree $T$ for a single-tuple update $\delta R$ to relation $R$. 
  }
  \label{fig:apply_update_algo}
\end{figure}

\subsection{Processing a Single-Tuple Update}

An update $\delta{R}$ to a relation $R$ may affect multiple view trees in the set of view trees constructed by our algorithm from Figure~\ref{fig:view_forest_main_algo}.\footnote{We focus here on updates to 
hierarchical queries without repeating relation symbols. In case a relation $R$ occurs several times in a query, we treat an update to $R$ as a sequence of updates to each occurrence of $R$.}
We apply $\delta{R}$ to each such view tree in sequence, by propagating changes along the path from the leaf $R$ to the root of the view tree. For each view on this path, we update the view result with the change computed using the standard delta rules~\cite{Chirkova:Views:2012:FTD} (see Example~\ref{ex:intro}). 
To simplify the reasoning about the maintenance task, we assume that each view tree has a copy of its base relations.
We use $\textsc{Apply}(T, \delta{R})$ from Figure~\ref{fig:apply_update_algo} to propagate an update $\delta{R}$ in a view tree $T$; if $T$ does not refer to $R$, the procedure has no effect.

Updates to indicator views, however, may trigger further changes in the views constructed over them.
Consider, for instance, the heavy indicator $H_B(A,B)$ constructed over the view $All_B(A,B)$ and the light indicator $\nexists L_B(A,B)$ in Figure~\ref{fig:view_tree_example_21}. An insert $\delta{R} = \{(a,b,d) \rightarrow 1\}$ into $R$ may change the multiplicity $H_B(a,b)$ from 0 to non-zero, thus changing $\exists H_B(A,B)$ and its dependent views: $V_B(A,B)$, $V'_B(A)$, and $V_A(A)$. But if the multiplicity $H_B(a,b)$ stays 0 or non-zero after applying $\delta{R}$, then $\exists H_B$ also stays unchanged.

\begin{figure}[t]
\centering
\setlength{\tabcolsep}{3pt}
\renewcommand{\arraystretch}{1.1}
\renewcommand{\linenumber}{{\rownumber\TAB}}
\setcounter{magicrownumbers}{0}
\begin{tabular}[t]{@{}l@{}}
    \toprule
    \textsc{UpdateIndTree}(indicator tree $T_{Ind}$, update $\delta{R}):$ indicator change \\
    \midrule
    \linenumber\LET $I(\calS) = {\text{root of }}T_{Ind}$ \\[0.15ex]
    \linenumber\LET $\textit{key} = \tup{x}[\calS], \text{ where } \delta{R} = \{ \tup{x} \rightarrow m \}$ \\[0.15ex]
    \linenumber\LET $\textit{\#before} = I(key)$ \\[0.15ex]
    \linenumber $\textsc{Apply}(T_{Ind}, \delta{R})$\\[0.15ex]
    \linenumber \IF $(\hspace{0.2mm}\textit{\#before} = 0\hspace{0.2mm}) \hspace{0.2mm}\land\hspace{0.2mm} (\hspace{0.2mm}I(key) > 0\hspace{0.2mm})$\, \RETURN $\{ key \rightarrow 1 \}$\\[0.15ex]
    \linenumber\IF $(\hspace{0.2mm}\textit{\#before} > 0\hspace{0.2mm}) \hspace{0.2mm}\land\hspace{0.2mm} (\hspace{0.2mm}I(key) = 0\hspace{0.2mm})$\, \RETURN $\{ key \rightarrow -1 \}$\\[0.15ex]
    \linenumber\RETURN $\emptyset$\\[0.15ex]
    \bottomrule
\end{tabular}
\caption{Updating an indicator view tree $T_{Ind}$ for a single-tuple update $\delta{R}$ to relation $R$.}
\label{fig:update_indicator_tree}
\end{figure}

\begin{figure}[t]
\centering
\setlength{\tabcolsep}{3pt}
\renewcommand{\linenumber}{\makebox[3ex][r]{\rownumber\TAB}}
\renewcommand{\arraystretch}{1.1}
\setcounter{magicrownumbers}{0}
\begin{tabular}[t]{@{}l@{}}
    \toprule
    \textsc{UpdateTrees}(view trees $\mathcal{T}$, indicator triples $\mathcal{T}_{Ind}$, update $\delta{R}$) \\
    \midrule
    \linenumber\FOREACH $T \in \mathcal{T}$ \DO $\textsc{Apply}(T, \delta{R})$ \\[0.2ex]
    \linenumber\FOREACH $(T_{All}, T_L, T_H) \in \mathcal{T}_{Ind}$ such that $R \in T_{All}$ \DO \\[0.2ex]
    \linenumber\TAB\LET $All(\calS) = {\text{root of }}T_{All},\; L(\calS) = {\text{root of }}T_{L},\; H(\calS) = {\text{root of }}T_{H}$ \\[0.2ex]
    \linenumber\TAB\LET $\textit{key} = \tup{x}[\calS], \text{ where } \delta{R} = \{ \tup{x} \rightarrow m \}$ \\[0.2ex]
    \linenumber\TAB \LET $\textit{\#before} = All(key)$  \\[0.2ex]
    \linenumber\TAB $\textsc{Apply}(T_{All}, \delta{R})$ \\[0.2ex]
    \linenumber\TAB\LET $\textit{\#change} = All(key) - \textit{\#before}$ \\[0.2ex]
    \linenumber\TAB\LET $\delta(\exists{H}) = \textsc{UpdateIndTree}(T_{H}, \delta{All} = \{ key \rightarrow \textit{\#change} \,\})$ \\[0.2ex]    
    \linenumber\TAB\FOREACH $T \in \mathcal{T}$ \DO $\textsc{Apply}(T, \delta(\exists{H}))$\\[0.2ex]
    \linenumber\TAB\IF $(\textit{key} \notin \pi_{\calS}R) \,\lor\, (\textit{key} \in \pi_{\calS}R^{\calS})$ \\[0.2ex]    
    \linenumber\TAB\TAB\FOREACH $T \in \mathcal{T}$ \DO $\textsc{Apply}(T, \delta{R}^{\calS} = \delta{R})$\\[0.2ex]
    \linenumber\TAB\TAB\LET $\delta(\exists L) = \textsc{UpdateIndTree}(T_{L}, \delta{R}^{\calS} = \delta{R})$ \\[0.2ex]
    \linenumber\TAB\TAB \LET $\delta(\exists{H}) = \textsc{UpdateIndTree}(T_{H}, \delta(\nexists{L}) = -\delta(\exists{L}))$ \\[0.2ex]
    \linenumber\TAB\TAB\FOREACH $T \in \mathcal{T}$ \DO $\textsc{Apply}(T, \delta(\exists{H}))$\\[0.2ex]    
    \bottomrule
\end{tabular}
\caption{Updating a set $\calT$ of view trees and a set $\mathcal{T}_{Ind}$ of triples of indicator trees  for a single-tuple update $\delta{R}$ to relation $R$. 
}
\label{fig:update_view_trees}
\end{figure}

Figure~\ref{fig:update_indicator_tree} shows the function \textsc{UpdateIndTree} that applies an update $\delta{R}$ to an indicator tree $T_{Ind}$ with a root view $I(\calS)$. The function returns the change $\delta(\exists{I})$ in the support of the indicator view $I$, to be further propagated to other views.
The free variables $\calS$ of $I$ appear in each input relation from $T_{Ind}$, and $\delta{R}$ fixes their values to constants; thus, $|\delta(\exists{I})| \leq 1$.

Figure~\ref{fig:update_view_trees} gives our algorithm for maintaining a set of view trees $\mathcal{T}$ and a set of indicator tress $T_{Ind}$ under an update $\delta{R}$. 
We first apply $\delta{R}$ to the view trees from $\mathcal{T}$ (Line 1). Then, we consider the triples $(T_{All}, T_L, T_H)$ of indicator trees from $\mathcal{T}_{Ind}$ that are affected by $\delta{R}$. 
We maintain the heavy indicator tree $T_{H}$ with the root $H(\calS) = All(\calS),\nexists{L(\calS)}$ for changes in both $All$ and $\nexists{L}$.
We apply $\delta{R}$ to $T_{All}$ (Line 6) and subsequently $\delta{All}$ to $T_H$ (Line 8).
The latter may trigger a change $\delta(\exists H)$ in the support of $H$, which we apply to the view trees from $\mathcal{T}$ (Line 9).
If the update $\delta{R}$ belongs to the light part $R^{\calS}$ (Line 10), we apply $\delta{R}^{\calS}$ to the view trees from $\mathcal{T}$ and to the light indicator tree $T_L$ (Lines 11-12). We then propagate the opposite change $\delta(\nexists{L})$ in the support of the root $L$ of $T_L$, if any, to $T_H$ and further to the view trees from $\mathcal{T}$ (Lines 13-14).

\begin{exa}
\rm
We analyze the time needed to maintain the views from Figure~\ref{fig:view_tree_example_21} under a single-tuple update to any input relation. 
For the view tree constructed for the case of heavy $(A,B)$-values (second row right),
propagating an update from any relation to the root view takes constant time. 
For instance, an update $\delta{R}$ to $R$ changes 
the view $R'(A,B)$ with $\delta{R'}(a,b)=\delta{R}(a,b,d)$;
the view $V_B(A,B)$ with $\delta{V_B}(a,b) = \exists{H_B}(a,b),\delta{R'}(a,b),S'(a,b)$; 
changes to the views $V'_B(A)$ and $V_A(A)$ are similar.
The auxiliary views $S'(A,B)$ and $V'_C(A)$ enable the constant-time updates in this case by
aggregating away the $E$-values in $S(A,B,E)$ and the $C$-values in $V_C(A,B)$.

Consider now the view tree defined over the light parts of input relations (bottom-left).
The update $\delta{R}=\{(a,b,d)\rightarrow m\}$ affects the light part $R^A$ of $R$ when $(a,b,d)\notin R$ or $a\in \pi_A R^A$. If so, computing $\delta V_B(a,d,E) = \delta{R}^A(a,b,d),S^A(a,b,E)$ takes $\bigO{N^\eps}$ time since $a$ is light in $S^A$. The size of $\delta V_B$ is also $\bigO{N^\eps}$. 
Computing $\delta V_A$ at the root 
requires pairing each $E$-value from $\delta V_B$ with the $(C,F)$-values in $V_C$ for the given $a$. 
Since $a$ is light in $T^A$, the number of such $(C,F)$-values in $T^A$ is $\bigO{N^\eps}$.
Thus, computing $\delta V_A$ takes $\bigO{N^{2\eps}}$ time.
A similar analysis shows that updates to $S^A$ and $T^A$ also take $\bigO{N^{2\eps}}$ time, while updates to $U^A$ take $\bigO{N^{3\eps}}$ time.

For the view tree constructed for the case of heavy $A$-values (bottom-middle), updates to $R^{AB}$ and $S^{AB}$ take $\bigO{N^\eps}$ time, while updates to $T$ and $U$ take constant time.
The indicator view trees (top and middle row) encode the results of $\dfw_0$-hierarchical queries,
thus maintaining their views takes constant time per update.

The indicator views $\exists H_A(A)$ and $\exists H_B(A,B)$ may change under updates to any relation and respectively under updates to $R$ and $S$. For instance, the update $\delta R$ can trigger a new single-tuple change in $\exists H_A$ when the multiplicity $H_A(a)$ increases from 0 to non-zero or vice versa. 
Applying this change $\delta(\exists H_A)$ to the view trees containing $\exists H_A$ takes constant time; the same holds for propagating a change $\delta(\exists H_B)$ to the view trees containing $\exists H_B$.

In conclusion, maintaining the views from Figure~\ref{fig:view_tree_example_21} under a single-tuple update to any relation takes $\bigO{N^{3\eps}}$ overall time.
\qed
\end{exa}


We next state the complexity of updates in our approach.

\begin{prop}
\label{prop:update_time}
Given a hierarchical query $Q(\calF)$ with dynamic width $\dfw$, a canonical variable order $\omega$ for $Q$,
a database of size $N$, and $\eps \in [0,1]$, 
maintaining the views in the set of view trees $\tau(\omega,\calF)$ under a single-tuple update to any input relation takes $\bigO{N^{\dfw\eps}}$ time.
\end{prop}

\subsection{Rebalancing Partitions}
\label{app:rebalancing}
As the database evolves under updates, we periodically rebalance the relation partitions and views to account for a new database size and updated degrees of data values. The cost of rebalancing is amortized over a sequence of updates. 

\paragraph{Major Rebalancing.}
We loosen the partition threshold to amortize the cost of rebalancing over multiple updates. Instead of the actual database size $N$, the threshold now depends on a number $M$ for which the invariant $\floor{\frac{1}{4}M} \leq N < M$ always holds.
If the database size falls below $\lfloor \frac{1}{4} M \rfloor$ or reaches $M$, we perform {\em major rebalancing}, where we halve or respectively double $M$, followed by strictly repartitioning the light parts of input relations with the new threshold $M^{\eps}$ and recomputing the views. Figure~\ref{fig:major_rebalancing} shows the major rebalancing procedure.

\begin{figure*}[t]
\centering
\setlength{\tabcolsep}{3pt}
\renewcommand{\linenumber}{{\rownumber\TAB}}
{\setcounter{magicrownumbers}{0}
\begin{tabular}[t]{l}
    \toprule
    \textsc{MajorRebalancing}(view trees $\mathcal{T}$, indicator triples $\mathcal{T}_{Ind}$, threshold $\theta$) \\
    \midrule
    \linenumber\FOREACH $(T_{All}, T_L, T_H) \in \mathcal{T}_{Ind}$ \DO\\[0.2ex]
    \linenumber\TAB\FOREACH $R^{\mathcal{F}} \in T_L, R \in T_{All}$  \DO\\[0.2ex]
    \linenumber\TAB\TAB $R^{\mathcal{F}} = \{ \tup{x} \rightarrow R(\tup{x}) \mid \tup{x} \in R,\, key = \tup{x}[\mathcal{F}],\, |\sigma_{\mathcal{F}=key}R| < \theta \}$\\[0.2ex]
    \linenumber\TAB $\textsc{Recompute}(T_{L}),\; \textsc{Recompute}(T_{H})$ \\[0.2ex]
    \linenumber\FOREACH $T \in \mathcal{T}$ \DO $\textsc{Recompute}(T)$ \\[0.2ex]
    \bottomrule
\end{tabular}}
\caption{Recomputing the light parts of base relations and affected views.}
\label{fig:major_rebalancing}
\end{figure*}

\begin{prop}
\label{prop:major_time}
Given a hierarchical query $Q(\calF)$ with static width $\fw$, a canonical variable order 
$\omega$ for $Q$,
a database of size $N$, and $\eps \in [0,1]$, 
major rebalancing of the views in the set of view trees $\tau(\omega,\calF)$ takes $\bigO{N^{1+(\fw-1)\eps}}$ time.
\end{prop}

The cost of major rebalancing is amortized over $\Omega{(M)}$ updates. After a major rebalancing step, it holds that $N = \frac{1}{2}M$ (after doubling), or $N = \frac{1}{2}M - \frac{1}{2}$ or $N = \frac{1}{2}M - 1$ (after halving).
To violate the size invariant $\floor{\frac{1}{4}M} \leq N < M$ and trigger another major rebalancing, the number of required updates is at least $\frac{1}{4}M$.
In the extended technical report,  we prove the amortized $\bigO{N^{(\fw-1)\eps}}$ time of major rebalancing
\cite{Trade_Offs_arxiv}.
By Proposition \ref{prop:width_delta_inequal}, we have $\dfw = \fw$ or $\dfw = \fw -1$; hence, the amortized major rebalancing time is $\bigO{M^{\delta\eps}}$.

\paragraph{Minor Rebalancing.}
After an update $\delta{R} = \{ \tup{x} \rightarrow m\}$ to relation $R$, we check the light part and heavy part conditions of each partition of $R$.
Consider the light part $R^{\calS}$ of $R$ partitioned on a schema $\calS$.
If the number of tuples in $R^{\calS}$ that agree with $\tup{x}$ on $\calS$ exceeds 
$\frac{3}{2}M^\eps$, then we delete those tuples from $R^{\calS}$.
If the number of tuples that agree with $\tup{x}$ on $\calS$ in $R^{\calS}$ is zero and in $R$ is below $\frac{1}{2}M^\eps$, then we insert those tuples into $R^{\calS}$. 
Figure~\ref{fig:minor_rebalancing} shows this {\em minor rebalancing} procedure.

\begin{figure*}[t]
\centering
\setlength{\tabcolsep}{3pt}
\renewcommand{\linenumber}{{\rownumber\TAB}}
{\setcounter{magicrownumbers}{0}
\begin{tabular}[t]{l}
    \toprule
    \textsc{MinorRebalancing}(trees $\mathcal{T}$, tree $T_L$, tree $T_H$, source $R$, $key$, \textit{insert}) \\
    \midrule
    \linenumber \LET $L(\mathcal{F}) = {\text{root of }}T_{L},\; H(\mathcal{F}) = {\text{root of }}T_{H}$ \\[0.2ex]    
    \linenumber \FOREACH $\tup{x} \in \sigma_{\mathcal{F}=key}R$ \DO\\[0.2ex]
    \linenumber \TAB \LET $cnt = $ \IF $(\textit{insert\,})\ R(\tup{x})$ \ELSE $-R(\tup{x})$\\[0.2ex]
    \linenumber \TAB\FOREACH $T \in \mathcal{T}$ \DO $\textsc{Apply}(T, \delta{R^{\mathcal{F}}} = \{ \tup{x} \rightarrow cnt \})$\\[0.2ex]    
    \linenumber \TAB \LET $\delta(\exists{L}) = \textsc{UpdateIndTree}(T_L, \delta{R^{\mathcal{F}}} = \{ \tup{x} \rightarrow cnt \})$ \\[0.2ex]
    \linenumber \TAB \LET $\delta(\exists{H}) = \textsc{UpdateIndTree}(T_H, \delta(\nexists{L}) = -\delta(\exists{L}))$ \\[0.2ex]
    \linenumber \TAB\FOREACH $T \in \mathcal{T}$ \DO $\textsc{Apply}(T, \delta(\exists{H}))$\\[0.2ex]
    \bottomrule
\end{tabular}}
\caption{Deleting heavy tuples from or inserting light tuples into the light part of relation $R$.}
\label{fig:minor_rebalancing}
\end{figure*}

\begin{prop}
\label{prop:minor_time}
Given a hierarchical query $Q(\calF)$ with dynamic width $\dfw$, a canonical variable order 
$\omega$ for $Q$,
a database of size $N$, and $\eps \in [0,1]$, 
minor rebalancing of the views in the set of view trees $\tau(\omega,\calF)$ takes $\bigO{N^{(\dfw+1)\eps}}$ time.
\end{prop}

The cost of minor rebalancing is amortized over $\Omega(M^{\eps})$ updates. 
This lower bound on the number of updates is due to the gap between the two thresholds in the heavy and light part conditions. The extended technical report 
proves the amortized $\bigO{N^{\dfw\eps}}$ time of minor rebalancing \cite{Trade_Offs_arxiv}.

\begin{figure*}[t]
\centering
\setlength{\tabcolsep}{3pt}
\renewcommand{\linenumber}{\makebox[3ex][r]{\rownumber\TAB}}
\setcounter{magicrownumbers}{0}
\begin{tabular}[t]{l}
    \toprule
    \textsc{OnUpdate}(view trees $\mathcal{T}$, indicator triples $\mathcal{T}_{Ind}$, update $\delta{R}$) \\
    \midrule
    \linenumber$\textsc{UpdateTrees}(\mathcal{T}, \mathcal{T}_{Ind}, \delta{R})$ \\[0.2ex]
    \linenumber\IF $(N = M)$\\[0.2ex]
    \linenumber\TAB $M = 2M$\\[0.2ex]
    \linenumber\TAB $\textsc{MajorRebalancing}(\mathcal{T}, \mathcal{T}_{Ind}, M^\eps)$\\[0.2ex]
    \linenumber\ELSE \IF $(N < \left\lfloor \frac{1}{4}M \right\rfloor)$\\[0.2ex]
    \linenumber\TAB $M = \left\lfloor \frac{1}{2}M \right\rfloor - 1$\\[0.2ex]
    \linenumber\TAB $\textsc{MajorRebalancing}(\mathcal{T}, \mathcal{T}_{Ind}, M^\eps)$\\[0.2ex]
    \linenumber\ELSE \\[0.2ex]

    \linenumber\TAB\FOREACH $(T_{All}, T_{L}, T_{H}) \in \mathcal{T}_{Ind}$ such that $R \in T_{All}$ \DO \\[0.2ex]
    \linenumber\TAB\TAB\LET $R^{\mathcal{F}} \in T_L \text{ be light part of $R$ partitioned on $\mathcal{F}$ }$\\[0.2ex]
    \linenumber\TAB\TAB\LET $\textit{key} = \tup{x}[\mathcal{F}], \text{ where } \delta{R} = \{ \tup{x} \rightarrow m \}$ \\[0.2ex]
    \linenumber\TAB\TAB \IF $(\,|\sigma_{\mathcal{F}=key}R^{\mathcal{F}}| = 0 \,\land\, \,|\sigma_{\mathcal{F}=key}R| < \frac{1}{2}M^\eps)$\\[0.2ex]
    \linenumber \TAB\TAB\TAB $\textsc{MinorRebalancing}(\mathcal{T}, T_L, T_H, R, key, \textsf{true})$\\[0.2ex]
    \linenumber\TAB\TAB\ELSE\IF $(\,|\sigma_{\mathcal{F}=key}R^{\mathcal{F}}| \geq \frac{3}{2}M^\eps\,)$\\[0.2ex]
    \linenumber \TAB\TAB\TAB $\textsc{MinorRebalancing}(\mathcal{T}, T_L, T_H, R, key, \textsf{false})$\\[0.2ex]
    \bottomrule
\end{tabular}
\caption{
Updating a set of view trees $\mathcal{T}$ and a set of triplets of indicator view trees $\mathcal{T}_{Ind}$ under a sequence of single-tuple updates to base relations.}
\label{fig:on_update}
\end{figure*}

\medskip
Figure~\ref{fig:on_update} gives the trigger procedure \textsc{OnUpdate} that maintains a set of view trees $\mathcal{T}$ and a set of indicator trees $\mathcal{T}_{Ind}$ under a sequence of single-tuple updates to input relations. We first apply an update $\delta{R}$ to the view trees from $\mathcal{T}$ and indicator trees from $\mathcal{T}_{Ind}$ using \textsc{UpdateTrees} from Figure~\ref{fig:update_view_trees}. If this update leads to a violation of the size invariant $\floor{\frac{1}{4}M} \leq N < M$, we invoke \textsc{MajorRebalancing} to recompute the light parts of the input relations and affected views. 
Otherwise, for each triple of indicator trees from $\mathcal{T}_{Ind}$ with the light part $R^{\mathcal{F}}$ partitioned on $\mathcal{F}$, 
we check if the heavy or light condition is violated;
if so, we invoke \textsc{MinorRebalancing} to move the $R$-tuples having the $\mathcal{F}$-values of the update $\delta{R}$ either into or from the light part $R^{\mathcal{F}}$ of relation $R$.

We state the amortized maintenance time of our approach under a sequence of single-tuple updates.

\begin{prop}
\label{prop:amortized_update_time}
Given a hierarchical query $Q(\calF)$ with dynamic width $\dfw$, a canonical variable order $\omega$ for $Q$, a database of size $N$, and $\eps \in [0,1]$, maintaining the views in the set of view trees $\tau(\omega,\calF)$ under a sequence of single-tuple updates
 takes $\bigO{N^{\dfw\eps}}$ amortized time per single-tuple update.
\end{prop}

The proof of Proposition~\ref{prop:amortized_update_time} is based on prior work (Section 4.1 in~\cite{Kara:ICDT:19}).
The adapted proof is given in the technical report (Section F.4 in \cite{Trade_Offs_arxiv}).

\section{Matching Lower Bound for $\dfw_1$-Hierarchical Queries}
\label{sec:lower_bound}
Corollary~\ref{theo:delta_hierarchical_queries} says that, 
given a database of size $N$ and $\eps \in [0,1]$,
any $\dfw_i$-hierarchical query with $i \in {\mathbb{N}}$ 
can be evaluated with $\bigO{N^{i\eps}}$ amortized update time and $\bigO{N^{1-\eps}}$ enumeration delay.  
For $\dfw_1$-hierarchical queries, this upper bound 
 is matched by a lower bound conditioned on the Online Matrix-Vector Multiplication Conjecture~\cite{Henzinger:OMv:2015}.
The following proposition extends the lower bound result from prior work~\cite{BerkholzKS18} to amortized update time. The adapted proof can be found in the technical report~(Proposition 10 in \cite{Trade_Offs_arxiv}).

\begin{prop}
\label{prop:lower_bound}
Given a $\dfw_1$-hierarchical query without repeating relation symbols, $\gamma >0$, 
and a database of size $N$, there is no algorithm that maintains the query with arbitrary preprocessing  time, $\bigO{N^{\frac{1}{2} - \gamma}}$ amortized update time, and $\bigO{N^{\frac{1}{2} - \gamma}}$ enumeration delay, unless the {Online Matrix-Vector Multiplication}
conjecture fails.
\end{prop}

\begin{wrapfigure}{l}{0.365\textwidth}
  \vspace{-0.36cm} 
   \begin{center}
   \tdplotsetmaincoords{73}{170}
  
\scalebox{0.8}{
    \begin{tikzpicture}[xscale=1.6, yscale=1.2,tdplot_main_coords]
  
      \coordinate (O) at (0,0,0);


      \draw[thick,dotted,color=gray] (0,0,0) -- (0.5,0,0);
      \draw[thick,->] (0.5,0,0) -- (1.6,0,0) node[anchor=north]{\small $\log_{N}${delay}};
      \draw[thick,dotted,color=gray] (0,0,0) -- (0,4.5,0);
      \draw[thick,->] (0,4.5,0) -- (0,5,0) node[anchor=north]{\small $\log_{N}${preprocessing time}};
      \draw[thick,dotted,color=gray] (0,0,0) -- (0,0,0.5);
      \draw[thick,->] (0,0,0.5) -- (0,0,1.3) node[anchor=south]{\small $\log_{N}${update time}};

      \node[color=black] at (-0.1,-0.1,0) () {\small $0$};
      \coordinate (P1) at (1,1,0);
      \node at (1,1,-0.2) () {\small $(1,0,1)$};
      \coordinate (P1x) at (1,0,0);
      \node at (0.95,-0.6,0.0) () {\small $1$};

    \node[] at (-0.5,1,0.8) () {\small $(2,1,0)$};
    
      \coordinate (P1y) at (0,1,0);

      \coordinate (P2) at (0,2,1);
      \coordinate (P2xy) at (0,2,0);
      \coordinate (P2y) at (0,2,0);
      
      \coordinate (P2z) at (0,0,1);

      \coordinate (P3) at (0,2,2);
      \coordinate (P3z) at (0,0,2);

      \draw[dotted, color=black, line width = 0.7pt] (P1x) -- (P1);
      \draw[dotted, color=black, line width = 0.7pt] (P1y) -- (P1);

      \node[color=black] at (-0.1,0.5,0) () {\footnotesize $1$};
      \node[color=black] at (-0.18,1.5,-0.1) () {\footnotesize $2$};

      \node[] at (0.25,0,1.05) () {\footnotesize $\dfw = 1$};

      \draw[dotted, color=black,line width = 0.7pt] (P2) -- (P2y);
      \draw[dotted, color=black,line width = 0.7pt] (P2) -- (P2z);

      \coordinate (X_half) at (0.5,0,0);
      \coordinate (Z_half) at (0,0,0.5);
      \coordinate (XZ_half) at (0.5,0,0.5);
      \fill[gray!60, opacity=0.4] (XZ_half) -- (Z_half) -- (0,4.5,0.5) -- (0.5,4.5,0.5) -- cycle;
      \fill[gray!60, opacity=0.4] (XZ_half) -- (X_half) -- (0.5,4.5,0) -- (0.5,4.5,0.5) -- cycle;
      \fill[gray!60, opacity=0.4] (0.5,4.5,0.5) -- (0.5,4.5,0) -- (0,4.5,0) -- (0,4.5,0.5) -- cycle;
      \draw[color=gray] (0.5,0,0) -- (0.5,4.5,0);
      \draw[color=gray] (0.5,0,0) -- (0.5,0,0.5);
      \draw[color=gray] (0.5,0,0.5) -- (0,0,0.5);
      \draw[color=gray] (0.5,0,0.5) -- (0.5,4.5,0.5);
      \draw[color=gray] (0.5,4.5,0.5) -- (0.5,4.5,0);
      \draw[color=gray] (0.5,4.5,0.5) -- (0,4.5,0.5);
      \draw[color=gray] (0.5,4.5,0) -- (0,4.5,0);
      \draw[color=gray] (0,4.5,0.5) -- (0,0,0.5);
      \draw[color=gray] (0,4.5,0) -- (0,4.5,0.5);

      \node at (0.5,-0.7,0) () {\small$\frac{1}{2}$};
      \node at (0.05,-0.25,0.6) () {\small$\frac{1}{2}$};
        
        \draw[dotted, color=gray, line width = 0.7pt] (0.5,1.5,0.5) -- (0.5,1.5,0);
        \draw[dotted, color=gray,  line width = 0.7pt] (0.5,1.5,0) -- (0.0,1.5,0);
        
        \draw[color=blue, thick] (P1) -- (P2);
        \draw[->, thick, color=black, dotted] (0.8,1.0,-0.5) -- (0.51,1.55,0.45);
        \node at (0.8,1.0,-0.75) () {\color{blue}\small $(\frac{3}{2},\frac{1}{2},\frac{1}{2})$};
        
        \node at (01.1,1.1,-1.05) () {\color{blue} \small weakly Pareto optimal};
        
    \end{tikzpicture}
    }
       \end{center}
  \label{fig:optimality}
  \vspace*{-1em}
\end{wrapfigure}
The blue line connecting the points $(1,0,1)$ and $(2,1,0)$ in the left figure 
visualizes the trade-offs 
of our approach for $\dfw_1$-hierarchical queries. 
The gray cuboid is infinite in the dimension of preprocessing time.
Each point strictly included in the gray cuboid corresponds  
to a combination of some preprocessing time and
$\bigO{N^{\frac{1}{2}-\gamma}}$ amortized update time
and delay for $\gamma >0$. 
Following Proposition~\ref{prop:lower_bound}, this is not attainable, 
unless the {Online Matrix-Vector Multiplication} conjecture fails.
Each point on the surface of the cuboid corresponds to Pareto 
worst-case optimality in the update-delay trade-off space. 
For $\eps = \frac{1}{2}$,
our approach needs $\bigO{N^{\frac{1}{2}}}$
amortized update time and delay, which is 
weakly Pareto worst-case optimal: there can be no 
tighter upper bounds for {\em both} the update time
and delay. Since $\fw\in\{1,2\}$ for $\dfw_1$-hierarchical queries, 
the preprocessing time is $O(N^{\frac{3}{2}})$.

\section{{Examples Showcasing Our Approach}}
\label{sec:examples}
We exemplify our approach for the static and dynamic evaluation of two $\dfw_1$-hierarchical queries.
We start with the query from Example~\ref{ex:static_overall_time}.

\begin{exa}\label{ex:intro}
\rm
Consider the $\dfw_1$-hierarchical and non-free-connex query $Q(A,C)$ $=$ $R(A,B),$ $S(B,C)$
from Example~\ref{ex:static_overall_time} whose relations have size at most $N$. 
We partition $R$ and $S$ on $B$: A $B$-value $b$ is {\em light} in $R$ if $|\{ a \mid (a,b)\in R\}| \leq N^\eps$ and {\em heavy} otherwise (similar for $S$). Since each heavy $B$-value is paired with at least $N^\eps$ $A$-values in $R$, there are at most $N^{1-\eps}$ heavy $B$-values. 
There are four cases to consider: $B$ is either light or heavy in each of $R$ and $S$. We can reduce them to two cases: either $B$ is light in both relations, or $B$ is heavy in at least one of them. We keep the light/heavy information in two indicator views: $L_B(B) = R^B(A,B),S^B(B,C)$, where $R^B$ and $S^B$ are the light parts of $R$ and respectively $S$; and $H_B(B) = All_B(B),\nexists L_B(B)$, where $All_B(B) = R(A,B),S(B,C)$.
The $\exists$ operator before indicators denotes their use with set semantics, i.e.,  
the tuple multiplicities are $0$ or $1$. The $\nexists$ operator flips the multiplicity.

Figure~\ref{fig:intro-example} gives the evaluation and maintenance strategies for our query. A strategy is depicted by a view tree, with one view per node such that the head of the view is depicted at the node and its body is the join of its children.

To support light/heavy partitions, we need to keep the degree information of the $B$-values in the two relations. The light/heavy indicators can be computed in linear time, e.g., for $L_B$ we start with the light parts of $R$ and $S$, aggregate away $A$ and respectively $C$ and then join them on $B$.

If $B$ is light, we compute the view $V_B(A,C)$ in time $\bigO{N^{1+\eps}}$: We iterate over $S^B$ and for each of its tuples $(b,c)$, we  fetch the $A$-values in $R^B$ paired with $b$ in $R$. The iteration over $S^B$ takes linear time and for each $b$ there are at most $N^\eps$ $A$-values in $R$. The view $V_B(A,C)$ is a subset of $Q$'s result. 

\begin{figure}[t]
  \centering
  \begin{minipage}[b]{0.32\linewidth}
    \small
    \centering
    \begin{tikzpicture}[xscale=1.4, yscale=1]
      \node at (0.0, 0.0) (All_{B}) {{\color{black} $All_B({ B})$}};
      \node at (-0.6, -0.8) (V_{B}) {{\color{black} $All_{A}({ B})$}} edge[-] (All_{B});
      \node at (0.6, -0.8) (V_{C}) {{\color{black} $All_{C}({ B})$}} edge[-] (All_{B});
      \node at (-0.6, -1.6) (R) {{\color{black} $R(\underline{A},B)$}} edge[-] (V_{B});
      \node at (0.6, -1.6) (S) {{\color{black} $S(B,\underline{C})$}} edge[-] (V_{C});
    \end{tikzpicture}
  \end{minipage}
  \begin{minipage}[b]{0.32\linewidth}
    \small
    \centering
    \begin{tikzpicture}[xscale=1.4, yscale=1]
      \node at (0.0, 0.0) (L_{B}) {{\color{black} $L_{B}({ B})$}};
      \node at (-0.6, -0.8) (V_{B}) {{\color{black} $L_{A}({ B})$}} edge[-] (L_{B});
      \node at (0.6, -0.8) (V_{C}) {{\color{black} $L_{C}({ B})$}} edge[-] (L_{B});
      \node at (-0.6, -1.6) (R_l) {{\color{black} $R^B(\underline{A},B)$}} edge[-] (V_{B});
      \node at (0.6, -1.6) (S_l) {{\color{black} $S^B(B,\underline{C})$}} edge[-] (V_{C});
    \end{tikzpicture}
  \end{minipage}
  \begin{minipage}[b]{0.32\linewidth}
    \small
    \centering
    \begin{tikzpicture}[xscale=1.4, yscale=1]
      \node at (0.0, 0.0) (H_{B}) {{\color{black} $H_{B}({ B})$}};
      \node at (-0.6, -1.6) (All_{B}) {{\color{black} $All_B({ B})$}} edge[-] (H_{B});
      \node at (0.6, -1.6) (L_{B}) {{\color{black} $\nexists L_{B}({ B})$}} edge[-] (H_{B});
    \end{tikzpicture}
  \end{minipage}

  \medskip


  \begin{minipage}[b]{0.49\linewidth}
    \centering
    \begin{tikzpicture}[xscale=1, yscale=1]
      \node at (0.0, 0.0) (V_{B}) {{\color{black} $V_{B}({ \underline{A}, \underline{C} })$}};
      \node at (-1.2, -1.8) (R) {{\color{black} $\color{black} R^B(\underline{A}, B)$}} edge[-] (V_{B});
      \node at (1.2, -1.8) (S) {{\color{black} $\color{black} S^B(B, \underline{C})$}} edge[-] (V_{B});
    \end{tikzpicture}
  \end{minipage}
  \begin{minipage}[b]{0.49\linewidth}
    \small
    \centering
    \begin{tikzpicture}[xscale=1, yscale=1]
      \node at (0.0, 0.0) (V_{B}) {{\color{black} $V_{B}({ {\color{black} B}})$}};
      \node at (-1.3, -1.0) (H_{B}) {{\color{black} $\exists H_{B}({ B})$}} edge[-] (V_{B});
      \draw[black, dashed] (-0.55,-0.75) rectangle (0.55,-1.25);
      \node at (0.0, -1.0) (V_{A}) {{\color{black} $R'({ B})$}} edge[-] (V_{B});
      \node at (0.0, -1.8) (R) {{\color{black} $R(\underline{A}, B)$}} edge[-] (V_{A});
      \draw[black, dashed] (0.75,-0.75) rectangle (1.85,-1.25);
      \node at (1.3, -1.0) (V_{C}) {{\color{black} $S'({ B})$}} edge[-] (V_{B});
      \node at (1.3, -1.8) (S) {{\color{black} $S(B, \underline{C})$}} edge[-] (V_{C});
    \end{tikzpicture}
  \end{minipage}
  \caption{The view trees for $Q(A,C) = R(A,B), S(B,C)$ in Example~\ref{ex:intro}. The dashed boxes enclose views that are only needed in the dynamic case.}
  \label{fig:intro-example}
\end{figure}

If $B$ is heavy, we construct the view $V_B(B)$ with up to $N^{1-\eps}$ heavy $B$-values. For each such value $b$, we can enumerate the distinct tuples $(a,c)$ such that $R(a,b)$ and $S(b,c)$ hold. Distinct $B$-values may, however, have the same tuple $(a,c)$. Therefore, if we were to enumerate such tuples for one $B$-value after those for another $B$-value, the same tuple $(a,c)$ may be output several times, which violates the enumeration constraint.
To address this challenge, we use the union algorithm~\cite{Durand:CSL:11}. We use the $N^{1-\eps}$ buckets of $(a,c)$ tuples, one for each heavy $B$-value, and an extra bucket $V_B(A,C)$ constructed in the light case. From each bucket of a $B$-value, we can enumerate the distinct $(a,c)$ tuples with constant delay by looking up into $R$ and $S$. The tuples in the materialized view $V_B(A,C)$ can be enumerated with constant delay. We then use the union algorithm to enumerate the distinct $(a,c)$ tuples with delay given by the sum of the delays of the buckets. For each such tuple, we sum up the positive multiplicities of its occurrences in the buckets. This yields an overall $\bigO{N^{1-\eps}}$ delay for the enumeration of the distinct tuples in the result of $Q$.

We now turn to the dynamic case. The preprocessing time and delay remain the same as in the static case, while each single-tuple update can be processed in $\bigO{N^\eps}$ amortized time.
To support updates, we need to maintain tuple multiplicities in addition to the degree information of the $B$-values in the two relations. The multiplicity of a result tuple is the sum of the multiplicities of its duplicates across the  $\bigO{N^{1-\eps}}$ buckets.
We also need two views to support efficient updates to $R$ and $S$; these are marked with the dashed boxes in Figure~\ref{fig:intro-example}.
For simplicity, we assume that each view tree maintains copies of its base relations. 

Consider a single-tuple update $\delta R = \{(a,b) \rightarrow m \}$ to relation $R$. 
We maintain each view affected by $\delta R$ using the hierarchy of materialized views from Figure~\ref{fig:intro-example}.
The changes in those views are expressed using the classical delta rules~\cite{Chirkova:Views:2012:FTD}.
We update the views $R'(B)$ and $V_B(B)$ in the bottom-right tree with $\delta R'(b) = \delta R(a,b)$ and $\delta V_B(b) = \exists H_B(b), \delta R'(b), S'(b)$ in constant time; the same holds for updating the views $All_A(B)$, $All_B(B)$, and $H_B(B)$.

The update $\delta R$ affects the light part $R^B$ of $R$ if the $B$-value $b$ already exists among the $B$-values in $R^B$ or does not exist in $R$.
For such change $\delta R^B$, we update $V_B(A,C)$ with $\delta V_B(a,C)$ $= \delta R^B(a,b), S^B(b,C)$ in time $\bigO{N^\eps}$ since $b$ is light in $S^B$; updating $L_B(B)$ and $H_B(B)$ takes constant time.

The update $\delta R$ may trigger a new single-tuple change in $\exists H_B$, affecting $V_B(B)$. 
The change $\delta (\exists H_B)$ is non-empty only when the multiplicity $H_B(b)$ changes from 0 to non-zero or vice versa.
For such change $\delta (\exists H_B)$, we update $V_B(B)$ via constant-time lookups in $R'(b)$ and $S'(b)$.

The update $\delta R$ may change the degree of $b$ in $R$ from light to heavy or vice versa. 
In such cases, we need to rebalance the partitioning of $R$ and possibly recompute some of the views.
Although such rebalancing steps may take time more than $\bigO{N^\eps}$, they happen periodically and their amortized cost remains the same as for a single-tuple update (Section~\ref{sec:updates}).\qed
\end{exa}

\begin{figure}[t]
  \centering

  \begin{minipage}[b]{0.32\linewidth}
    \small
    \centering
    \begin{tikzpicture}[xscale=1, yscale=1]
      \node at (0.0, 0.0) (All_{B}) {{\color{black} $All_{B}({ B})$}};
      \node at (-0.6, -0.8) (V_{B}) {{\color{black} $All_{A}({ B})$}} edge[-] (All_{B});
      \node at (-0.6, -1.6) (R) {{\color{black} $R(\underline{A},B)$}} edge[-] (V_{B});
      \node at (0.6, -1.6) (S) {{\color{black} $S(B)$}} edge[-] (All_{B});
    \end{tikzpicture}
  \end{minipage}
  \begin{minipage}[b]{0.32\linewidth}
    \small
    \centering
    \begin{tikzpicture}[xscale=1.4, yscale=1]
      \node at (0.0, 0.0) (L_{B}) {{\color{black} $L_{B}({ B})$}};
      \node at (-0.6, -0.8) (V_{B}) {{\color{black} $L_{A}({ B})$}} edge[-] (L_{B});
      \node at (-0.6, -1.6) (R_l) {{\color{black} $R^{{B}}(\underline{A},B)$}} edge[-] (V_{B});
      \node at (0.6, -1.6) (S_l) {{\color{black} $S^{{B}}(B)$}} edge[-] (L_{B});
    \end{tikzpicture}
  \end{minipage}
  \begin{minipage}[b]{0.32\linewidth}
    \small
    \centering
    \begin{tikzpicture}[xscale=1.4, yscale=1]
      \node at (0.0, 0.0) (H_{B}) {{\color{black} $H_{B}({ B})$}};
      \node at (-0.6, -1.6) (All_{B}) {{\color{black} $All_B({ B})$}} edge[-] (H_{B});
      \node at (0.6, -1.6) (L_{B}) {{\color{black} $\nexists L_{B}({ B})$}} edge[-] (H_{B});
    \end{tikzpicture}
  \end{minipage}

  \medskip

  \begin{minipage}[b]{0.32\linewidth}
    \small
    \centering    
    \begin{tikzpicture}[xscale=1.4, yscale=1]
      \node at (0.0, 0.0) (All_{B}) {{\color{black} $V_{B}({ A})$}};
      \node at (-0.6, -1.8) (R) {{\color{black} $R(\underline{A},B)$}} edge[-] (All_{B});
      \node at (0.6, -1.8) (S) {{\color{black} $S(B)$}} edge[-] (All_{B});
      
    \end{tikzpicture}
  \end{minipage}
  \begin{minipage}[b]{0.32\linewidth}
    \small
    \centering
    \hspace*{-0.35cm}
    \begin{tikzpicture}[xscale=1.4, yscale=1]
      \node at (0.0, 0.0) (V_{B}) {{\color{black} $V_{B}({ \underline{A}})$}};
      \node at (-0.6, -1.8) (R) {{\color{black} $R^{B}(\underline{A},B)$}} edge[-] (V_{B});
      \node at (0.6, -1.8) (S) {{\color{black} $S^{B}(B)$}} edge[-] (V_{B});
      
    \end{tikzpicture}
  \end{minipage}
  \begin{minipage}[b]{0.32\linewidth}
    \small
    \centering
    \hspace*{-0.3cm}
    \begin{tikzpicture}[xscale=1.4, yscale=1]
      \node at (0.0, 0.0) (V_{B}) {{\color{black} $V_{B}({ {\color{black} B}})$}};
      \node at (-1.1, -1.0) (H_{B}) {{\color{black} $\exists H_{B}({ B})$}} edge[-] (V_{B});
      \node at (0.0, -1.0) (V_{A}) {{\color{black} $R'({ B})$}} edge[-] (V_{B});
      \node at (0.0, -1.8) (R) {{\color{black} $R(\underline{A},B)$}} edge[-] (V_{A});
      \node at (1.1, -1.8) (S) {{\color{black} $S(B)$}} edge[-] (V_{B});
    \end{tikzpicture}
  \end{minipage}
  
  \caption{The view trees for $Q(A) = R(A,B), S(B)$ in Example~\ref{ex:intro2}. The bottom-left view tree is the only one needed in the static case, all others are needed in the dynamic case.}
  \label{fig:intro-example2}
\end{figure}

Next, we demonstrate our approach for the $\dfw_1$-hierarchical query
from Example~\ref{ex:dynamic_overall_time}.
 
\begin{exa}\label{ex:intro2}
\rm
Consider the $\dfw_1$-hierarchical free-connex query $Q(A) = R(A,B), S(B)$
from Example~\ref{ex:dynamic_overall_time}  whose relations have size at most $N$. 
Figure~\ref{fig:intro-example2} shows the single view tree (bottom-left) that our approach constructs in the static case, and the other five view trees needed in the dynamic case. 
In the static case, since $Q$ is free-connex, its result  can be computed in $\bigO{N}$ time and then its tuples can be enumerated with $\bigO{1}$ delay. Our approach does not partition the relations in the static case. 
We compute the view $V_B(A)$ in time $\bigO{N}$ by iterating over the tuples in $R$ and looking up for each tuple $(a,b)$ in $R$ the multiplicity of $b$ in $S$ in $\bigO{1}$ time. The result can be enumerated from the view $V_B(A)$ with $\bigO{1}$ delay. 

In the dynamic case, we partition relations $R$ and $S$ on the bound join variable $B$ and create the indicators $L_B(B)$ and $H_B(B)$ as in Figure~\ref{fig:intro-example2}. 
In the light case, we compute the view $V_B(A)$ in $\bigO{N}$ time: For each $(a,b)$ in the light part $R^B$ of $R$, we check the multiplicity of $b$ in the light part $S^B$ of $S$ using a constant-time lookup.
In the heavy case, we compute the view $V_B(B)$ in $\bigO{N}$ time using the heavy indicator $\exists H_B$, the input relation $S$, and the projection $R'(B)$ of $R$ on $B$.

We can enumerate the tuples in the query result with $\bigO{N^{1-\eps}}$ delay: Since  there are at most $N^{1-\eps}$ heavy $B$-values in $V_B(B)$, each with its own list of $A$-values in $R$, we need $\bigO{N^{1-\eps}}$ delay to enumerate the distinct $A$-values paired with the heavy $B$-values. In addition, we can enumerate from the view $V_B(A)$ created for the light $B$-values with constant delay. To obtain the multiplicity of each output tuple, we sum up the positive multiplicities of the duplicates of the tuple across the  $\bigO{N^{1-\eps}}$ buckets.

A single-tuple update to $R$ triggers constant-time updates to all views. A single-tuple update to $S$ triggers constant-time updates to the indicators and $V_B(B)$. In the light case, the update to $V_B(A)$ is given by $\delta V_B(A) = R^B(A,b), \delta S^B(b)$, which 
requires $\bigO{N^\eps}$ time since $b$ is light in $R^B$. We may need to rebalance the partitions, which gives an amortized update time of $\bigO{N^\eps}$. 
\qed
\end{exa}

Since both queries in Examples~\ref{ex:intro} and \ref{ex:intro2} are $\dfw_1$-hierarchi\-cal and do not have repeating relation symbols, there is no algorithm that can maintain them under single-tuple updates with $\bigO{N^{\frac{1}{2}-\gamma}}$ amortized update time and $\bigO{N^{\frac{1}{2}-\gamma}}$ delay  for $\gamma>0$ unless the {Online Matrix-Vector Multiplication} conjecture fails (Proposition~\ref{prop:lower_bound}). Our approach meets this lower bound for $\eps=\frac{1}{2}$.

\section{Conclusion and Future Work}
\label{sec:conclusion}
This paper investigates the preprocessing-update-delay trade-off for hierarchical queries and introduces an approach that recovers a number of prior results when restricted to hierarchical queries.
There are several lines of future work. Of paramount importance is the generalization of our trade-off from hierarchical to conjunctive queries. The results of this paper can be immediately extended to hierarchical queries with group-by aggregates and order-by clauses. In particular, this extension would capture the prior result on constant-delay enumeration for such queries in the context of factorized databases~\cite{OlteanuS:SIGREC16}. An open problem is to find lower bounds for $\dfw_i$-hierarchical queries for $i>1$. We conjecture our update/delay upper bounds $\bigO{N^{i\eps}}/ \bigO{N^{1-\eps}}$ are worst-case optimal,  as it is the case for $i=0$ with $\eps=1$~\cite{BerkholzKS17} and $i=1$ with $\eps=\frac{1}{2}$ (Proposition~\ref{prop:lower_bound}).

\section*{Acknowledgment}
  \noindent This project has received funding from the European Union’s Horizon 2020 research and innovation programme under grant agreement No 682588.

  
\bibliographystyle{alphaurl}
\bibliography{bibliography}

\appendix
\section{Proofs of the Results in Section~\ref{sec:trade_offs_static_evaluation}}  
\label{appendix:main_results_static}
\begin{citedthm}[\ref{thm:main_static}]
  Given a hierarchical query with static width $\fw$, a database  
  of size $N$, and $\eps \in [0,1]$, the query result can be enumerated with $\bigO{N^{1-\eps}}$ delay after $\bigO{N^{1 + (\fw -1)\eps}}$ 
  preprocessing time.
\end{citedthm}

The theorem follows from Propositions~\ref{prop:equivalence}, \ref{prop:preproc_time}, and \ref{prop:enumeration}. 
Let $Q(\calF)$ be a hierarchical query
and $\omega$ an arbitrary canonical 
variable order for $Q$. 
Without loss of generality, assume that $\omega$ 
consists of a single tree.    
The preprocessing stage materializes the views in the view trees $\{T_1, \ldots , T_k\}$ returned by 
$\tau(\omega, \calF)$ from Figure~\ref{fig:view_forest_main_algo} in the static mode. 
By Proposition~\ref{prop:preproc_time}, these views can be materialized 
in $\bigO{N^{1 + (\fw -1)\eps}}$. 
By Proposition~\ref{prop:equivalence},
$Q(\mathcal{F})$ is equivalent to $\bigcup_{i\in[k]} Q_i(\mathcal{F})$, where 
$Q_i(\mathcal{F})$ is the query defined by the 
join of the leaves in $T_i$.
By Proposition~\ref{prop:enumeration},  
the result of $Q(\calF)$ can be enumerated from these materialized views with delay $\bigO{N^{1-\eps}}$. 

If the canonical variable order for $Q$ consists of several trees 
$\omega_1, \dots , \omega_m$, we 
construct a set $\calT_i$ of view trees for each $\omega_i$, where $i \in [m]$. 
The result of the query is the Cartesian product of the tuple sets obtained from each $\calT_i$.
Given that each set $\calT_i$ of view trees admits 
$\bigO{N^{1-\eps}}$ enumeration delay,  
the tuples in the Cartesian product can be enumerated with the same delay
using the \textsc{Product} algorithm (Figure~\ref{fig:enumeration-product}), 
since $m$ is independent of the database size $N$.

\section{Proofs of the Results in Section~\ref{sec:trade_offs_dynamic_evaluation}}  
\label{appendix:main_results_dynamic}
\begin{citedthm}[\ref{thm:main_dynamic}]
  Given a hierarchical query with static width $\fw$ and dynamic width $\dfw$, a database of size $N$, and $\eps \in [0,1]$, 
  the query result can be enumerated with $\bigO{N^{1-\eps}}$ delay after $\bigO{N^{1 + (\fw -1)\eps}}$ preprocessing time and 
  $\bigO{N^{\dfw\eps}}$ amortized update time for single-tuple updates.
\end{citedthm}

The theorem follows from 
Propositions~\ref{prop:equivalence}, \ref{prop:preproc_time}, \ref{prop:enumeration}, \ref{prop:update_time}, and
\ref{prop:amortized_update_time}. 
Let $Q(\calF)$ be a hierarchical query
and $\omega$ an arbitrary canonical variable order
for $Q$.
Without loss of generality, assume that 
$\omega$ consists of a single tree. 
The preprocessing stage materializes 
the views in the set of view trees returned by $\tau(\omega,\calF)$ from Figure~\ref{fig:view_forest_main_algo} in the dynamic mode.
The preprocessing time $\bigO{N^{1 + (\fw -1)\eps}}$ follows from Proposition~\ref{prop:preproc_time}, 
which captures both the static and dynamic modes (see the proof). 
The equivalence between the constructed view trees and the query 
follows from Proposition~\ref{prop:equivalence}.
The delay $\bigO{N^{1-\eps}}$ needed when enumerating the query result 
 from these materialized views
follows from Proposition~\ref{prop:enumeration}.
The time $\bigO{N^{\dfw\eps}}$
to maintain these materialized views
under a single-tuple update 
follows from Proposition~\ref{prop:update_time}.
By Proposition~\ref{prop:amortized_update_time},
the amortized maintenance time 
under a sequence of single-tuple updates is 
$\bigO{N^{\dfw\eps}}$.
  
If the canonical variable order consists of several view trees,
the reasoning is analogous to the proof of Theorem~\ref{thm:main_static}.

\section{Proofs of the Results in Section~\ref{sec:preprocessing}}
\label{appendix:preprocessing}
\subsection{Proof of Proposition~\ref{prop:equivalence}}

\begin{citedprop}[\ref{prop:equivalence}]
Let $\{T_1,\ldots,T_k\} = \tau(\omega,\calF)$ 
be the set of view trees constructed by the algorithm in Figure~\ref{fig:view_forest_main_algo} for a given hierarchical query $Q(\mathcal{F})$ and a canonical variable order $\omega$ for $Q$.
Let $Q^{(i)}(\mathcal{F})$ be the query defined by the conjunction of the leaf atoms in $T_i$, $\forall i\in[k]$. Then, $Q(\mathcal{F}) \equiv \bigcup_{i\in[k]} Q^{(i)}(\mathcal{F})$.
\end{citedprop}

We use two observations.
(1) The procedure $\textsc{BuildVT}$ constructs a view tree whose leaf atoms are exactly the same as the leaf atoms of the input variable order.
(2) Each of the procedures $\textsc{NewVT}$ and $\textsc{AuxView}$ 
constructs a view tree whose set of leaf atoms is the union of the sets of leaf atoms of the input trees. 
For a variable order or view tree $T$ and schema a set $\calS$ of variables occurring in
$T$, we define $Q_{T}(\calS) = \JOIN_{R(\calX) \in\atoms(T)} R(\calX)$.

The proof is by induction over the structure of $\omega$. We show that for any subtree $\omega'$
of $\omega$, it holds:
\begin{align}
\label{equ:equivalence}
Q_{\omega'}(\mathcal{F}\cap \vars(\omega')) \equiv 
\bigcup_{T \in \tau(\omega',\calF)} Q_T(\mathcal{F}\cap \vars(\omega')).
\end{align}
\begin{itemize}[label={}, leftmargin=0pt]
  
  \setlength\itemsep{0.4em}
  
	\item {\em Base case}:
If $\omega'$ is an atom, the procedure $\tau$
returns it and the base case holds trivially. 

	\item {\em Inductive step}:
Assume that $\omega'$ has root variable $X$ and subtrees
$\omega'_1, \ldots , \omega'_k$.
Let $keys = \anc(X) \cup \{X\}$, $\calF_X = \anc(X) \cup (\calF 
\cap \vars(\omega'))$, 
and
$Q_X(\mathcal{F}_X) = \JOIN_{R(\calX)\atoms(\omega')} R(\calX)$.
The procedure $\tau$ distinguishes the following cases: 

	\item \hspace{.5em} {\em Case 1: (mode = \text{`static'} $\land$ $Q_X(\mathcal{F}_X)$ \text{ is free-connex})\,$\lor$ (mode = \text{`dynamic'} $\land$ $Q_X(\mathcal{F}_X)$ \text{is $\dfw_0$-hierarchical})}.
The procedure $\tau$ returns a view tree $T$ constructed by the procedure 
$\textsc{BuildVT}(\cdot,\omega',\calF_X)$. The leaves of $T$ 
are the  atoms of $\omega'$. This implies 
Equivalence~\ref{equ:equivalence}.

	\item \hspace{.5em} {\em Case 1 does not hold and $X\in \calF$:\,\,} 
	The set of view trees $\tau(\omega',\calF)$ is defined as follows:
for each set $\{T_i\}_{i \in [k]}$ with $T_i \in \tau(\omega'_i,\calF)$, 
the set $\tau(\omega',\calF)$ contains the
view tree  $\textsc{NewVT}(\cdot,keys,\{\hat{T}_i\}_{i \in [k]})$
where $\hat{T}_i = \textsc{AuxView}(\text{root of } \omega_i',T_i)$
  for $i\in[k]$.
  
  Using the induction hypothesis, we rewrite as follows:
\begin{align*}
& Q_{\omega'}(\calF\cap \vars(\omega')) =   \JOIN_{i\in[k]} Q_{\omega_i'}(\calF\cap \vars(\omega_i')) 
 \overset{\text{IH}}{\equiv}    \JOIN_{i\in[k]} \Big(\bigcup_{T \in \tau(\omega'_i,\calF)} Q_T(\mathcal{F}\cap \vars(\omega'_i))\Big) \\[-0.1cm]
\equiv & \,  
\bigcup_{\forall i\in[k]: T_i \in \tau(\omega'_i,\calF)}  
\hspace{-0.3cm} \JOIN_{i\in[k]}  Q_{T_i}(\mathcal{F}\cap \vars(\omega'_i))
 \equiv \hspace{-0.1cm} \,   \bigcup_{\forall i\in[k]: T_i \in \tau(\omega'_i,\calF)}  
\hspace{-0.3cm}  Q_{\textsc{NewVT}(\cdot,keys,\{\hat{T}_i\}_{i \in [k]})}(\mathcal{F}\cap \vars(\omega')) \\[-0.1cm]
= &  \, 
\bigcup_{T \in \tau(\omega',\calF)} Q_T(\mathcal{F}\cap \vars(\omega')).
\end{align*}

	\item \hspace{.5em} {\em Case 1 does not hold and $X\not\in \calF$:\,\,}
The procedure $\tau$ creates 
the views
$All_X(keys) = \JOIN_{R(\calX) \in \atoms(\omega')} R(\calX)$,   
$L_X(keys) = \JOIN_{R(\calX) \in \atoms(\omega')} R^{keys}(\calX)$, and
$H_X(keys)$ = $All_X(keys)$ $\Join$ $\nexists L_X(keys)$.
 It then returns the view trees $\{ltree\} \cup htrees$ defined as
 follows:
 \begin{itemize}
  \item $ltree = \textsc{BuildVT}(\cdot, \omega^{keys}, \calF)$,
  where $\omega^{keys}$ has the same structure as $\omega'$ but each 
  atom is replaced by its light part;
  
  \item   for each set $\{T_i\}_{i \in [k]}$ with $T_i \in \tau(\omega'_i,\calF)$,
   $htrees$ contains the view tree 
$\textsc{NewVT}(\cdot,keys,$ $\{\exists H_X\}$ $\cup$ $\{\hat{T}_i\}_{i \in [k]})$
where $\hat{T}_i = \textsc{AuxView}(\text{root of } \omega_i',T_i)$
  for $i\in[k]$.
 \end{itemize}
From  
$ALL_X(keys) = L_X(keys) \cup H_X(keys)$, we derive the following
equivalence. 
For simplicity, we skip the schemas of queries:
\begin{align}
\label{eq:heavy_light}
& \bigcup_{\forall i\in[k]: T_i \in \tau(\omega'_i,\calF)}  
\hspace{-0.5cm} \JOIN_{i\in[k]}\  Q_{T_i}  \equiv Q_{ltree} 
\cup \bigcup_{\forall i\in[k]: T_i \in \tau(\omega'_i,\calF)}  
\hspace{-0.5cm} Q_{\textsc{NewVT}(\cdot,keys,\{\exists H_X\} \cup \{\hat{T}_i\}_{i \in [k]})}
\end{align}  

Using Equivalence \eqref{eq:heavy_light} and
the induction hypothesis, we obtain:
  
\begin{align*}
 Q_{\omega'}   =  & \JOIN_{i\in[k]} Q_{\omega_i'} 
 \overset{\text{IH}}{\equiv}   
\JOIN_{i\in[k]} \Big(\bigcup_{T \in \tau(\omega'_i,\calF)} Q_T \Big) 
 \equiv
\bigcup_{\forall i\in[k]: T_i \in \tau(\omega'_i,\calF)}  
\hspace{-0.3cm} \JOIN_{i\in[k]}  Q_{T_i} \\
  \overset{\eqref{eq:heavy_light}}{\equiv} &  
Q_{ltree} \cup \bigcup_{\forall i\in[k]: T_i \in \tau(\omega'_i,\calF)}  
\hspace{-0.5cm} Q_{\textsc{NewVT}(\cdot,keys,\{\exists H_X\} \cup \{\hat{T}_i\}_{i \in [k]})} \\
   = & \ 
Q_{ltree} \cup \bigcup_{T \in htrees} Q_{T} = \bigcup_{T \in \tau(\omega',\calF)} Q_T
\end{align*} 
\end{itemize}

\subsection{Proof of Proposition~\ref{prop:preproc_time}}

\begin{citedprop}[\ref{prop:preproc_time}]
Given a hierarchical query $Q(\calF)$ with
 static width $\fw$, a canonical variable order 
$\omega$ for $Q$,
a database  of size $N$, and $\eps \in [0,1]$, 
the views in the set of view trees $\tau(\omega,\calF)$
can be materialized in $\bigO{N^{1+(\fw-1)\eps}}$ time.
\end{citedprop}

We analyze the procedure $\tau$ from 
Figure \ref{fig:view_forest_main_algo}
for both of the cases
mode = \text{`static'}  and mode = \text{`dynamic'}.
We show that in both cases the time to materialize 
the set of view trees $\tau(\omega,\calF)$ is 
$\bigO{N^{1+(\fw-1)\eps}}$.

We explain the 
intuition behind the complexity analysis. 
If the procedure $\tau$ runs in \text{`static'} mode  
and $Q$ is free-connex,
or  it runs in  \text{`dynamic'} mode   
and $Q$ is $\dfw_0$-hierarchical, the procedure  
constructs a view tree that can be materialized 
in $\bigO{N}$ time.
Otherwise, there must be at least one bound variable $X$ in $\omega$
such that the subtree $\omega_X$ rooted at $X$ contains 
free variables.
In this case, the algorithm partitions the relations at the leaves   
of $\omega_X$ into heavy and light parts  and 
creates view trees for computing parts of the query. 
The time to materialize the views of the view trees where at least one leaf 
relation is
heavy is $\bigO{N}$. 
The overall time to materialize the view trees $\tau(\omega,\calF)$ 
is dominated by the time 
to materialize the views of the view trees where all 
leaf relations are light.
In the worst case, the root variable of $\omega$ is bound
and we need to materialize a view that joins the light parts of all leaf relations 
in $\omega$ and has the entire set $\calF$ as free variables.
We can compute such a view $V(\calF)$ as follows.
We first aggregate away all bound variables that are not ancestors 
of free variables in $\omega$. By using the algorithm  
InsideOut~\cite{FAQ:PODS:2016}, this can be done 
in $\bigO{N}$ time. Then, we choose one atom to iterate over 
the tuples of its relation (outer loop of the evaluation). For each such tuple, 
  we iterate over
 the matching tuples 
in the relations of the other atoms (inner loops of the evaluation). 
To decide which atom to take for the outer loop 
and which ones for the inner loops of our evaluation strategy,
we use an optimal
 integral edge cover $\boldsymbol{\lambda}$ of 
 $\calF$. The schema of each atom that is 
 mapped to $0$ by $\boldsymbol{\lambda}$ must be subsumed 
 by the schema of an atom mapped to $1$.  
 Hence, we can take one of the atoms mapped to 1 
 to do the outer loop. The other atoms that are mapped to 1
 are used for the inner loops. For the atoms that are 
 mapped to 0, it suffices to do constant-time lookups
 during the  iteration over the tuples of the other atoms.
 By exploiting the degree constraints on light relation parts,
 the  view $V(\calF)$ can be materialized in  
 $\bigO{N^{1+ (\rho(\calF)-1)\eps}}$ time.
 By Proposition~\ref{prop:rho_rhostar}, 
 $\rho(\calF) = \rho^{\ast}(\calF)$.
 Considering the time needed to aggregate  
away bound variables before computing $V(\calF)$, we get 
 $\bigO{N^{\max\{1,1+ (\rho^{\ast}(\calF)-1)\eps\}}}$
 overall time complexity.
We show that 
$\max\{1,1+ (\rho^{\ast}(\calF)-1)\eps\}$ is upper-bounded
by $1+ (\fw(Q)-1)\eps$.

The proof is structured following the basic building blocks 
of the procedure $\tau$. 
Lemmas \ref{lem:new_view_tree}-\ref{lem:indicators}
give upper bounds on the times to materialize the views 
in the view trees returned by the procedures 
\textsc{NewVT} (Figure~\ref{fig:view_creation}),
\textsc{AuxView} (Figure~\ref{fig:aux_view_creation}) 
\textsc{BuildVT} (Figure~\ref{fig:factorized_view_tree_algo}), and 
\textsc{IndicatorVTs} (Figure~\ref{fig:skew_aware_views}).
Lemma~\ref{lem:preproc_rho} states the complexity of the procedure 
 $\tau$ based on a measure $\xi$ defined over canonical
 variable orders. 
 The section closes with the proof of 
 Proposition~\ref{prop:preproc_time} that bridges the measure  $\xi$ 
 and the static width of hierarchical queries. 
 
We introduce the measure  $\xi$. 
Let $\omega$ be a canonical variable order, $\calF \subseteq \vars(\omega)$, 
and $X$ a variable or atom in $\omega$.
We denote by $\omega_X$ the subtree of $\omega$ rooted at $X$ and by 
$Q_X$ a query that joins the atoms at the leaves of $\omega_X$. 
We define 
$$\xi(\omega,X,\calF) = \max_{\substack{Y \in \vars(\omega_X)\\ 
(\anc(Y) \cup \{Y\}) \not\subseteq \calF}}
\{\rho_{Q_X}^{\ast}(\vars(\omega_Y) \cap \calF)\}.$$
If $\omega_X$ does not contain a variable $Y$ with
 $(\anc(Y) \cup \{Y\}) \not\subseteq \calF$, then 
 $\xi(\omega,X,\calF) = 0$. If $X$ has children $X_1, \ldots, X_k$, then 
\begin{align}
\label{eq:subtrees}
\xi(\omega, X,\calF)\geq \max_{i \in [k]} \{\xi(\omega, X_i,\calF)\}.
\end{align}

We start with an observation that each view $V$ constructed 
by the procedures 
\textsc{BuildVT} (Figure~\ref{fig:factorized_view_tree_algo}), 
\textsc{NewVT} (Figure~\ref{fig:view_creation}), 
\textsc{AuxView} (Figure~\ref{fig:aux_view_creation}),  
  \textsc{IndicatorVTs} (Figure~\ref{fig:skew_aware_views}), and 
$\tau$ (Figure~\ref{fig:view_forest_main_algo})
at some node $X$ of a 
variable order $\omega$ contains in its schema 
all variables 
in the root path of $X$ and no variables which are not in 
$\omega_X$. Moreover, $V$ results from the join of its 
child views. 
This can be shown by a straightforward 
 induction over the structure of $\omega$.
 
\begin{obs}
\label{obs:view_schemas}
Let $\omega$ be a canonical variable order
and $V(\calF)$ a view constructed at some node $X$
of $\omega$ by one of the procedures
$\textsc{BuildVT}$, 
$\textsc{NewVT}$,
\textsc{AuxView}, 
\textsc{IndicatorVTs}, and
$\tau$.
It holds 
\begin{enumerate}
\item $\anc(X) \subseteq \calF \subseteq \anc(X) \cup \vars(\omega)$.
\item If 
$V_1(\calF_1), \ldots , V_k(\calF_k)$
are the child views of $V(\calF)$, then
$V(\calF) = V_1(\calF_1), \ldots , V_k(\calF_k)$.
\end{enumerate}
\end{obs}

The next lemma gives a bound on the 
time to materialize the views in a view
tree returned by the
procedure \textsc{NewVT} in Figure~\ref{fig:view_creation}.

\begin{lem}
\label{lem:new_view_tree}
Given a set $\{T_i\}_{i \in [k]}$ of view trees 
with root views $\{V_i(\calS_i)\}_{i \in [k]}$,
let  $M_i$ be the time
to materialize the views in $T_i$,
for $i \in [k]$. 
If the query $V(\calS) = V_1(\calS_1), \ldots , V_k(\calS_k)$
with $\calS \subseteq \bigcap_{i \in [k]} \calS_i$
is $\dfw_0$-hierarchical, the views in the view tree 
$\textsc{NewVT}(\cdot,\calS, \{T_i\}_{i \in [k]})$
 can be materialized in 
$\bigO{\max_{i \in [k]}\{M_i\}}$ time.
\end{lem}

\begin{proof}
The procedure $\textsc{NewVT}$ 
defines the view 
$V(\calS)\allowbreak = V_1(\calS_1),\allowbreak \ldots ,\allowbreak 
V_k(\calS_k)$ (Line~2).
The view tree $T$ returned by $\textsc{NewVT}$ is 
defined as follows (Line~4): 
If $k = 1$ and $\calS = \calS_1$, then $T = T_1$;
otherwise, $T$ is the view tree that
has root $V(\calS)$ and subtrees
$T_1, \ldots, T_k$.  
By assumption, 
the time to materialize the views in the trees
$T_1, \ldots, T_k$ is 
$\bigO{\max_{i \in [k]}\{M_i\}}$.
Hence, the sizes of the materialized root views
$V_1(\calS_1), \ldots, V_k(\calS_k)$ 
must be
$\bigO{\max_{i \in [k]}\{M_i\}}$.
Assume that the query defining 
$V(\calS)$ is $\dfw_0$-hierarchical.  
Hence,  we can construct a free-top 
canonical variable order for the query. 
We materialize $V(\calS)$ as follows. 
Traversing the variable order bottom-up,
we aggregate away all bound variables 
using the InsideOut 
algorithm  \cite{FAQ:PODS:2016}.
Since  the query defining $V(\calS)$ is $\alpha$-acyclic, this aggregation 
phase can be done in time linear in the size 
of the views $V_1(\calS_1), \ldots, V_k(\calS_k)$.  
Thus, the aggregation phase requires 
$\bigO{\max_{i \in [k]}\{M_i\}}$ time. 
It follows that the time  to materialize
the views in the tree returned by \textsc{NewVT}
is $\bigO{\max_{i \in [k]}\{M_i\}}$.
\end{proof}

We proceed with a lemma 
that gives a bound on the 
time to materialize the views in the view
tree returned by the
procedure 
\textsc{AuxView} in Figure~\ref{fig:aux_view_creation}.

\begin{lem}
\label{lem:auxiliary_view}
Let $T$ be  a view tree 
and $M$ 
the time to materialize the views in $T$. 
The views in the view tree $\textsc{AuxView}(\cdot,T)$ can be materialized in time 
$\bigO{M}$.
\end{lem}

\begin{proof}
Assume that the parameters of the procedure 
$\textsc{AuxView}$ are
$\calZ$ and $T$.
Let $V(\calS)$ be the root of $T$.
If  the condition in Line~3 of the 
 procedure \textsc{AuxView} 
 does not hold, 
the procedure
returns $T$ (Line~5).
Otherwise, 
it returns a view tree $T'$ that results from 
$T$ by adding  a view 
$V'(\anc(Z))$ on top of $V(\calS)$ (Line~{4}).
Since $\anc(Z) \subset \calS$, 
$V'(\anc(Z))$ results from $V(\calS)$
by aggregating away 
the variables in $\calS - \anc(Z)$.
Since  the size of $V(\calS)$ must be $\bigO{M}$
and the variables can be aggregated away in time 
linear in the size of $V(\calS)$, the overall time to materialize 
the views in the output tree $T'$ is $\bigO{M}$.
\end{proof}

 The following lemma says that if the input to 
the procedure 
 \textsc{BuildVT} in Figure~\ref{fig:factorized_view_tree_algo}
 represents a free-connex query, 
the procedure outputs a view tree
 whose views can be materialized in time linear
 in the database size.

\begin{lem}
\label{lem:fact_view_tree}
Let $\omega$ be a canonical variable order, 
$X$ a node in $\omega$, $N$ the size of the leaf relations 
of $\omega$, and $\calF$ a set of variables. 
If the query 
$Q_X(\calF') = \text{ join of }\atoms(\omega_X)$ 
with $\calF' = \calF \cap (\anc(X) \cup \vars(\omega_X))$
is free-connex, the views in the view tree
$\textsc{BuildVT}(\cdot,\omega_X, \calF)$ 
 can be materialized in $\bigO{N}$ time. 
\end{lem}

\begin{proof} 
The proof is by induction over the structure of the variable order 
$\omega_X$. 

\smallskip 
{\em Base case}:
Assume that $X$ is a single atom $R(\calX)$. In this case,
the procedure $\textsc{BuildVT}$ returns this atom, which can obviously be
materialized in $\bigO{N}$ time.    

\medskip 
{\em Inductive step}:
Assume that $X$ is a variable with child nodes $X_1, \ldots, X_k$
and $Q_X(\calF') = \text{ join of }\atoms(\omega_X)$ a free-connex query. 
Let $\calF_i = 
\calF \cap(\anc(X_i) \cup \vars(\omega_{X_i}))$, for $i\in [k]$.

\smallskip
We first show that for each $i \in [k]$:  
\begin{align}
\label{subtree_freeconnex}
Q_{X_i}(\calF_i) = \text{join of }\atoms(\omega_{X_i})
\text{ is free-connex.} 
\end{align}

\noindent
 An $\alpha$-acyclic query is free-connex  
if and only if after adding an atom $R(\calX)$, where $\calX$ is the set of free
variables, the query 
remains $\alpha$-acyclic~\cite{BraultPhD13}.
A query is $\alpha$-acyclic if it has a (not necessarily free-top) variable order 
with static width $1$~\cite{OlteanuZ15,BeeriFMY83}. 
Let $Q_{X}'$ be the query that results from 
$Q_{X}$ by adding a new atom $R(\calF')$. 
Likewise, let $Q_{X_i}'$ be the query that we obtain 
from $Q_{X_i}$ by adding a new atom $R_i(\calF_i)$,
for $i \in [k]$.
Since $Q_{X}$ is free-connex, there 
must be a variable order $\omega' = (T_{\omega'}, \dep_{\omega'})$ 
for $Q_{X}'$, 
such that $\fw(\omega') = 1$. 
In the following, we turn $\omega'$ into a variable 
order $\omega_{i}' = (T_{\omega_i'}, \dep_{\omega_i'})$ 
for $Q_{X_i}'$ with $\fw(\omega_{i}') = 1$,
for $i \in [k]$. From this, it follows 
that $Q_{X_i}$ is free-connex, for each $i \in [k]$.
To obtain $\omega_{i}'$, we traverse 
$\omega'$ bottom-up and
eliminate all variables and atoms (including
$R(\calF')$) that
do not occur in $Q_{X_i}'$.
When eliminating a node $Y$ with a parent node
 $Z$, we append the children of $Y$ to $Z$. If $Y$ does 
 not have any parent node, 
the subtrees rooted at its children nodes  become independent.
 Finally, we append $R_i(\calF_i)$ under the lowest 
 variable $Y$ in the obtained variable order such that 
 $Y$ is included in $\calF_i$.  In the following we show
 that for $i \in [k]$:
 
 \begin{enumerate}
 \item $\omega_i'$ is a valid variable order for $Q_{X_i}'$.
 \item $\fw(\omega_{i}') = 1$.
 \end{enumerate}
 
 \textit{(1) $\omega_i'$ is a valid variable order for $Q_{X_i}'$,
 for $i \in [k]$}:
 The following property follows from the 
 construction of the variable order $\omega_i'$:
 
 \smallskip
 \noindent
 $(*)$:  Any two variables
 in $\vars(Q_{X_i}')$ that are on the same 
  root-to-leaf path in $\omega'$ remain on the same 
 root-to-leaf path in $\omega_{i}'$.

 \smallskip
 \noindent
Each atom $K(\calX)$ in $\atoms(Q_{X_i}') - \{R_i(\calF_i)\}$
is also an atom in $\atoms(Q_{X}')$. Hence, the variables in 
$\calX$ must be on the same root-to-leaf path in
$\omega'$. Due to $(*)$, they also 
must be on the same root-to-leaf path in 
$\omega_{i}'$. 
It remains to show that all variables 
in $\calF_i$ are on the same root-to-leaf path
in $\omega_{i}'$. 
In the canonical variable order $\omega$, each
variable in $\calF' -\{X\}$ is either above or below $X$. Hence, 
$X$ depends on all variables 
in $\calF'-\{X\}$, which means that all variables in $\calF'$ 
must be on the same root-to-leaf path 
in $\omega'$. Due to $\calF_i \subseteq \{X\} \cup \calF'$ and 
Property ($*$), all variables in $\calF_i$ 
must be  on the same  root-to-leaf path in 
$\omega_{i}'$.

\smallskip
\textit{(2) $\fw(\omega_{i}') = 1$, for $i \in [k]$}:
Let $Y \in \vars(\omega_i')$ for some $i \in [k]$.
We need to show that 
$\rho^{*}_{Q_{X_i}'}(\{Y\} \cup \dep_{\omega_i'}(Y)) = 1$.
By Proposition~\ref{prop:rho_rhostar}, it suffices to show that 
$Q_{X_i}'$ contains an atom that covers 
$\{Y\} \cup \dep_{\omega_{i}'}(Y)$, i.e., whose schema includes 
the latter set.
By construction, $Y$ must be included in $\omega'$.
First, observe that any two variables that are not 
dependent in $Q_{X}'$ cannot be dependent in 
$Q_{X_i}'$. Moreover, each variable $Z$ included 
in the root path 
of $Y$ in $\omega_i'$, is included in the
root path of $Y$ in $\omega'$. Hence: 

\smallskip 
\noindent
($**$) $\{Y\} \cup \dep_{\omega_i'}(Y) 
\subseteq (\{Y\} \cup \dep_{\omega'}(Y)) \cap \vars(Q_{X_i}')$.

\smallskip 
\noindent
Due to  $\fw(\omega') = 1$ and Proposition~\ref{prop:rho_rhostar}, 
there must be an atom 
$K(\calX) \in \atoms(Q_X')$ such that 
$\{Y\} \cup \dep_{\omega'}(Y) \subseteq \calX$.
First, assume that $K(\calX) \neq R(\calF')$. 
Since $\calX$ includes $Y$, the atom $K(\calX)$ must be under 
the variable $Y$ in $\omega_{X_i}$, which means that 
$\atoms(Q_{X_i}')$ includes $K(\calX)$. Due to Property 
($**$), $K(\calX)$ covers 
$\{Y\} \cup \dep_{\omega_i'}(Y)$.
Now assume that $K(\calX) = R(\calF')$.
This means that  $\{Y\} \cup \dep_{\omega'}(Y) \subseteq \calF'$.
By Property ($**$), 
$\{Y\} \cup \dep_{\omega_{i}'}(Y) \subseteq
(\{Y\} \cup \dep_{\omega'}(Y)) \cap  \vars(Q_{X_i}') \subseteq 
\calF' \cap \vars(Q_{X_i}')$. 
Since  
$\calF_i = \calF' \cap \vars(Q_{X_i}')$, 
$R_i(\calF_i)$ covers 
$\{Y\} \cup \dep_{\omega_i'}(Y)$.

\smallskip 
This completes the proof of \eqref{subtree_freeconnex}.

\smallskip
Let $T = \textsc{BuildVT}(\cdot,\omega_X, \calF)$.
To construct  the view tree $T$, 
the procedure \textsc{BuildVT} first constructs the view trees  
$\{T_i\}_{i \in [k]}$ with $T_i = \textsc{BuildVT}(\cdot ,\omega_{X_i}, \calF)$
for each $i \in [k]$ (Line~2).
By Property~\eqref{subtree_freeconnex} and the induction hypothesis, 
the views in each view tree $T_i$
can be materialized in $\bigO{N}$ time.  
In the following we show that all views in $T$ can be materialized 
in $\bigO{N}$ time.
We distinguish whether $\anc(X) \cup \{X\}$ is 
included in $\calF$ (Lines~4-7)
or not (Lines~8-10): 

\smallskip
\textit{Case $\anc(X) \cup \{X\} \subseteq \calF$:}
In this case,  it holds $T = \textsc{NewVT}(\cdot,\calF_X, subtrees)$, where 
$subtrees = \{\, \textsc{AuxView}(X_i, T_i) \,\}_{i\in[k]}$
and $\calF_X = \anc(X) \cup \{X\}$.
The procedures   
\textsc{NewVT} and 
\textsc{AuxView} are given in 
Figures~\ref{fig:view_creation} and 
\ref{fig:aux_view_creation}, respectively. 
By the induction hypothesis and 
Lemma~\ref{lem:auxiliary_view}, the views
in $subtrees$ can be materialized in  
$\bigO{N}$ time. 
Let $V_1'(\calF_1'), \ldots , V_k(\calF_k')$ be the roots
of the trees in $subtrees$. 
The overall size of these root views 
must be $\bigO{N}$.
Observation \ref{obs:view_schemas}.(1) implies that 
for any $i,j \in [k]$ with $i \neq j$, it holds 
$\calF_i'  \cap \calF_i' = \calF_X$. Hence, 
the query 
$V_X(\calF_X) = V_1(\calF_1'), \ldots , V_k(\calF_k')$
is $\dfw_0$-hierarchical.  
Since $\calF_X = \anc(X) \cup \{X\} \subseteq \bigcap_{i \in [k]}\calF_i'$,
it follows from 
Lemma~\ref{lem:new_view_tree} that the views in $T$
can be materialized in $\bigO{N}$ time.

\smallskip
\textit{Case $\anc(X) \cup \{X\} \not\subseteq \calF$}:
n this case, we have $T = \textsc{NewVT}(\cdot,\calF_X, subtrees)$,
where $\calF_X = \anc(X) \cup (\calF \cap \vars(\omega_X))$ and 
$subtrees = \{T_i\}_{i \in [k]}$.
Let $V_i'(\calF_i')$ be the root of $T_i$,
for $i \in [k]$. By the definition of the procedure 
\textsc{NewVT}, the tree $T$ results from the trees 
 $\{T_i\}_{i \in [k]}$ by adding a new root view
 defined by 
 $V_X(\calF_X) = V_1'(\calF_1'), \ldots , V_k'(\calF_k')$.
It follows from  Observation \ref{obs:view_schemas}.(2),
that $V_X(\calF_X)$ can be rewritten 
as 
$V_X(\calF_X) = \text{ join of } \atoms(\omega_X)$.
  We show that the view $V_X(\calF_X)$ 
 can be materialized in  
$\bigO{N}$ time. 
The set
$\atoms(\omega_X)$ must contain an atom $R(\calY)$ with 
$\calF_X \subseteq \calY$ (Lemma 35 in~\cite{Trade_Offs_arxiv}).
Hence, we can easily materialize 
the view $V_X(\calF_X)$ by 
using the InsideOut algorithm 
\cite{FAQ:PODS:2016}
to aggregate away all variables that are not included in 
$\calF_X$. Since
the query defining the view $V_X(\calF_X)$ is ($\alpha$-)acyclic, 
the whole computation takes 
$\bigO{N}$ time.
\end{proof}

The next lemma upper bounds the time to materialize the views constructed by the procedure \textsc{BuildVT} in Figure~\ref{fig:factorized_view_tree_algo} for a variable order 
$\omega^{keys}_X$. This variable order has the same structure as $\omega_X$ yet each atom $R(\calY)$ is replaced
by the light part $R^{keys}(\calY)$ of relation $R$ partitioned 
on the variable set $keys$ (cf.\@~Section~\ref{sec:skew_aware_trees}).

\begin{lem}
\label{lem:fact_view_tree_light_case}
Given a canonical variable order $\omega$, 
a node $X$ in $\omega$, the size $N$ of the leaf relations
in $\omega$, $keys = \anc(X) \cup \{X\}$, 
$\calF \subseteq \vars(\omega)$, and 
$\eps \in [0,1]$. 
The view tree     
$\textsc{BuildVT}(\cdot,\omega^{keys}_X, \calF)$ 
can be materialized in 
$\bigO{N^{\max\{1, 1+(\xi(\omega^{keys}, X,\calF)-1)\eps\}}}$ time. 
\end{lem}

\begin{proof}
For a node $X$ in $\omega^{keys}$, we set  
$$m_X = \max\{1, 1+(\xi(\omega^{keys}, X,\calF)-1)\eps\}.$$
The proof is by induction on the structure of $\omega^{keys}_X$.

\begin{itemize}[label={}, leftmargin=0pt]
  
  \setlength\itemsep{0.4em}

  \item \hspace{1em} {\em Base case}: 
  If $\omega^{keys}_X$ is a 
   single atom $R(\calX)$,
   the procedure $\textsc{BuildVT}$
    returns this atom, which 
    can be materialized in 
  $\bigO{N}$ time. Since $m_X \geq 1$,
  this completes the base case.

  \item \hspace{1em} {\em Inductive step}:
  Assume $X\in\vars(\omega^{keys})$ and has child nodes 
  $X_1, \ldots , X_k$.
  The procedure first calls 
  $\textsc{BuildVT}(\cdot,\omega^{keys}_{X_i}, \calF)$
  for each $i \in [k]$ and produces the view trees
  $\{T_i\}_{i \in [k]}$ (Line~2). By induction hypothesis,
  we need $\bigO{N^{m_{X_i}}}$
  time to materialize the views in each view tree $T_i$
  with $i \in [k]$. 
  The procedure $\textsc{BuildVT}$
   distinguishes whether 
   $(\anc(X) \cup \{X\}) \subseteq \calF$ (Lines~4-7) or not
   (Lines~8-10).

   \item \hspace{1em} \textit{Case $(\anc(X) \cup \{X\}) \subseteq \calF$:}
   The view tree $T$ returned by  the procedure 
   $\textsc{BuildVT}$
   is $\textsc{NewVT}(\cdot,\calF_X,$ $subtrees)$, where 
   $subtrees$ is defined as 
   $\{\, \textsc{AuxView}(X_i, T_i) \,\}_{i\in[k]}$
   and $\calF_X = \anc(X) \cup \{X\}$.
   By induction hypothesis
   and Inequality \eqref{eq:subtrees}, the overall time 
   to materialize the views in $\{T_i\}_{i \in [k]}$ 
   is $\bigO{N^{m_X}}$. For each view tree $T_i$, 
     $\textsc{AuxView}(X_i, T_i)$
    adds at most one view with schema 
    $\anc(X_i)$ on top of the root view of $T_i$. 
    Then,  $\anc(X_i)$ is 
   a subset of the schema of the root view of $T_i$.
   Since the size of the root view of $T_i$ must be bounded by
   $\bigO{N^{m_X}}$, the view added by 
   $\textsc{AuxView}$ can be materialized in 
   $\bigO{N^{m_X}}$ time.
   Assume that  $V_1(\calF_1), \ldots , V_k(\calF_k)$ are the roots
   of the view trees in $subtrees$. In case $k= 1$ and $\calF = \calF_i$,
   $\textsc{NewVT}(V_X,\calF_X, subtrees)$
   returns $V_1(\calF_1)$; otherwise, it returns 
   a view tree that has $V_X(\calF_X) = V_1(\calF_1), \ldots , V_k(\calF_k)$
   as root view and $subtrees$ as subtrees. 
   By the definition of 
   \textsc{AuxView},
   it holds $\calF_i \cap \calF_j = (\anc(X) \cup \{X\}) = \calF_X$
   for any  $i,j \in [k]$.
   Hence, the view $V_X(\calF_X)$ can be computed by iterating over 
   the tuples  in a view $V_i(\calF_i)$ with $i \in [k]$  
   and filtering out those tuples that do not have matching tuples 
   in all views
   $V_j(\calF_j)$ with $j \in [k]-\{i\}$. 
   Since the size of $V_i(\calF_i)$ is   
   $\bigO{N^{m_X}}$ and materialized views allow constant time lookups, 
   the view $V_X(\calF_X)$ can be computed 
   in $\bigO{N^{m_X}}$ time. It follows that    
   the view tree $T$ 
   returned by $\textsc{BuildVT}$ can
   be materialized in $\bigO{N^{m_X}}$ time.
   This completes the inductive step for this case.   
   
   \item \hspace{1em} \textit{Case $(\anc(X) \cup \{X\}) \not\subseteq \calF$:}
   The procedure $\textsc{BuildVT}$ sets 
    $\calF_X =\anc(X) \cup (\calF \cap \vars(\omega_X^{keys}))$ and
   $subtrees$ $= \{T_i\}_{i \in [k]}$. 
   The view tree $T$ returned by the procedure $\textsc{BuildVT}$
    is $\textsc{NewVT}(\cdot, \calF_X, subtrees)$. 
    We show that all views in the view tree  $T$ 
    can be materialized 
   in $\bigO{N^{m_X}}$ time.
   We analyze 
   the steps in $\textsc{NewVT}$. In case $subtrees$
   consists of a single tree $T'$ such that the schema 
   of the root view of $T'$ is $\calF_X$, the procedure 
   $\textsc{NewVT}$ returns the view tree $T'$. By induction hypothesis
   and Inequality \eqref{eq:subtrees}, 
   the views in $T = T'$ can be materialized in $\bigO{N^{m_X}}$ time.
   Otherwise, let  
   $V(\calF_i)$ be the root view of $T_i$, for $i \in [k]$. 
   The tree $T$ returned by $\textsc{NewVT}$ consists
   of the root  view 
   $$V_X(\calF_X) = V_1(\calF_1), \ldots ,V_k(\calF_k)$$
   with subtrees $\{T_i\}_{i \in [k]}$. By induction hypothesis
   and Inequality \eqref{eq:subtrees}, 
   the views in the trees $\{T_i\}_{i \in [k]}$ can be materialized 
   in $\bigO{N^{m_X}}$ time.  
   It suffices to show that 
   $V_X(\calF_X)$ can be materialized in $\bigO{N^{m_X}}$ time. 
   Using Observation \ref{obs:view_schemas}.(2),
   we rewrite the view $V_X(\calF_X)$ using the leaf atoms of $\omega^{keys}_X$:
   $$V_X(\calF_X) = \text{ join of }\atoms(\omega^{keys}_X).$$
   We materialize the view $V_X(\calF_X)$ as follows. 
   Using the Inside\-Out algorithm
   \cite{FAQ:PODS:2016},
   we first aggregate away 
   all variables  in  
   $\vars(\omega^{keys}_X) - \calF_X$ 
   that are not above a variable from $\calF_X$. Since the view $V_X$ is defined by an $\alpha$-acyclic 
   query, the time required by this step  
   is $\bigO{N}$.  
   Let 
   $V_X'(\calF_X) = R_1(\calF_1), \ldots ,R_k(\calF_k)$
   be the resulting query.
   We distinguish between two subcases.

   \item \hspace{1em} \textit{Subcase 1:
For all $R_i(\calF_i)$, it holds  $\calF_i \cap \vars(\omega^{keys}_X) = \emptyset$  
} 

This means that 
$\calF_X$ and each $\calF_i$ are contained 
in $\anc(X) \cup \{X\}$. 
Since $\omega^{keys}$ is canonical, the inner nodes 
of each root-to-leaf path are the variables of an atom. Hence,  
there is an $R_i(\calF_i)$ with $i \in [k]$
such that $\calF_i$ subsumes $\calF_X$ and each $\calF_j$ with 
$j \in [k]$.  
Thus, we can  materialize the result of 
$V_X'(\calF_X)$ in $\bigO{N}$ time 
by iterating over the tuples in 
$R_i$ and doing 
constant-time lookups in the other relations.   

\item \hspace{1em} \textit{Subcase 2:
There is an $R_i(\calF_i)$ with $\calF_i \cap \vars(\omega^{keys}_X) \neq \emptyset$}
 
 Let $\boldsymbol{\lambda} = (\lambda_{R_i(\calF_i)})_{i \in [k]}$
be an edge cover 
of $\calF_X \cap  \vars(\omega^{keys}_X)$
with $\sum_{i \in [k]} \lambda_{R_i(\calF_i)}
= \rho_{V_X'}^{\ast}(\calF_X \cap \vars(\omega^{keys}_X))$.
Since $V_X'$ is hierarchical, 
we can assume that 
each $\lambda_{R_i(\calF_i)}$ is either $0$ or $1$
(Proposition~\ref{prop:rho_rhostar}).  
There must be at least 
one $R_i(\calF_i)$ with $\lambda_{R_i(\calF_i)} = 1$, otherwise 
there cannot be any variable from $\calF_X$ in $\omega^{keys}_X$ and 
we fall back to Subcase 1. 
Since $\omega_X^{keys}$ is canonical,
for each atom $R_i(\calF_i)$ with 
$\lambda_{R_i(\calF_i)} = 0$, there must be 
a {\em witness atom} $R_j(\calF_j)$ such that 
$\lambda_{R_j(\calF_j)} = 1$ and
$\calF_i \subseteq \calF_j$.
The atoms $R_1(\calF_1), \ldots , R_k(\calF_k)$
can still contain variables not included in $\calF_X$.
Each such variable appears above at least one 
variable from $\calF_X$ in $\omega^{keys}_X$.
We first compute the
result of the view 
$V_X''(\bigcup_{i \in [k]}\calF_k) = R_1(\calF_1), \ldots , R_k(\calF_k)$
 as follows. 
We choose 
an arbitrary atom 
 $R_i(\calF_i)$ with
$\lambda_{R_i(\calF_i)} = 1$
 and iterate over the tuples in $R_i$. 
 For each such tuple, we iterate over the matching tuples in the other atoms mapped 
to $1$ by $\boldsymbol{\lambda}$.
For atoms that are not mapped to $1$, it suffices
to do constant-time lookups while iterating over
one of their witnesses.   
To obtain the result of $V_X'$ from $V_X''$, we 
aggregate away all variables 
not included in $\calF_X$.
Recall that for each atom $R_i(\calF_i)$, there is an
atom in $\atoms(\omega^{keys}_X)$ that is the light part 
of a relation partitioned on $keys= \anc(X) \cup \{X\}$.
Hence, each tuple in the relation of an atom 
mapped to $1$ by   
$\boldsymbol{\lambda}$ 
can be paired with 
 $\bigO{N^{\eps}}$ tuples
in the relation of any other atom mapped to $1$.
This means that the time to materialize 
$V_X''$ and hence $V_X'$ is
$\bigO{N^{m'}}$ where 
$m' = 1+(\rho_{V_X'}^{\ast}(\calF_X \cap \vars(\omega^{keys}_X))-1)\eps$.
Since $V_X'$ results from 
$V_X$ by aggregating away variables in  
$\vars(\omega^{keys}_X) -\calF_X$, 
we have 
$\rho_{V_X'}^{\ast}(\calF_X \cap \vars(\omega^{keys}_X))
=
\rho_{V_X}^{\ast}(\calF_X \cap \vars(\omega^{keys}_X))$.
It follows from  $\anc(X) \cup \{X\} \not\subseteq \calF$
that    
$\rho_{V_X}^{\ast}(\calF_X \cap \vars(\omega^{keys}_X))$ $=$
$\xi(\omega^{keys}, X,\calF)$. Hence, the view $V_X'$ can be materialized in 
$\bigO{N^{1+(\xi(\omega^{keys}, X,\calF)-1)\eps}}$
time.

\end{itemize}

We sum up the analysis for the case 
$(\anc(X) \cup \{X\}) \not\subseteq \calF$: the initial aggregation step
and the computation in Subcase 1 take $\bigO{N}$ time;
the computation in Subcase 2  takes
$\bigO{N^{1+(\xi(\omega^{keys}, X,\calF)-1)\eps}}$ time.
Thus, given $m_X = \max\{1, 1+(\xi(\omega^{keys}, X,\calF)-1)\eps\}$,
the time to materialize the result of $V_X$
is $\bigO{N^{m_X}}$.
This completes the inductive step in case 
$(\anc(X) \cup \{X\}) \not\subseteq \calF$.
\end{proof}

The next lemma states that the view trees 
returned by the procedure
\textsc{IndicatorVTs} from Figure \ref{fig:skew_aware_views}
can be materialized in time linear in the database size. 

\begin{lem}
\label{lem:indicators}
Let $\omega$ be a canonical variable order, 
$X$ a variable in $\omega$, and 
$N$ the size of the leaf relations in the variable order $\omega$. 
The views in the view trees returned by 
$\textsc{IndicatorVTs}(\omega_X)$ 
can be materialized in  $\bigO{N}$ time.
 \end{lem}

\begin{proof}
In Lines~3 and 4, the procedure constructs 
the view tree 
 $alltree$, which is  defined by $\textsc{BuildVT}({\text{``All\,''}},\omega_X,keys)$ 
and the view tree
 $ltree = \textsc{BuildVT}({\text{``\,L\,''}},\omega^{keys}_X,keys)$,
where  $keys$ consists of the set $\anc(X) \cup \{X\}$. The variable order
$\omega^{keys}_X$ results from $\omega_X$
by replacing each atom $R(\calX)$ by the atom 
$R^{keys}(\calX)$, 
which denotes the light part of relation $R$ partitioned on $keys$.
These light parts can be computed in $\bigO{N}$ time.
The queries 
$Q_X(keys) = \text{ join of }\atoms(\omega_X)$
and
$Q_X^{keys}(keys) = \text{ join of }\atoms(\omega_X)$
are free-connex. 
By using Lemma~\ref{lem:fact_view_tree}, we derive that 
the views in $alltree$ and $ltree$
can be materialized in 
$\bigO{N}$ time. Hence, the roots $allroot$
and $lroot$ of $alltree$ and $ltree$, respectively,
can be materialized in $\bigO{N}$ time as well.
It remains to analyze the time to materialize the views 
in the view tree $htree = \textsc{NewVT}(\cdot,\calF,\{allroot,\neg lroot\})$ (Line~7).
It follows from Observation \ref{obs:view_schemas}.(1) that
$V(\calF) = allroot, \nexists lroot$ is $\dfw_0$-hierarchical.
By using 
 Lemma \ref{lem:new_view_tree},
 we derive that the views in $htree$ 
   can be materialized in $\bigO{N}$ time. 
 Overall, all  views in the view trees 
 $(alltree, ltree,htree)$
  can be materialized in 
 $\bigO{N}$ time.
\end{proof}

We use Lemmas 
\ref{lem:new_view_tree}-\ref{lem:indicators}
to show an upper bound 
on the time to materialize the views 
in any tree produced by
the procedure $\tau$ in Figure~\ref{fig:view_forest_main_algo}.

\begin{lem}
\label{lem:preproc_rho}
Let $\omega$ be a canonical variable order, 
$X$ a node in $\omega$, $\calF \subseteq \vars(\omega)$, 
$N$ the size of the leaf relations in $\omega$,
and 
$\eps \in [0,1]$.
The views in the trees 
returned by  
$\tau(\omega_X, \calF)$
can be materialized in 
$\bigO{N^{\max\{1, 1+(\xi(\omega, X,\calF)-1)\eps\}}}$
time.
\end{lem}

\begin{proof}
For simplicity, we set
$$m = \max\{1, 1+(\xi(\omega, X,\calF)-1)\eps\}.$$
 The proof is by induction on the structure of 
$\omega_X$.

\medskip 
{\em Base case}:
Assume that $\omega_X$ is a 
 single atom $R(\calX)$. In this case, the procedure $\tau$ returns 
 this atom (Line~1).  
The atom can obviously be materialized in 
$\bigO{N}$ time. It holds 
 $\xi(\omega,X,\calF) = 0$, since $\omega_X$ does not contain 
any node which is a variable.
This means that  
$m= 1$. Then, the 
statement in the lemma holds for the base case.

\medskip 
{\em Inductive step}: 
 Assume that $X$ is a variable with children 
 nodes $X_1, \ldots , X_k$.
 Let  $keys = \anc(X) \cup \{X\}$, 
 $\calF_X = \anc(X) \cup (\calF \cap \vars(\omega_X))$, and
 $Q_X(\mathcal{F}_X) = \text{join of } \atoms(\omega)$.
 Following the control flow in $\tau(\omega_X,\calF)$, we make
 a case distinction.

\smallskip
\textit{Case 1:  
mode = `static' $\wedge$ $Q_X(\calF_X)$ is free-connex or mode = `dynamic' $\wedge$ $Q_X(\calF_X)$ is $\dfw_0$-hierarchical (Lines~5-7)}:

The procedure $\tau$ returns the view tree $\textsc{BuildVT}({\text{``\,V\,''}},\omega_X, \calF_X)$ (Line~7).
Since $\dfw_0$-hierarchical queries are in particular free-connex, 
it follows from  Lemma \ref{lem:fact_view_tree} that
$\textsc{BuildVT}({\text{``\,V\,''}},\omega_X, \calF_X)$  
 can be materialized 
in time $\bigO{N}$. This completes the inductive step 
for Case 1.

\smallskip
\textit{Case 2:  Case 1 does not hold and $X \in \calF$ (Lines~8-11)}:
\newline
The set of view trees $\tau(\omega_X,\calF)$ is defined as follows:
for each set $\{T_i\}_{i\in [k]}$ with $T_i \in \tau(\omega_{X_i},\calF)$, 
the set $\tau(\omega_X,\calF)$ contains 
the view tree  
$\textsc{NewVT}(\cdot, keys, \{\hat{T}_i\}_{i \in [k]})$, 
where $\hat{T}_i = \textsc{AuxView}(X_i,T_i)$ for each  
$i \in [k]$.
We consider one such set
$\{T_i\}_{i \in [k]}$ of view trees.  
By induction hypothesis, the views in each 
$T_i$ can be materialized in
$\bigO{N^{\max\{1, 1+(\xi(\omega, X_i,\calF)-1)\eps\}}}$ 
time.
It follows from Inequality \eqref{eq:subtrees}, that 
the overall time to materialize 
the views in these view trees is 
$\bigO{N^{m}}$.
By Lemma~\ref{lem:auxiliary_view}, 
the views in each view tree $\hat{T}_i$ with $i \in [k]$
can be materialized in 
$\bigO{N^{m}}$ time. 
Let $V_i(\calF_i)$ be the root view of $\hat{T}_i$,
for $i\in [k]$.
 It follows from Observation \ref{obs:view_schemas}.(1) that  
$keys$ is included in each $\calF_i$ 
and the query 
$V_X(keys) = V_1(\calF_1), \ldots , V_k(\calF_k)$
is $\dfw_0$-hierarchical.
Hence, it follows from Lemma~\ref{lem:new_view_tree} 
that the views in the view tree $\textsc{NewVT}(\cdot, keys, \{\hat{T}_i\}_{i \in [k]})$  
can be materialized in time 
$\bigO{N^{m}}$. 
This completes the inductive step in this case.

\smallskip
\textit{Case 3:  Case 1 does not hold and $X \not\in \calF$ (Lines~12-17)}:
\newline
The procedure $\tau$ first calls 
 $\textsc{IndicatorVTs}(\omega_X)$ (Line~12) given in Figure 
 \ref{fig:skew_aware_views}, which constructs  the indicator view trees
 $alltree$, $ltree$, and $htree$. 
By Lemma \ref{lem:indicators}, the views in these view trees 
 can be materialized in $\bigO{N}$ time.
 Let $H_X$ be the root of $htree$.
 The only difference between the construction of the view 
 trees returned  in \textit{Case~2} above and the view trees 
 in the set $htrees$ defined in Lines~13-15 is that 
 the roots of the view trees in the latter set have 
 $\exists H_X$ as additional child view.  
 By the same argumentation
 as in \textit{Case 2}, it follows 
 that the views in $htrees$ can be materialized
 in $\bigO{N^{m}}$ time. 
 Let $ltree = \textsc{BuildVT}({\text{``\,V\,''}},\omega_X^{keys}, \calF_X)$
 as defined in Line~16, where 
$\omega_X^{keys}$ shares the same structure as $\omega_X$,
but each atom $R(\calX)$ is replaced with $R^{keys}(\calX)$
denoting the light part of relation $R$ partitioned on $keys$.
It follows from Lemma~\ref{lem:fact_view_tree_light_case}
that the views in the view tree $ltree$ can be materialized in $\bigO{N^m}$
time. Thus, all views of the views trees in the set
$htrees \cup \{ltree\}$ can be materialized in 
$\bigO{N^m}$ time. This completes the inductive
step for \textit{Case~3}. 
\end{proof}


Using Lemma~\ref{lem:preproc_rho}, we prove Proposition \ref{prop:preproc_time}.
%
Without loss of generality, assume that  $\omega$ consists of a single connected component. 
Otherwise, we apply the same reasoning for each 
connected component. 
We also assume that
$Q$ contains at least one atom with non-empty schema. 
Otherwise, $\tau(\omega,\emptyset)$ returns 
a single atom with empty schema, which can obviously be materialized in constant time. 

By Lemma~\ref{lem:preproc_rho}, 
the view trees generated by 
$\tau(\omega,\calF)$ can be materialized in time
$\bigO{N^{\max\{1, 1+ (\xi(\omega, X,\calF)-1) \eps\}}}$,
where $X$ is the root variable of $\omega$.
It remains to show:
\begin{align}
\label{eq:bridge} 
\max\{1, 1+ (\xi(\omega, X,\calF)-1) \eps\} \leq 1+ (\fw-1) \eps. 
\end{align} 
First, assume  that 
$\xi(\omega, X,\calF) =0$.
This means that $\max\{1, 1+ (\xi(\omega, X,\calF)-1) \eps\} = 1$.
Since $Q$ contains at least one atom 
with non-empty schema, we have $\fw \geq 1$.
Thus, Inequality~\eqref{eq:bridge} 
holds. 
Now, let  $\xi(\omega, X,\calF) = \ell \geq 1$. We show that 
$\fw \geq \ell$.
It follows from $\xi(\omega, X,\calF) = \ell$ that $\omega$
contains 
a bound variable $Y$ such that 
$\rho^{\ast}_{Q}(B) = \ell$, where
$B = \vars(\omega_Y) \cap \calF$.
The inner nodes of each root-to-leaf path of a canonical variable order 
are the variables of an atom. 
Hence, for each variable $Z \in B$, there 
must be an atom in $Q$ that contains both $Y$ and $Z$.
This means that $Y$ and $Z$ depend on each other.
Let $\omega' = (T,\dep_{\omega'})$ be an arbitrary free-top variable order
for $Q$. Since all variables in $B$ depend on $Y$,
each of them must be on a root-to-leaf path with $Y$. Since 
$Y$ is bound and the variables in $B$ are free, the set 
$B$ must be included in $\anc(Y)$.
Hence,  $B \subseteq \dep_{\omega'}(Y)$. 
This means $\rho^{\ast}_{Q}(\{Y\} \cup \dep_{\omega'}(Y)) \geq \ell$,
which implies $\fw(\omega') \geq \ell$.
It follows $\fw \geq \ell$. 

\section{Proofs of the Results in Section~\ref{sec:enumeration}}
\label{appendix:enumeration}

\begin{citedprop}[\ref{prop:enumeration}]
The tuples in the result of a hierarchical query $Q(\calF)$ over a database of size $N$ can be enumerated with $\bigO{N^{1-\eps}}$ delay using the view trees constructed by $\tau(\omega,\calF)$ for a canonical variable order $\omega$ for $Q$.
\end{citedprop}

Following Proposition~\ref{prop:equivalence}, the union of queries defined by the set of view trees constructed by $\tau(\omega,\calF)$ is equivalent $Q(\calF)$. We enumerate the tuples over $\calF$ from this set of view trees using the $\mathit{next}$ calls of these trees in the set.

We first discuss the case of one view tree. In case there are no indicator views, then the view tree consisting of a hierarchy of views admits constant delay~\cite{OlteanuZ15}. In the static case, this holds for free-connex hierarchical queries; in the dynamic case, this holds for $\dfw_0$-hierarchical queries (Section~\ref{sec:preprocessing-factorization}). 

The view subtrees constructed over the light parts of input relations only do not bring additional difficulty. 
By construction (Section~\ref{sec:preprocessing}), the root view $V$ of such a subtree $T$ contains all the free variables that are present in $T$. 
In this case, the $\mathit{open}$ and $\mathit{next}$ calls stop at $V$ and do not explore the children of $V$.
This means that for enumeration purposes, we can discard the descendants of $V$.

By grounding the heavy indicators in $T$, we obtain instances of $T$ that may represent overlapping relations. We next analyze the enumeration delay in the presence of heavy indicators as a function of the view tree instances of a view tree created for $Q$.

Consider one heavy indicator. Since its size is $\bigO{N^{1-\eps}}$, it may lead to that many view tree instances. From each instance, we can enumerate with constant delay, and we can also look up a tuple with schema $\calS$ in constant time. Given there are  $\bigO{N^{1-\eps}}$ such tuples, we can enumerate from $T$ with $\bigO{N^{1-\eps}}$ delay.

Consider $p$ heavy indicators $\exists H_1(\mathcal{X}_1),\ldots,$ $\exists H_p(\mathcal{X}_p)$ whose parents $V_1(\mathcal{X}_1),\ldots,V_p(\mathcal{X}_p)$ are along the same path in the view tree. Let us assume $V_i$ is an ancestor of $V_j$ for $i<j$. By construction, there is a total strict inclusion order on their sets of variables, with the indicator above having less variables than at a lower depth: $\mathcal{X}_1\subset\cdots\subset\mathcal{X}_p$. Each indicator draws its tuples from the input relations whose schemas include that of the indicator. There is also an inclusion between the parent views: $V_i \subseteq \pi_{\mathcal{X}_i} V_{i+1}, \forall i\in[p-1]$. This holds since $V_i$ is defined by the join of the leaves underneath, so the view $V_j$ that is a descendant of $V_i$ is used to define $V_i$ in joins with other views or relations.
The size of $V_i$ is at most that of $\exists H_i$ since they both have the same schema and the former is defined by the join of the latter with other views. Since the size of $\exists H_i$ is $\bigO{N^{1-\eps}}$, it follows that the size of $V_i$ is also $\bigO{N^{1-\eps}}$. 
When grounding $\exists H_i$, we create an instance for each tuple $t$ that is in both $\exists H_i$ and $V_i$: If $t$ were not in $V_i$, then there would be at least one sibling of $\exists H_i$ that does not have it. When opening the descendants of $V_i$ before enumeration, only these tuples in $V_i$ that also occur in $\exists H_i$ and in all its siblings can be extended at the descendant views, including all views $V_j$ for $j>i$.
The overall number of groundings for the $h$ heavy indicators is therefore $\bigO{N^{1-\eps}}$. Let $n_i$ be the number of instances of $\exists H_i$. Then, the delay for enumerating from the union of $\exists H_i$ instances is $\sum_{i\leq j\leq p} n_j$ using the \textsc{Union} algorithm, which also accounts for the delay incurred for enumeration from unions at instances of all $\exists H_j$ that are descendants of $\exists H_i$. The overall delay is that for the union of instances for $\exists H_1$: $\sum_{1\leq j\leq p} n_j =  \bigO{p\times N^{1-\eps}} = \bigO{N^{1-\eps}}$.

Consider again the $p$ heavy indicators, but this time their parents $V_1,\ldots,V_p$ are {\it not} all along the same path in the view tree. Each path is treated as in the previous case. 
We distinguish two cases. In the first case, there is no parent $V_i$ that is an ancestor of several other parents in our list. Let $W$ be a common ancestor of several parents. Then, the enumeration algorithm uses each tuple of $W$ (possibly extended by descendant views) as context for the instances of these parents. A next tuple is produced {\em in sequence} at each of these parents over their corresponding schemas. These tuples are then composed into a larger tuple over a larger schema at their common ancestor using the \textsc{Product} algorithm.
The number of branches is bounded by the number of atoms in the query, which means that the overall delay remains $\bigO{N^{1-\eps}}$. In the second case, a parent $V_i$ is a common ancestor of  several other parents in our list. We reason similarly to the one-path case and obtain that the overall delay is $\bigO{p\times N^{1-\eps}}=\bigO{N^{1-\eps}}$.

So far we discussed the case of enumerating from one view tree. In case of a set of view trees we use the \textsc{Union} algorithm to enumerate the distinct tuples.
In case the query has several connected components, i.e., it is a Cartesian product of hierarchical queries, we use the \textsc{Product} algorithm.

\section{Proofs of the Results in Section~\ref{sec:updates}}
\label{appendix:updates}

\subsection{Proof of Proposition~\ref{prop:update_time}}\label{appendix:update_time}

\begin{citedprop}[\ref{prop:update_time}]
Given a hierarchical query $Q(\calF)$ with dynamic width $\dfw$, a canonical variable order $\omega$ for $Q$,
a database of size $N$, and $\eps \in [0,1]$, 
maintaining the views in the set of view trees $\tau(\omega,\calF)$ under a single-tuple update to any input relation takes $\bigO{N^{\dfw\eps}}$ time.
\end{citedprop}

We first give the maintenance time for the views constructed by \textsc{BuildVT} given a $\dfw_0$-hierarchical query (Lemma~\ref{lem:fact_view_tree_update}).
We then show the maintenance time for the views constructed by $\tau$ given a hierarchical query (Lemma~\ref{lem:fact_view_forest_update}). The maintenance time uses a new measure, which we relate to dynamic width (Lemma~\ref{lem:kappa_dynamic}).
We finally show the running times of \textsc{UpdateIndTree} and \textsc{UpdateTrees}.

\begin{lem}
\label{lem:fact_view_tree_update}
Given a $\dfw_0$-hierarchical query $Q(\calF)$, a canonical variable order $\omega$ for $Q$, and 
a database of size $N$, the views constructed by $\textsc{BuildVT}(\cdot, \omega, \calF)$ from Figure~\ref{fig:factorized_view_tree_algo} in the dynamic mode can be 
maintained under a single-tuple update to any input relation in $\bigO{1}$ time.
\end{lem}
\begin{proof}
At each node $X$ of a canonical variable order $\omega$ for a $\dfw_0$-hierarchical query, the set $\calF_X$ of free variables is either $\anc(X) \cup \{X\}$ if $X$ is free, or 
$\anc(X)$ if $X$ is bound because the set $\calF \cap \vars(\omega_X)$ of free variables in $\omega_X$ is empty for $\dfw_0$-hierarchical queries. 
The functions \textsc{AuxView} and \textsc{NewVT} maintain 
the following invariant for $\dfw_0$-hierarchical queries  in the dynamic mode: 
If $X$ has a sibling node in $\omega$, then the view created at node $X$ has 
$\anc(X)$ as free variables. If $X$ is bound, then already $\calF = \anc(X)$; otherwise, \textsc{AuxView} constructs an extra view with $\anc(X)$ as free variables.

Now consider an update $\delta{R}$ to a relation $R$. Due to the hierarchical property of the input query, the update $\delta{R}$ fixes the values of all variables on the path from the leaf $R$ to the root to constants. While propagating an update through the view tree, the delta at each node $X$ requires joining with the views constructed for the siblings of $X$. Each of the sibling views has $\anc(X)$ as free variables, as discussed above. Thus, computing the delta at each node makes only constant-time lookups in the sibling views. 
Overall, propagating the update through the view tree constructed for a $\dfw_0$-hierarchical query using $\textsc{BuildVT}$ takes constant time.
\end{proof}

Consider now a canonical variable order $\omega$ for a hierarchical query and a set $\calF$ of free variables.
Given a node $X$ in $\omega$, let $Q_X$ denote the join of $\atoms(\omega_X)$.
We define $\kappa(\omega,\calF)$ as:
$$\max_{X \in \vars(\omega)-\calF}\ \max_{R(\calY) \in \atoms(\omega_X)}
\{\rho_{Q_X}^{\ast}((\vars(\omega_X)\cap\calF)-\calY)\},$$
The measure $\kappa(\omega,\calF)$ is the maximal fractional edge cover number of $Q_X$ over the
free variables occurring in the subtree $\omega_X$ of $\omega$ rooted at a bound variable $X$, when the variables of one atom $R(\calY)$ in $\omega_X$ are excluded. 

\begin{lem}
\label{lem:fact_view_forest_update}
Given a hierarchical query $Q(\mathcal{F})$, a canonical variable order $\omega$ for $Q$,
a database of size $N$, and $\eps \in [0,1]$, the views constructed by $\tau(\omega, \mathcal{F})$ from Figure~\ref{fig:view_forest_main_algo} in the dynamic mode can be maintained under a single-tuple update in $\bigO{N^{\kappa(\omega,\calF)\eps}}$ time.
\end{lem}
\begin{proof}
If $Q$ is $\dfw_0$-hierarchical, the function $\tau$ returns a view tree for $Q$ that admits $\bigO{1}$ update time, per Lemma~\ref{lem:fact_view_tree_update}.

Consider now a view tree created by $\tau$ for a non-$\dfw_0$-hierarchical query.
Let us restrict this view tree such that the views created in the light case are treated as leaf views. This restricted view tree encodes the result of a $\dfw_0$-hierarchical query! As the procedure $\tau$ traverses the variable order in a top-down manner, every bound variable $X$ with a free variable below is replaced by a set of view trees where $X$ is free (heavy case) and by a view tree whose root view aggregates away $X$ and includes only free variables (light case). 
Thus, single-tuple updates to the leaves of this restricted view tree take constant time. 
That is, updates to the relations that are not part of the views materialized in the light case are constant. 

However, updates to the relations that are part of the views materialized in the light case might not be constant. 
The view tree $ltree$ constructed by $\textsc{BuildVT}$ at a bound variable $X$ is defined over the light parts of relations partitioned on $keys = \anc(X)\cup\{X\}$ (Line 16 in Figure~\ref{fig:view_forest_main_algo}).
Each view $V_Z$ in $ltree$ constructed at a variable $Z$ includes all the free variables in $\omega_Z$.
A single-tuple update $\delta{R}$ to any relation $R$ in $ltree$ fixes the values of the variables $keys$, thus reducing the size of other relations in $ltree$ to $\bigO{N^\eps}$.
The maintenance cost for $V_Z$ under the update $\delta{R}$ with schema $\calY$ is $\bigO{N^{m_Z\eps}}$, where $m_Z=\rho^*_{Q_Z}{((\vars(\omega_Z)\cap \calF) - \calY)}$. 
The maintenance cost for $ltree$ is dominated by the maintenance cost for its root $V_X$.

The change computed at $V_X$ for the single-tuple update consist of 
$\bigO{N^{m_X\eps}}$ tuples and needs to be propagated further up in the tree. 
Because there are no further light cases on the path from $X$ to the root, the propagation cost is constant per tuple. 
The overall time needed to maintain $V_X$ and propagate the change at $V_X$ up to the root is $\bigO{N^{m'_X\eps}}$, where 
$m'_X =\max_{R(\calY) \in \atoms(\omega_X)}\{\rho_{Q_X}^{\ast}((\vars(\omega_X)\cap\calF)-\calY)\}$.
In the worst case, the root variable of $\omega$ is bound;
then, maintaining the root view and its descendants takes $\bigO{N^{\kappa(\omega,\calF)\eps}}$ time. 

The views constructed by $\tau$ in the light cases thus determine the overall maintenance $\bigO{N^{\kappa(\omega,\calF)\eps}}$ time.
\end{proof}

We next relate the measure $\kappa(\omega,\calF)$ to dynamic width.

\begin{lem}
\label{lem:kappa_dynamic}
Given a canonical variable order $\omega$ for a hierarchical query $Q(\mathcal{F})$ with dynamic width $\dfw$, it holds that $\kappa(\omega,\calF) \leq \dfw$.
\end{lem}

\begin{proof}
Given any variable order $\omega'$ for $Q$ and a variable 
$X$ in $\omega'$, we denote by $Q_X^{\omega'}$ the query 
that joins the atoms in $\atoms(\omega_X')$.
 To prove $\kappa(\omega, \calF) \leq \dfw$, we need to show that
\begin{align}
\label{eq:1}
\kappa(\omega, \calF) \leq \dfw(\omega^f)
\end{align}
for any free-top variable order  $\omega^f$ for $Q$.
It follows from the definition of $\kappa(\omega,\calF)$
that  $\omega$ has a bound variable 
$X$ and an atom $R(\calY) \in \atoms(Q_X^{\omega})$ such that 
\begin{align*}
 \kappa(\omega, \calF) = \rho_{Q_X^{\omega}}^*(\calB) 
\end{align*}
where $\calB = (\vars(\omega_X) \cap \calF) - \calY$.
Since $\omega$ is canonical, it holds: 
\begin{itemize}[leftmargin=*]
\item[] $(*)$ Each atom in $Q$ containing a variable 
from $\calB$ must contain $X$.
\end{itemize}
Let $\omega^f = (T,\dep_{\omega^f})$ be a free-top 
variable order for $Q$. 
Property $(*)$ implies that the variables 
in $\calB$ depend on $X$.
Since $X$ is bound and 
all variables in $\calB$ are free, 
the latter variables 
cannot be below $X$ in $\omega^f$. 
Hence, $\calB \subseteq \dep_{\omega^f}(X)$. 
Since $R(\calY)$ contains $X$,
it must be included in $\atoms(\omega^f_X)$.  
To prove Inequality \eqref{eq:1}, 
it thus suffices to show:
\begin{align}
\label{eq:toShow2}
\rho_Q^{\ast}(\calB) \geq \rho_{Q_X^{\omega}}^*(\calB).
\end{align}
 By Property $(*)$, each atom in $Q$ covering a variable from $\calB$
  contains $X$. Hence, all such atoms are contained in 
 $\atoms(Q_X^{\omega})$. 
This implies that  
any fractional edge cover 
$\boldsymbol{\lambda}'$ of $\calB$ using atoms in 
$Q$ can be turned into 
a fractional edge cover 
$\boldsymbol{\lambda}$ of $B$ using atoms in 
$Q_X^{\omega}$ such that 
$\sum_{\lambda \in \boldsymbol{\lambda}} \lambda
 \leq \sum_{\lambda' \in \boldsymbol{\lambda}'} \lambda'$.
This implies Inequality \eqref{eq:toShow2} and hence 
Inequality \eqref{eq:1}.
\end{proof}

\begin{lem}
\label{lem:update_ind_tree}
Given an indicator tree $T_{Ind}$ constructed by \textsc{IndicatorVTs} from Figure~\ref{fig:skew_aware_views} and a single-tuple update $\delta{R}$, \textsc{UpdateIndTree} from Figure~\ref{fig:update_indicator_tree} runs in $\bigO{1}$ time.
\end{lem}
\begin{proof}
The tree $T_{Ind}$ encodes the result of a $\dfw_0$-hierarchical query and admits constant-time updates per Lemma~\ref{lem:fact_view_tree_update}. 
The remaining operations in \textsc{UpdateIndTree} also take constant time. 
\end{proof}

We next analyze the procedure \textsc{UpdateTrees} from Figure~\ref{fig:update_view_trees} under a single-tuple update.
Applying the update to each view tree from $\mathcal{T}$ (Line 1) takes $\bigO{N^{\dfw\eps}}$ time, per Lemmas~\ref{lem:fact_view_forest_update} and \ref{lem:kappa_dynamic}.
We then apply the update to each triple $(T_{All}, T_L, T_H)$ of indicator view trees. The tree $T_{All}$ is a view tree of a $\dfw_0$-hierarchical query, thus updating it takes constant time (Line 6). The tree $T_L$ is updated using $\textsc{UpdateIndTree}$ in constant time (Line 12), per Lemma~\ref{lem:update_ind_tree}. Both of these changes may trigger a change in $\exists T_H$, and 
propagating $\delta({\exists H})$ through each view tree from $\mathcal{T}$ (Lines 9 and 14) takes constant time since this change does not affect any view materialized in the light case. 
Updating each light part of relation $R$ and the affected view trees (Line 11) takes $\bigO{N^{\dfw\eps}}$ time, per Lemmas~\ref{lem:fact_view_forest_update} and \ref{lem:kappa_dynamic}. 

Overall, the procedure \textsc{UpdateTrees} maintains the views constructed by $\tau$ under a single-tuple update in $\bigO{N^{\dfw\eps}}$ time.

\subsection{Proof of Proposition~\ref{prop:major_time}}

\begin{citedprop}[\ref{prop:major_time}]
Given a hierarchical query $Q(\calF)$ with static width $\fw$, a canonical variable order 
$\omega$ for $Q$,
a database of size $N$, and $\eps \in [0,1]$, 
major rebalancing of the views in the set of view trees $\tau(\omega,\calF)$ takes $\bigO{N^{1+(\fw-1)\eps}}$ time.
\end{citedprop}

Consider the major rebalancing procedure from Figure~\ref{fig:major_rebalancing}. 
The light relation parts can be computed in $\bigO{N}$ time. 
Proposition~\ref{prop:preproc_time} implies that the affected views can be 
recomputed in time $\bigO{N^{1+(\fw-1)\eps}}$.

\subsection{Proof of Proposition~\ref{prop:minor_time}}

\begin{citedprop}[\ref{prop:minor_time}]
Given a hierarchical query $Q(\calF)$ with dynamic width $\dfw$, a canonical variable order 
$\omega$ for $Q$,
a database of size $N$, and $\eps \in [0,1]$, 
minor rebalancing of the views in the set of view trees $\tau(\omega,\calF)$ takes $\bigO{N^{(\dfw+1)\eps}}$ time.
\end{citedprop}

Figure~\ref{fig:minor_rebalancing} shows the procedure for minor rebalancing of the tuples with the partitioning value $key$ in the light part $R^{\calS}$ of relation $R$.  
Minor rebalancing either inserts fewer than $\frac{1}{2}M^\eps$ tuples into $R^{\calS}$ (heavy to light) or deletes at most $\frac{3}{2}M^\eps$ tuples from $R^{\calS}$ (light to heavy). 
Each action updates the indicator trees $T_L$ and $T_H$ in constant time (lines 5 and 6), per Lemma~\ref{lem:update_ind_tree}. 
Propagating the update to the light part of relation $R$ through each view tree from $\mathcal{T}$ (line 4) takes $\bigO{N^{\dfw\eps}}$ time, per Lemmas~\ref{lem:fact_view_forest_update} and \ref{lem:kappa_dynamic}. 
Propagating the change $\delta(\exists{H})$ through each view tree from $\mathcal{T}$ takes constant time (line 7), as 
discussed in the proof of Proposition~\ref{prop:update_time}.
Since there are $\bigO{M^\eps}$ such operations and the size invariant $\floor{\frac{1}{4}M} \leq N < M$ holds, the total time is $\bigO{N^{(\dfw+1)\eps}}$.
\end{document}